\documentclass[12pt,a4paper]{article}
\usepackage[margin=1in]{geometry}
\usepackage[T1]{fontenc}
\usepackage[utf8]{inputenc}
\usepackage{mathptmx}
\usepackage{courier}
\usepackage{amsmath,amssymb,amsthm,mathtools}
\usepackage{setspace,graphicx,booktabs,longtable,array,tabularx,etoolbox,needspace,float,placeins}
\usepackage[authoryear,round]{natbib}
\usepackage{caption}
\usepackage{flafter}
\usepackage{titlesec}
\usepackage[hidelinks,unicode]{hyperref}
\let\small\normalsize

\AtBeginEnvironment{thebibliography}{\singlespacing}
\AtBeginEnvironment{footnote}{\normalsize\singlespacing}
\DeclareCaptionFont{submissionhalf}{\setstretch{1}}
\renewcommand{\thesection}{\arabic{section}}
\renewcommand{\thesubsection}{\thesection.\arabic{subsection}}

\titleformat{\section}{\normalfont\normalsize\bfseries}{\thesection.}{0.6em}{}
\titleformat{\subsection}{\normalfont\normalsize\itshape}{\thesubsection.}{0.6em}{}
\titlespacing*{\section}{0pt}{9pt}{4pt}
\titlespacing*{\subsection}{0pt}{6pt}{3pt}
\renewcommand{\thetable}{\arabic{table}}
\renewcommand{\thefigure}{\arabic{figure}}
\newtheorem{theorem}{Theorem}
\newtheorem{proposition}[theorem]{Proposition}
\newtheorem{lemma}[theorem]{Lemma}
\newtheorem{corollary}[theorem]{Corollary}
\newtheorem*{corollary*}{Corollary}
\theoremstyle{definition}
\newtheorem{assumption}{Assumption}

\newcommand{\R}{\mathbb R}
\newcommand{\E}{\mathbb E}

\newcommand{\Var}{\operatorname{Var}}
\newcommand{\Cov}{\operatorname{Cov}}

\newcommand{\ind}{\mathbf 1}

\newcommand{\submissionTableStyle}{\normalsize\setstretch{1}}
\AtBeginEnvironment{longtable}{\submissionTableStyle}
\AtBeginEnvironment{tabular}{\submissionTableStyle\typeout{TABLE: baseline=\the\baselineskip}}

\newcommand{\note}[1]{\par\medskip\normalsize\setstretch{1}\noindent\raggedright Note: #1\par}

\newcommand{\MainRef}[1]{\textup{\ref{#1}}}
\newcommand{\SuppRef}[1]{\textup{\ref{#1}}}
\newcommand{\MainEqRef}[1]{\textup{(\MainRef{#1})}}
\newcommand{\SuppEqRef}[1]{\textup{(\SuppRef{#1})}}
\AtBeginDocument{\typeout{SUBMISSION LAYOUT; font=12; baseline=\the\baselineskip; textwidth=\the\textwidth; textheight=\the\textheight}}

\newcommand{\DispGaussianRatio}{10.57}
\newcommand{\DispGaussianMCSE}{0.007}
\newcommand{\DispGaussianCoverage}{1.17}

\apptocmd{\normalsize}{%
 \setlength{\abovedisplayskip}{6pt plus 2pt minus 1pt}%
 \setlength{\belowdisplayskip}{6pt plus 2pt minus 1pt}%
 \setlength{\abovedisplayshortskip}{3pt plus 1pt}%
 \setlength{\belowdisplayshortskip}{3pt plus 1pt}%
}{}{}
\AtBeginDocument{\normalsize\setstretch{1}}
\hypersetup{pdftitle={Valid Tests and Confidence Intervals in Difference-in-Difference},pdfauthor={Takahiro Hoshino and Ryoya Nakano},pdfkeywords={difference-in-differences, control group selection, selective inference, staggered adoption}}
\hypersetup{hypertexnames=false}
\begin{document}
\pdfbookmark[0]{Manuscript}{manuscript}
\begin{center}
{\normalsize\bfseries Valid Tests and Confidence Intervals in Difference-in-Difference\par}
\medskip
Takahiro Hoshino (corresponding author)\par
Faculty of Economics, Keio University, Tokyo 108-8345, Japan;\par
RIKEN Center for Advanced Intelligence Project, Tokyo, Japan\par
\href{mailto:hoshino@econ.keio.ac.jp}{hoshino@econ.keio.ac.jp}\par
\smallskip
Ryoya Nakano\par
Graduate School of Economics, Keio University, Tokyo 108-8345, Japan\par
\href{mailto:r.nakano@keio.jp}{r.nakano@keio.jp}
\end{center}
\par\medskip
\noindent\textbf{Abstract.} Choosing the control group whose pre-treatment outcomes most closely track those of the treated group can lead a conventional difference-in-differences test to reject too often, even when every candidate satisfies parallel trends. We keep the selected comparison and construct tests and confidence intervals for the average treatment effect on the treated by replaying the ranking and any pre-trend screen. We estimate the covariance of the outcome means from the same sample and establish uniform asymptotic coverage conditional on each report whose probability stays bounded away from zero, including when candidates fit equally well in the population. In one simulation with five controls and no effect, the conventional 5 percent test rejects in 20.87 percent of reports passing the winner's pre-trend screen. Our test rejects in 5.06 percent of those reports. In a Buenos Aires policing reanalysis, the conventional interval indicates fewer car thefts. The adjusted interval includes zero.
\par\smallskip
\noindent Journal of Economic Literature classification: C12, C21, C23.\par
\par\medskip
\newcommand{\FixedThree}{5.22}
\newcommand{\UsualThree}{13.75}
\newcommand{\SelectiveThree}{5.03}
\newcommand{\ScreenOnlyThree}{13.60}
\newcommand{\FixedFive}{5.01}
\newcommand{\UsualFive}{20.87}
\newcommand{\SelectiveFive}{5.06}
\newcommand{\ScreenOnlyFive}{20.70}
\newcommand{\FixedTen}{5.14}
\newcommand{\UsualTen}{35.57}
\newcommand{\SelectiveTen}{4.96}
\newcommand{\ScreenOnlyTen}{35.39}
\section{Introduction}\label{sec:intro}

Did police deployment reduce car theft on the blocks that officers protected? Using the Buenos Aires data of \citet{DS2004}, we compare 37 protected blocks with three overlapping pools of unprotected blocks defined by distance. Each comparison yields about 0.08 fewer thefts per block per recorded month, and each conventional 95 percent confidence interval excludes zero. Suppose we choose the pool whose thefts most closely tracked those on the protected blocks before deployment. The pool of all 839 unprotected blocks wins. Its estimate is still 0.078 fewer thefts, but the 95 percent interval accounting for how we chose the pool includes zero.

Why should the other two pools matter once we have chosen a control? The choice tells us something about all three pre-treatment paths. If the differences that decide which pool wins also covary with its estimated effect, the estimate has a different distribution in samples where that pool wins. For difference-in-differences (DiD), we construct tests and confidence intervals for the average treatment effect on the treated (ATT) that account for this choice. In Buenos Aires, we ask which other values of the estimated effect would have led us to report the same pool. We estimate how each candidate's outcome means covary with that effect estimate, vary the means accordingly, and repeat the original comparisons. To test a proposed effect, we compare the estimate with the values that reproduce the report, using their conditional probabilities under that effect.

A pre-trend test for a control fixed in advance addresses a familiar reporting decision. Under their conditions, \citet[Theorem~1 and Corollary~1]{deCDH2026} show that, conditional on passing the prespecified tests, the asymptotic probability of rejecting a true null with the usual two-sided effect test is at most its nominal level. Our framework includes this reporting rule when only one control is eligible. With several eligible controls, even a candidate whose own pre-treatment gaps are independent of its effect estimate must also beat its rivals. In the two-control model of Proposition~\ref{amp:two-main}, if the estimation error is more strongly related in absolute value to the rival's pre-treatment discrepancy than to its own, the usual test rejects too often conditional on the choice. With the reverse ordering it rejects too seldom. With equally strong relationships it rejects at the nominal rate. Both controls satisfy parallel trends throughout.

How large can overrejection be? Suppose the test statistics for the fixed comparisons are jointly Gaussian and each is standard normal under the common null. If the researcher always reports one of five candidates, the sharp upper bound on the probability of a false rejection with the usual two-sided 5 percent test is $1-0.95^5$, or about 22.62 percent, averaged over all choices. Ranking by pre-treatment fit alone can approach this bound (Section~\ref{amp:geometry-main-section} and Online Appendix Corollary~\SuppRef{v91:sharp-size}).

Figure~\ref{fig:matched-screen} illustrates the additional decision to report a winner only if it passes a pre-trend test. All candidates identify the same zero ATT. In one five-candidate design, the usual 5 percent test rejects in \FixedFive\ percent of reports when the control is fixed in advance. Ranking alone gives 21.16 percent across all 100,000 draws. When we rank and then screen, 98,359 draws yield a report. On those same reports, the usual test rejects in \UsualFive\ percent, our adjustment for screening alone in \ScreenOnlyFive\ percent, and our adjustment for ranking and screening in \SelectiveFive\ percent. With five candidates whose estimates depend on nearly the same combination of pre-treatment shocks, the usual and fully adjusted rates are 5.33 and 5.06 percent. The candidate count is the same in both designs. What varies is how the competitors' pre-treatment outcomes covary with the effect estimate.

\begin{figure}[!htbp]
\centering
\includegraphics[width=.99\linewidth]{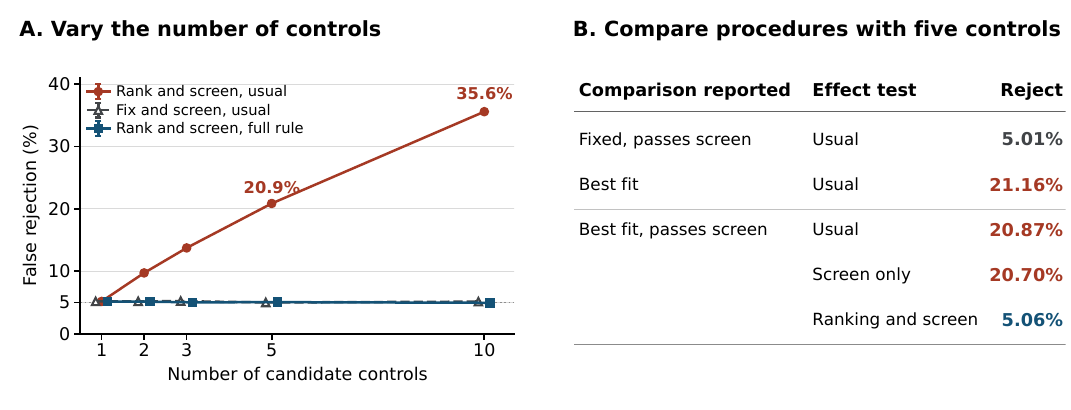}
\caption{Simulated rejection rates when the control is fixed or chosen by pre-treatment fit}
\label{fig:matched-screen}
\note{Each design has a zero effect, parallel trends for every candidate, and 100,000 draws. The dotted line marks the nominal 5 percent level. At five candidates, the fixed comparison passes its pre-trend screen 95,092 times. Ranking alone reports all draws. The last three tests in panel B use the same 98,359 winners that pass the screen, the same effect estimates, and the same estimated covariance. The screen-only test replays the winner's pre-test. The full rule also requires the winner to beat every rival. Panel A bars show one Monte Carlo standard error. See Online Appendix Table~\SuppRef{tab:matched-screen}.}
\end{figure}

Sampling variation can decide the winner when candidates fit equally well in the population. At perfect fit, sums of squared pre-treatment gaps can have nonnormal distributions even when the underlying outcome means are approximately normal. We approximate the joint distribution of the means and repeat the actual comparisons. For a fixed finite menu satisfying the conditions of Theorem~\ref{thm:plugin}, we can estimate the joint covariance from the selection sample and still obtain uniform asymptotic coverage conditional on each report whose probability stays bounded away from zero. The result includes perfect fit and population ties. When every sample yields a report, coverage averaged over all choices also converges to the nominal level.

Researchers record the eligible controls, fit measures, order of ranking and screening, and tie rule. They can compare levels or changes, screen before or after ranking, or search over several fit measures. When regions adopt a policy at different dates, researchers can specify in advance which sequences of cohort choices to pool, provided each sequence in a pool uses the same contrast for every reported effect. In our four-cohort experiment, pooling reduces the median interval length by about 21 percent among draws with finite intervals under both procedures. Across all draws, conditioning on the detailed sequence discards about a quarter of the mean Fisher information available under pooling (Section~\ref{sd:section}).

We specify candidate pools and rankings in six reanalyses of published policy data. For Washington's minimum-wage increase, ranking 39 jurisdictions by pre-policy employment in the hourly wage band from \$5 to below \$14 selects Utah. Both 95 percent intervals include zero. For the same estimated change in the share of adults employed in that band, the conventional interval rules out a one-percentage-point decline and the adjusted interval admits a seven-point decline. In the community health center study, the adjusted interval still supports a decline in mortality among adults aged 50 or older (Section~\ref{sec:applications}).

Our earlier paper \citep{NH2026} establishes exact Gaussian inference with known covariance when controls are ranked by the sum of squared pre-treatment gaps around their mean. Here we estimate covariance from the same sample and replay the original ranking rule even as population fits approach a tie. \citet{AKM2024} study inference after maximizing asymptotically normal scores. \citet[Section~3.3 and Assumption~4]{AKM2021} study norm maximization followed by a pre-test under bounded standardized population means. Theorem~\ref{thm:plugin} allows population fit to vary with sample size without that bound. We build on selective inference for recorded choices \citep{Loftus2015,Lee2016,Tian2017,TRTW2018}.

\citet{FPP2020} study specification search over matching variables in synthetic control and its effects on estimates and placebo tests. \citet{AAHIW2021} estimate weights in synthetic DiD, and \citet{CWZ2026} develop inference for synthetic controls. Drawing on \citet{Berk2013}, \citet{HN2026simultaneous} construct simultaneous intervals for a declared family of comparisons. We seek coverage among samples reporting the same comparison. Pre-trend tests and sensitivity analyses examine the role of the parallel-trends assumption in interpreting the reported effect \citep{Roth2022,RR2023,MikhaeilHarshaw2026}.

\section{Control selection in six policy data sets}
\label{sec:applications}\label{appv61:cross-applications}
Did police deter car thefts in Buenos Aires? After raising its minimum wage, how large a decline can Washington rule out in the share of adults working in the specified wage band? We revisit these questions using data from \citet{DS2004} and \citet{CDLZ2019}. In each case we specify several plausible controls, select one by pre-treatment fit, and compare two 95 percent intervals for exactly the same estimate. Four further reanalyses cover community health centers \citep{BGB2015}, the Tennessee Valley Authority (TVA) \citep{KM2014}, compulsory licensing \citep{MV2012}, and TennCare disenrollment \citep{GGN2014}. Online Appendix Table~\SuppRef{tab:application-provenance} distinguishes the published comparison designs from the candidate pools and ranking rules we use.

Let $P$ be the set of pre-treatment dates. For candidate $j$ and date $t\in P$, let $d_{jt}$ be the
treated-minus-candidate outcome gap and $\bar d_j=|P|^{-1}\sum_{t\in P}d_{jt}$
its mean over the $m_P=|P|$ dates. With $\lambda\ge0$ weighting the mean-gap term, we select the candidate minimizing the loss
\begin{equation}\label{appv62:loss}
 Q_j^{(\lambda)}=\sum_{t\in P}(d_{jt}-\bar d_j)^2
                         +\lambda m_P\bar d_j^2.
\end{equation}
Centered fit ($\lambda=0$) compares paths while allowing a persistent level difference. Setting $\lambda=1$ also balances levels and gives
$Q_j^{(1)}=\sum_{t\in P}d_{jt}^2$. For either rule, we subtract the selected
pre-treatment mean gap from the post-treatment mean gap to estimate the DiD effect.
We report both rules. We recorded the full weight grid $\{0,.25,.5,1\}$ after the centered-fit analysis and before the level-weight calculations. Online Appendix Section~\SuppRef{appv3:nested} records the menu chronology. Section~\ref{levelv62:main} explains selection by levels and paths.

Within each row of Table~\ref{tab:cross-applications} we hold the sample, control, point estimate and covariance fixed. The conventional interval treats the control as given. The selective interval accounts for how the control won the recorded ranking, including the tie rule. We estimate covariance using complete street-block records for police, linked households for Washington, state clusters for the county studies, patent parent classes for licensing, and historical regional gaps for TennCare. Online Appendix Sections~\SuppRef{app:empirical-details} and~\SuppRef{appv62:covariance} give the sampling assumptions and covariance constructions.

\begin{table}[tbp]\centering\setlength{\tabcolsep}{3pt}
\caption{Six policy data sets: inference after selecting controls}\label{tab:cross-applications}\label{tab:applications}
\begin{tabular}{crllr}\toprule
$\lambda$ & Estimate & Conventional 95\% & Selective 95\% & \shortstack{Interval length\\ratio}\\\midrule
\multicolumn{5}{l}{\textit{Police: thefts per block per recorded month}}\\
$0$ & $-0.07753$ & $[-0.12416,-0.03090]$ & $[-0.12267,0.04266]$ & $1.773$\\
$1$ & $-0.07753$ & $[-0.12416,-0.03090]$ & $[-0.12252,0.04708]$ & $1.819$\\
\addlinespace
\multicolumn{5}{l}{\textit{Washington: wage-band employment, pp, linked households}}\\
$0$ & $0.903$ & $[-0.807,2.613]$ & $[-7.234,2.424]$ & $2.824$\\
$1$ & $-1.229$ & $[-2.535,0.077]$ & $[-4.808,0.031]$ & $1.853$\\
\addlinespace
\multicolumn{5}{l}{\textit{Health centers: age-adjusted deaths/100,000 aged 50+, state clusters}}\\
$0$ & $-95.686$ & $[-129.694,-61.678]$ & $[-129.654,-50.036]$ & $1.171$\\
$1$ & $-95.686$ & $[-129.694,-61.678]$ & $[-129.694,-61.677]$ & $1.000$\\
\addlinespace
\multicolumn{5}{l}{\textit{TVA: log manufacturing employment, state clusters}}\\
$0$ & $0.403$ & $[0.235,0.571]$ & $[0.108,0.570]$ & $1.371$\\
$1$ & $0.403$ & $[0.235,0.571]$ & $[0.235,0.571]$ & $1.000$\\
\addlinespace
\multicolumn{5}{l}{\textit{Licensing: patents/subclass-year, 19 parent-class clusters}}\\
$0$ & $0.032$ & $[-0.073,0.137]$ & $[-0.027,0.157]$ & $0.880$\\
$1$ & $0.032$ & $[-0.073,0.137]$ & $[-0.053,0.172]$ & $1.070$\\
\addlinespace
\multicolumn{5}{l}{\textit{TennCare: employment percentage points, historical-gap model}}\\
$0$ & $3.011$ & $[0.040,5.981]$ & $[0.040,5.981]$ & $1.000$\\
$1$ & $3.886$ & $[0.391,7.381]$ & $[8.057,698.213]$ & $98.723$\\
\addlinespace
\bottomrule\end{tabular}
\note{Centered squared fit uses $\lambda=0$; uncentered squared fit uses $\lambda=1$. Within each row both intervals use the same control, estimate, and joint covariance. Each length ratio divides the selective interval length by the conventional interval length. Washington measures employment at hourly wages $\$5$ to below $\$14$ divided by adult population, in percentage points (pp), holding Quarterly Census of Employment and Wages (QCEW) totals fixed. TVA denotes the Tennessee Valley Authority. Licensing uses the exploratory availability menu. Both TennCare rows report unrestricted intervals. For $\lambda=1$, intersecting the interval with the known effect range under causal ATT identification gives $[8.057,100]$ (length ratio 13.15, Section~\ref{appv62:ggn}). Designs, intermediate weights, and covariance alternatives appear in Online Appendix Section~\SuppRef{appv61:cross-details}.}
\end{table}

\subsection{Deterrence and wage-band employment}
\label{appv63:policy-ranges}
\paragraph*{Police deployment.}\label{appv3:police}
After police were posted outside Jewish institutions in Buenos Aires in July 1994, \citet{DS2004} studied car thefts on protected blocks. We hold their 37 protected blocks fixed and compare three pools: all 839 unprotected blocks, the 678 at least two blocks away, and the 452 more than two blocks away. The authors use the farthest pool as their principal reference. In our reanalysis we rank the pools by theft patterns in April, May, June, and July 1--17, break ties in the listed order, and compare the pre-treatment mean gap with its August--December average. Our covariance estimate retains repeated observations and pool overlap (Online Appendix equation~\SuppEqRef{app:overlap}). Online Appendix Table~\SuppRef{tab:police-full-grid-v15} compares spatial covariance specifications.

The estimated reductions range from .078 to .081 car thefts per block per recorded month. All three conventional 95 percent intervals exclude zero. Both fit rules select all unprotected blocks. For that comparison, accounting for the centered-fit ranking changes the interval from $[-.124,-.031]$ to $[-.123,.043]$. The adjusted interval includes both a reduction and no change in thefts. Its two-sided $p$-value for zero rises from $.0011$ to $.181$. The farthest pool would tie the winner just 0.76 standard errors above the observed estimate along the line used to reconstruct the ranking (Section~\ref{sec:pivot}).

Researchers concerned about displacement of theft can exclude the nearest blocks before ranking \citep[Section~5]{Abadie2021}. Of the two remaining pools, the one including blocks exactly two blocks away wins, with estimate $-.080$ and adjusted interval $[-.124,.071]$. Fixing the farthest pool in advance gives $-.081$ and interval $[-.128,-.033]$ (Online Appendix Table~\SuppRef{tab:police-eligibility-v97}). The replication companion's \path{analysis/police_worked_example.py} returns all three candidate estimates and standard errors and reproduces the selected interval. If reporting also requires passing a 5 percent pre-trend screen, the same winner passes ($p=.443$). Accounting for ranking and screening then gives $[-.139,.043]$.

\paragraph*{Washington minimum wage.}\label{appv63:minimum-wage}
The outcome is the share of adults employed at hourly wages from \$5 to below \$14. A worker in this band leaves it if the hourly wage rises to \$14 or more. \citet[Online Appendix~C]{CDLZ2019} pool 39 eligible comparison jurisdictions. We instead rank those jurisdictions by their 1996--1998 wage-band employment paths, omit 1999, and compare Washington with the winner over 2000--2004. We use the Current Population Survey (CPS) for wages in all jurisdictions and hold Quarterly Census of Employment and Wages (QCEW) employment totals fixed.

Centered fit selects Utah. With linked-household covariance, the estimated increase is $.903$ percentage points and both intervals include zero. The usual lower limit is $-.807$ points, about 4 percent of Washington's pre-treatment share of 19.592 percent. After selection, the lower limit is $-7.234$ points, about 37 percent of that share. A seven-point decline in wage-band employment is compatible with this comparison. Oklahoma would tie Utah only 0.40 standard errors below the observed estimate along the conditional line. The adjusted lower limit is $-7.211$ under strict household linkage and $-.620$ under an unlinked covariance, at the same control and estimate. Including levels selects New York (Online Appendix Table~\SuppRef{tab:cps-covariance-v15} and Table~\ref{tab:cross-applications}).

\subsection{When conventional findings are preserved}
\paragraph*{Community health centers.}\label{appv62:county}
We study the 24 counties first receiving grants in 1967, with fixed
1960 population weights for ages 50 and older. We use 1959--1966 to fit
age-adjusted mortality and 1967--1971 for the post-treatment average.
Five pools group 2,440 counties without a grant in 1965--1980 by urban share.
Both fit rules choose the pool with at least 50 percent urban population.
The estimate is $-95.686$ deaths per 100,000, or about 2.75 percent of
the treated baseline rate of 3,475.91. Under state clustering, the usual
95 percent interval is $[-129.694,-61.678]$. The adjusted upper limit is $-50.036$ under centered fit and $-61.677$ when levels enter the ranking. Both intervals retain evidence of a mortality decline. A decline of 60 deaths per 100,000 is compatible only with the centered-fit adjusted interval. With levels, the nearest selection boundary is 5.63 standard errors away (Online Appendix Section~\SuppRef{appv62:fibers}).

\paragraph*{Regional development.}\label{appv63:tva}
For TVA, we compare 87 treated counties with three pools of non-TVA
counties. The outcome is log manufacturing employment, comparing the
1900--1930 pre-treatment average with 1960. Both rules select southern
non-TVA counties, with estimate $.40273$. Centered selection lowers
the conventional lower limit from $.235$ to $.108$ log points. With levels, the same control wins everywhere on the conditional line, so the intervals coincide
(Online Appendix Proposition~\SuppRef{diagv63:certificate}). Both adjusted intervals support a gain in manufacturing employment.

\paragraph*{Compulsory licensing.}\label{appv62:licensing}
We compare US-inventor patents in 336 licensed subclasses with three
unlicensed pools distinguished by historical enemy-patent availability.
Holding these institutional classifications fixed, both rules select
the available-but-unlicensed pool: a comparison with access to the
relevant patent stock but no recorded license. The estimate is $.03182$
patents per subclass-year. Using 19 parent-class clusters and normal critical values, centered fit shortens the interval by 12\%, and level-inclusive fit gives a length ratio of 1.07. Both intervals include zero. Treating subclasses as independent instead gives respective length ratios of 5.31 and 1.48. We introduced the availability menu after inspecting results for a numerical-code menu. Each interval conditions on the control choice within its recorded menu; Online Appendix Section~\SuppRef{appv62:details} gives the chronology.

\subsection{TennCare and historical regional variation}\label{appv62:ggn}
We study employment among Tennessee adults aged 21--64 with no more than a bachelor's degree and no children in the household. We compare the 2000--2004 baseline with 2006--2007, omitting the transition year. The four candidate controls are the other 16 southern jurisdictions pooled and each of their three Census divisions. We retain their shared Tennessee outcomes by modeling historical gaps against the three disjoint divisions. The principal working covariance uses ten centered annual gaps from 1995--2004 and a stationary bootstrap
\citep{PolitisRomano1994} with mean block length three. Online Appendix
Section~\SuppRef{appv62:ggn-macro} gives twelve covariance models, their
interval ranges, and the assumptions for interpreting the short time series.

Centered selection chooses South Atlantic donors and gives almost
identical intervals. Including levels selects all southern donors:
the conventional interval $[.391,7.381]$ becomes
$[8.057,698.213]$. The unrestricted length ratio is 98.72. At the observed estimate of $3.886$ percentage points, near an upper selection boundary, the conditional cumulative probability for that trial mean is $.9925$. It exceeds $.975$, so both endpoints of the equal-tailed 95 percent interval exceed the estimate (Online Appendix Section~\SuppRef{appv62:fibers}). Employment is binary, so under causal ATT identification the difference in mean employment lies in $[-100,100]$ percentage points. Intersecting with this range gives $[8.057,100]$, a length ratio of 13.15 using the same covariance. With centered fit, the adjusted intervals exclude zero with the 1995--2004 covariance window and include zero with 1986--2004. The level-inclusive adjusted lower limit stays positive in all twelve covariance cases (Online Appendix Section~\SuppRef{appv62:ggn-macro}).

\paragraph*{Why some interval adjustments are large.}
Researchers can explain large interval adjustments by measuring how close
the estimate lies to the boundary of the complete selection set on the
conditional line in Section~\ref{sec:pivot}. Suppose this set is a single
interval with finite upper endpoint $b$. Let $\widehat\tau<b$ be the observed
contrast with standard error $\sigma$ and put
$\delta_{\partial}=(b-\widehat\tau)/\sigma$. Applying the truncated-normal
boundary limit of \citet{KL2021} gives, as $\delta_{\partial}\downarrow0$ with the
selection interval and variance fixed,
\begin{equation}\label{diagv63:boundary-main}
 \frac{\operatorname{length}(C^{\rm sel}_{.05})}
      {\operatorname{length}(C^{\rm conv}_{.05})}
 \sim\frac{\log39}{2\Phi^{-1}(.975)}\frac1{\delta_{\partial}}
 \approx\frac{.9346}{\delta_{\partial}}.
\end{equation}
Here $C^{\rm sel}_{.05}$ and $C^{\rm conv}_{.05}$ are the selective and
conventional 95\% intervals, $\Phi^{-1}$ is the standard normal quantile
function, and $\sim$ denotes asymptotic equivalence.
In the level-inclusive TennCare comparison, the boundary lies $.009475$ standard errors above the estimate. The leading length-ratio approximation is 98.64, and the computed ratio is 98.72. Near a selection boundary, conditioning on the selected control can
therefore leave a wide, asymmetric range of effects compatible with the data. This distance explains interval length under the stated covariance. Section~\ref{cmpv2:failure} characterizes the repeated-sampling error-variance ratio. Online Appendix Lemma~\SuppRef{diagv63:boundary-lemma} derives the expansion;
Section~\SuppRef{diagv63:boundary} distinguishes an outer boundary of the full
set from an interior boundary of a disconnected component.

\section{Why control selection changes inference}\label{sec:setup}

\subsection{The comparison, the data, and control choice}
We show how control selection changes the distribution of a DiD estimate and how to construct a test or interval for the reported comparison. Section~\ref{sec:asymptotics} establishes large-sample validity with estimated covariance.

Region $0$ is treated. Regions $1,\ldots,J$ form the prespecified control pool. Index $n$ denotes a sampling problem in the sequence, with sample size $n_j$ in region $j$. Each unit is observed over $T$ periods. For unit $i=1,\ldots,n_j$ in region $j$, let $X_{njit}$ be its observed outcome at date $t=1,\ldots,T$ and write $X_{nji}=(X_{nji1},\ldots,X_{njiT})'$. A prime denotes transpose, $I_T$ is the order-$T$ identity, and $\operatorname{blockdiag}_j$ stacks the regional blocks along the diagonal. With $\E$ denoting expectation and $\Var$ the covariance of a random vector, define
\begin{align}
 \bar X_{nj}&=n_j^{-1}\sum_iX_{nji},\quad
 Y_n=(\bar X_{n0}',\ldots,\bar X_{nJ}')',\quad
 \mu_n=\E Y_n,\quad\mu_{nj}=\E X_{nji},\label{eq:regional-means}\\
 L_n&=\operatorname{blockdiag}_j(n_j^{-\frac{1}{2}}I_T),\quad
 V_n=\operatorname{blockdiag}_j(\Gamma_{nj}),\quad
 \Gamma_{nj}=\Var(X_{nji}).\label{eq:regional-scale}
\end{align}
We assume independent regions and, within each region, independent and identically distributed (i.i.d.) outcome vectors, allowing unrestricted dependence across periods within an observational unit. The numbers of regions and periods are fixed, while regional sample sizes may differ. Cluster and spatial covariance specifications, including those in Section~\ref{sec:applications}, satisfy the general theory when they provide the joint normal approximation and consistent covariance estimation in Assumption~\ref{ass:array}.

Under regional independence, $Y_n$ has covariance $\Sigma_n=L_nV_nL_n'$. For $n_j\ge2$, we estimate it using the unbiased within-region sample covariances from the same observations used for selection and effect estimation:
\begin{equation}\label{eq:regional-cov-est}
 \widehat\Gamma_{nj}=\frac{\sum_i(X_{nji}-\bar X_{nj})(X_{nji}-\bar X_{nj})'}{n_j-1},
 \qquad \widehat\Sigma_n=L_n\widehat V_nL_n',\quad
 \widehat V_n=\operatorname{blockdiag}_j(\widehat\Gamma_{nj}).
\end{equation}
Using $Y_n$ for all scores and estimates preserves dependence from their shared treated-region outcomes.

Fix disjoint sets of pre-treatment periods $P$ and post-treatment periods $A$, $|P|\ge2$, $|A|\ge1$, with candidates untreated through $\max A$. Write $D_{njt}=\bar X_{n0,t}-\bar X_{nj,t}$, $\bar D_{nj,P}=|P|^{-1}\sum_{t\in P}D_{njt}$. For candidate $j$, $B_jY_n$ stacks the centered pre-treatment gaps $(D_{njt}-\bar D_{nj,P})_{t\in P}$. With $\|\cdot\|$ the Euclidean norm, we select index $S_n$ and define the effect estimate and its mean for each candidate $s$ by
\begin{align}
 Q_j(Y_n)&=\sum_{t\in P}(D_{njt}-\bar D_{nj,P})^2
       =\|B_jY_n\|^2=Y_n'M_jY_n,\quad M_j=B_j'B_j,
       \label{eq:selection}\\
 S_n&=\arg\min_j Q_j(Y_n),\qquad
 T_{n,s}=|A|^{-1}\sum_{t\in A}D_{nst}-\bar D_{ns,P}
          =c_s'Y_n,\quad \theta_{n,s}=c_s'\mu_n.\label{eq:target}
\end{align}
Centering permits a persistent level difference. In $c_s$, the treated region's post-treatment and pre-treatment entries are $\frac{1}{|A|}$ and $-\frac{1}{|P|}$; candidate $s$ has the opposite weights, and all other entries are zero. The outcome, periods, candidate pool, weights and tie-breaking rule are fixed. Away from ties, $s$ wins exactly on the event
\begin{equation}\label{eq:joint}
 \widetilde E_{n,s}=\bigcap_{j\ne s}\{Y_n'A_{j,s}Y_n\ge0\},
 \qquad A_{j,s}=M_j-M_s.
\end{equation}

The gaps used for ranking also enter the DiD estimate. We use its distribution conditional on the event producing the report, including the recorded priority at a tie.

Let $\mu^0_{nj,t}$ be the mean potential outcome of region $j$ at date $t$
without treatment and $\mu^1_{n0,t}$ the mean for the treated region under its
observed treatment regime. Under consistency, no spillovers, and no anticipation during $P$, the mean DiD contrast decomposes into the ATT $\tau_n$ and the change in the untreated gap $\Delta_{n,s}$:
\begin{align}
 \tau_n&=|A|^{-1}\sum_{t\in A}(\mu^1_{n0,t}-\mu^0_{n0,t}),
       \label{sd:single-att}\\
 \Delta_{n,s}&=|A|^{-1}\sum_{t\in A}(\mu^0_{n0,t}-\mu^0_{ns,t})
       -|P|^{-1}\sum_{t\in P}(\mu^0_{n0,t}-\mu^0_{ns,t}),
       \label{sd:single-pt}\\
 \theta_{n,s}&=\tau_n+\Delta_{n,s}.\label{sd:single-decomposition}
\end{align}
The ATT averages effects over post-treatment dates. Parallel trends for $s$ gives $\Delta_{n,s}=0$, so $\theta_{n,s}=\tau_n$.
When this holds for all selectable controls, every choice identifies the same
ATT. Section~\ref{sd:section} bounds ATT coverage when the rule can select a
comparison with a nonzero untreated-gap change.

\subsection{When ranking controls changes the size of ATT tests}\label{cmpv2:failure}
We average within-choice error variances and compare them with the variances of fixed-control estimates, using the same weights. Let $\widehat\tau_s$ be the ATT estimate using candidate $s$, $\tau$ its common target, $U_s=\widehat\tau_s-\tau$ its estimation error, $S$ the choice made from pre-treatment outcomes, and $p_s=\Pr(S=s)$. Define
\begin{equation}\label{v58:main-expected-ratio}
 \mathcal R_E=\frac{\sum_{s:p_s>0}p_s\Var(U_s\mid S=s)}
 {\sum_s p_s\Var(U_s)}.
\end{equation}
The weights hold fixed how often each control is used. In the models below, every fixed comparison has variance one and the error has mean zero conditional on each control choice. Then $R=\mathcal R_E=\Var(U_S)$. More generally, differences in conditional mean errors across selected controls also contribute to the reported error variance, as shown in equation~\MainEqRef{v58:variance-ratio-distinction}.

Consider two controls with centered pre-treatment discrepancies $Z_1,Z_2$. A researcher chooses $S=\arg\min_s Z_s^2$. Let $Z_1,Z_2,e_1,e_2$ be independent standard normal variables, and set
\begin{equation}\label{cmpv2:latent}
 U_1=aZ_1+bZ_2+\sqrt{1-a^2-b^2}e_1,\qquad
 U_2=bZ_1+aZ_2+\sqrt{1-a^2-b^2}e_2,
 \quad a^2+b^2<1.
\end{equation}
For each fixed comparison, $a$ is its error's loading on its own discrepancy and $b$ its loading on the rival's discrepancy. The rule selects the smaller absolute discrepancy.

\begin{proposition}[Own fit and competing fit]\label{amp:two-main}
In \eqref{cmpv2:latent}, the reported error has mean zero and
\begin{equation}\label{amp:two-main-ratio}
 R=1+\frac{2}{\pi}(b^2-a^2).
\end{equation}
The same moments hold conditional on either selected control. Under the true common ATT null, the conventional two-sided level-$\alpha$ test, $0<\alpha<1$, rejects more often than $\alpha$ conditional on either choice when $|b|>|a|$, less often when $|a|>|b|$, and with probability $\alpha$ at equality. If $a=0$ and $0<|b|<1$, overrejection persists when the winner must also satisfy $|Z_s|\le c$ for fixed $c>0$. For ranking alone and $a=0$, size approaches $1-(1-\alpha)^2$ as $|b|$ approaches one.
\end{proposition}
Appendix~\MainRef{app:selection-example} proves the coverage comparisons and constructs a three-period DiD with positive definite full covariance and parallel trends for both controls. At $a=0,b=.9$, the conventional 95\% interval has coverage 90.617\% after selection by fit (equation~\MainEqRef{cmpv2:quadrature}). Multiplying discrepancies, errors, and interval widths by a common positive factor preserves selection and coverage. With known covariance, our selective interval retains the estimate and covers exactly 95\% conditional on either control. Online Appendix Section~\SuppRef{diagv63:correlated} gives variance results for correlated candidates, and equation~\SuppEqRef{amp:two-nonnormal} extends the variance identity to independent, identically distributed symmetric nonnormal fit fluctuations.

\subsection{The expected variance as a function of the candidate pool}\label{amp:geometry-main-section}
To separate candidate count from dependence, we model the shocks shared by candidates' fit scores and ATT errors. Fix unit vectors $u_1,\ldots,u_J\in\R^m$ before observing outcomes. The common fluctuation has dimension $m$. Let $I_m$ be the $m$-dimensional identity and $\|\cdot\|$ the Euclidean norm. Let $Z,\zeta_0,\ldots,\zeta_J$ be independent $N_m(0,I_m)$ vectors and let $e_1,\ldots,e_J$ be independent standard normal variables, independent of them. For $0<\rho<1$ and $\delta\ge0$, define
\begin{align}
 Q_s^\delta&=\|(I_m-u_su_s')Z+\delta(\zeta_0-\zeta_s)\|^2,
 &S_\delta&=\arg\min_s Q_s^\delta,\label{amp:geometry-model}\\
 U_s&=\rho u_s'Z+\sqrt{1-\rho^2}e_s.
 \nonumber
\end{align}
At $\delta=0$, the rule selects a direction with the largest squared projection of $Z$, using a fixed priority for ties. The projection enters that candidate's ATT error with loading $\rho$. The scale $\delta$ multiplies independent regional shocks, with $\zeta_0$ shared across comparisons. Let $C=(u_s'u_t)_{s,t\le J}$ be the correlation matrix of the projections $u_s'Z$. Online Appendix Section~\SuppRef{disp:gaussian} constructs these scores and errors in a DiD panel with parallel trends for every control and positive definite full covariance for $\delta>0$.

For $\delta>0$, let $w_s(z,h;\delta)$ be the probability that $s$ has the smallest score conditional on $Z=z,\zeta_0=h$. Online Appendix equation~\SuppEqRef{v58:selection-weight} evaluates this probability by a one-dimensional integral. Define
\begin{equation}\label{v58:main-functional}
 G_\delta(C;m)=\E_{Z,\zeta_0}\sum_s(u_s'Z)^2w_s(Z,\zeta_0;\delta),
 \qquad G_0(C)=\E\max_s W_s^2,\quad W\sim N_J(0,C).
\end{equation}
For $\delta=0$ write $G_0(C;m)=G_0(C)$. The matrix $C$ can be singular when candidates track the same fluctuation.

\begin{theorem}[Expected variance and candidate dependence]\label{amp:geometry-main}
For every fixed candidate pool in \eqref{amp:geometry-model},
\begin{equation}\label{amp:geometry-main-zero}
 R_\delta(C;m,\rho)=\mathcal R_E
 =1+\rho^2\{G_\delta(C;m)-1\}.
\end{equation}
The function depends on the candidate directions only through $C$ and $m$. It is continuous in $\delta$, with
\begin{equation}\label{amp:geometry-main-noise}
 \lim_{\delta\downarrow0}R_\delta=1+\rho^2\{G_0(C)-1\},\qquad
 \lim_{\delta\to\infty}R_\delta=1.
\end{equation}
For $q>1$, $R_\delta>q$ if and only if
\begin{equation}\label{amp:geometry-main-ten}
 G_\delta(C;m)>1+\frac{(q-1)}{\rho^2}.
\end{equation}

\end{theorem}
Online Appendix Theorem~\SuppRef{v58:normal-functional} proves the exact function for each fixed candidate pool.

The dimension spanned by the candidate directions bounds $G_\delta$. At $\rho=.95$, a variance ratio above ten requires $G_\delta>1+9/.95^2$ and hence at least eleven dimensions. Well-spread directions raise the attainable amplification, and path noise can reduce it. At zero path noise, duplicating a direction or its negative leaves $G_0(C)$ unchanged. Online Appendix Theorem~\SuppRef{v57:normal-geometry} gives explicit bounds.

For general normal outcomes, fit gaps, covariance of fit inputs and ATT errors, and the selection regions determine the expected-variance ratio (Theorem~\MainRef{v58:expected-variance-normal}). Online Appendix Section~\SuppRef{disp:gaussian} illustrates this ratio. Under a true common ATT null, Online Appendix Proposition~\SuppRef{v57:normal-coverage} derives false rejection in the projection model from its tails. At zero path noise, two or more nonparallel directions raise size above the nominal level. Duplicating a direction adds no further increase. The result gives a sufficient condition for inflation to persist at positive path noise.

\paragraph*{How high can the false-rejection rate be?}
For the usual two-sided test, write $Z_s^*$ for the standardized statistic of candidate $s$ and $c$ for the standard normal $1-\alpha/2$ quantile. Under the joint Gaussian null, the probability that all $|Z_s^*|\le c$ is at least $(1-\alpha)^J$. Any rule that always selects one comparison therefore has unconditional false-rejection probability at most $1-(1-\alpha)^J$. The bound is sharp even for selection based only on pre-treatment fit: with orthogonal directions, nearly noiseless fit, and a strong fit--error link, the selected absolute error approaches the largest of $J$ independent absolute standard normals. With ten controls, size can approach 40.13\% for a nominal 5 percent test (Online Appendix Corollary~\SuppRef{v91:sharp-size}). This is a worst-case result across fixed designs.

\subsection{Conditional error variance under nonnormal outcomes}\label{amp:nonnormal-main-section}
Under nonnormal outcomes, selection can also favor high conditional error variance even when all covariances between pre-treatment fit information and ATT errors are zero. Let $X$ collect all pre-treatment selection information and let $S=S(X)$. Set $I_s=\ind\{S=s\}$, and assume every fixed error has mean zero and variance one. Define $q_s(X)=\E(U_s^2\mid X)$ and $\nu_s=\E(U_s\mid S=s)$ for $p_s>0$. For any such joint distribution,
\begin{equation}\label{amp:general-main}
 \mathcal R_E=1+\sum_s\Cov\{I_s,q_s(X)\}-\sum_{s:p_s>0}p_s\nu_s^2.
\end{equation}
Here $q_s(X)$ is a conditional second moment. When $\E(U_s\mid X)=0$ for every $s$, each covariance measures whether a control tends to be selected when its conditional error variance is high. Online Appendix Proposition~\SuppRef{amp:conditional-moments} gives the unequal-variance version; Theorem~\MainRef{v58:expected-variance-normal} specializes it to normal outcomes through derivatives of selection probabilities.

\begin{theorem}[Close fit with higher conditional uncertainty]\label{amp:rare-main}
Let $J\ge1$ and let $A_1,\ldots,A_J$ be independent Bernoulli indicators
of close pre-treatment fit with probability $\epsilon\in(0,1)$, measurable
from $X$. Use $\lambda\in(0,1)$ for the variance contribution of the close-fit component in this model. Suppose
\[
 \E(U_s\mid X)=0,\qquad
 \Var(U_s\mid X)=1-\lambda+\frac{\lambda A_s}{\epsilon},
 \qquad A_S=\max_{s\le J}A_s.
\]
Every fixed comparison and the selected comparison are unbiased, every fixed error has variance one, and
\begin{equation}\label{amp:rare-main-ratio}
 R=1+\lambda\left\{\frac{1-(1-\epsilon)^J}{\epsilon}-1\right\}.
\end{equation}
For any $q>1$, $R>q$ exactly when $\frac{[1-(1-\epsilon)^J]}{\epsilon}>1+\frac{(q-1)}{\lambda}$. The ratio increases with candidate count, with diminishing increments. Online Appendix Corollary~\SuppRef{amp:rare} gives the comparative statics and limiting ratios.
\end{theorem}
Whenever a candidate has the close-fit state $A_s=1$, the researcher selects a candidate in that state and thereby selects greater conditional ATT uncertainty. Online Appendix Corollary~\SuppRef{amp:rare} constructs this selection using centered pre-treatment fit in a three-period DiD with independent regions, positive definite full covariance, and parallel trends for every control. Every covariance between the pre-treatment changes used for fit and the ATT errors is zero. The jointly normal law with the same full means and covariance has ratio one. At a fixed contribution $\lambda$ to marginal variance, rarer close-fit states carry larger conditional variance. With independent states, adding candidates raises this probability. As $J$ increases, the variance ratio approaches $1+\lambda(1/\epsilon-1)$.

Candidate dependence also changes the chance of finding a close-fit state. Retain the conditional moments and selection rule in Theorem~\ref{amp:rare-main}, but allow arbitrary dependence among the indicators with $\Pr(A_s=1)=\epsilon$. For $\pi_{\rm close}=\Pr(\bigcup_s\{A_s=1\})$,
\begin{equation}\label{v60:main-dependent-ratio}
 \mathcal R_E=1+\lambda(\frac{\pi_{\rm close}}{\epsilon}-1),\qquad
 1\le\mathcal R_E\le1+\lambda\{\min(J,\frac{1}{\epsilon})-1\}.
\end{equation}
Both bounds are attainable. For $K$ independent groups sharing a state within each group, replace $J$ by $K$ in \eqref{amp:rare-main-ratio} (Online Appendix Corollary~\SuppRef{v60:close-fit-groups}). When candidates share close-fit states, researchers can search many labels while encountering few distinct opportunities to select high conditional uncertainty. Online Appendix Theorem~\SuppRef{v60:dependent-close-fit} gives the proof; Proposition~\SuppRef{v60:dependent-close-fit-panel} constructs DiD panels with identical regional marginal distributions and full covariance but different amplification and coverage. These finite-sample constructions concern variance and conventional coverage. Validity of the feasible interval relies on the joint normal approximation and covariance consistency in Section~\ref{sec:asymptotics}.

Table~\ref{tab:amplification-conditions} summarizes when searching across controls can change uncertainty. When several controls compete, retain the selection rule, pre-treatment comparisons for all candidates, and their ATT estimates. Our normal approximation uses the joint covariance of the fit inputs and ATT estimates. If each candidate's selection event is independent of its fixed-comparison coverage event, coverage for that candidate is preserved whenever it is selected with positive probability. A unique population-best fit separated by a positive gap is selected with probability approaching one when fitted scores are consistent.

\begin{table}[tbp]\centering\setlength{\tabcolsep}{4pt}
\caption{Candidate pools and the conditions for variance amplification}\label{tab:amplification-conditions}
\begin{tabular}{@{}>{\raggedright\arraybackslash}p{.24\linewidth}>{\raggedright\arraybackslash}p{.38\linewidth}>{\raggedright\arraybackslash}p{\dimexpr.38\linewidth-4\tabcolsep\relax}@{}}
\toprule
Candidate configuration & Mathematical determinant & Implication for inference\\\midrule
One clearly best fit & Population fit gap and consistency of estimated fit scores & Selection converges to the population-best comparison.\\
Similar fits in the same direction & Duplicate directions leave $G_0(C)$ unchanged at zero path noise & Candidate count alone need not increase error variance.\\
Similar fits in different directions & Directional span and spread in $C$, loading $\rho$, path noise $\delta$ & Broadly spread contrasts and strong fit--error dependence can yield large variance increases.\\
A few tight clusters & Zero-noise gain from cluster centers and within-cluster angular spread & Distinct shock directions matter more than candidate labels.\\
Fit-dependent conditional variance & $\Cov\{I_s,q_s(X)\}$, beyond fit--error covariance & Close fit can select high-variance errors even with zero fit--error correlations.\\\bottomrule
\end{tabular}
\note{Rows refer to the model-specific results in Section~\ref{sec:setup} and the corresponding Online Appendix proofs. The target is the ratio of expected selected-comparison and fixed-comparison variances. For models in which the error has mean zero conditional on each selected control, it also equals the unconditional reported-error variance ratio. Coverage is derived separately.}
\end{table}

\subsection{Representing the researcher's selection rule}\label{poly:setup}
A researcher can verify whether a rule falls in our class by writing down the comparisons that decide the report. Begin with a finite menu. At each step, compare two scores, check a cutoff, or follow the recorded priority. If these checks are polynomial sign comparisons in the joint data summaries, they give (S). The recorded rule can use earlier outcomes to determine which comparison comes next.

\paragraph*{Selection structure (S).}
Fix a dimension $d$ for the outcome vector $y\in\mathbb R^d$, a finite set $\mathcal S$ of reporting labels, and integers $M_0,D\ge1$. The label records the comparison under a fixed loss, the rule and comparison under rule search, or a class of compatible cohort paths in Section~\ref{sd:section}. Let the known nonzero polynomials $p_{n1},\ldots,p_{nM_n}$ have degree at most $D$, with $0\le M_n\le M_0$. A predetermined map $\ell_n:\{-1,0,1\}^{M_n}\to\mathcal S$ assigns a report to each configuration of negative, zero, and positive signs:
\begin{equation}\label{poly:selector}
 S_n(y)=\ell_n\{\operatorname{sgn}p_{n1}(y),\ldots,
                         \operatorname{sgn}p_{nM_n}(y)\},
 \qquad E_{n,s}=\{S_n(Y_n)=s\},\quad s\in\mathcal S.
\end{equation}
The coefficients, map and deterministic nonzero contrasts $c_{n,s}$ may vary with $n$; the dimension, degree bound, maximum number $M_0$ of defining polynomials and number of labels remain fixed. Constant polynomials and $M_n=0$ are allowed. Within each loss region, we resolve structural ties using the recorded priority and remove identically zero comparisons. Population fit scores may be equal or zero. For a strict sign vector $\sigma$, let $\sigma_k$ be its $k$th entry.
\begin{align}
 \mathcal B_n&=\bigcup_{k=1}^{M_n}\{p_{nk}(Y_n)=0\},\label{poly:boundary}\\
 \mathcal E_{n,s}^{\circ}&=
   \bigcup_{\substack{\sigma\in\{-1,1\}^{M_n}:\ell_n(\sigma)=s}}
   \bigcap_{k=1}^{M_n}\{y\in\mathbb R^d:\sigma_kp_{nk}(y)>0\}.\label{poly:signcells}
\end{align}
The event $E_{n,s}$ includes all samples assigned label $s$, including ties assigned by priority. The set $\mathcal E_{n,s}^{\circ}$ collects all its open sign regions, so outside $\mathcal B_n$ the event is exactly $Y_n\in\mathcal E_{n,s}^{\circ}$. The ranking-only inequalities $\widetilde E_{n,s}$ in \eqref{eq:joint} describe one particular rule. When several paths produce the same report, we retain the union of their regions.

Squared-error fit compares quadratic forms. Absolute-error and Huber losses compare low-degree polynomials on regions defined by gap signs and cutoffs. For a fixed candidate they use the same DiD contrast. Appendix~\MainRef{poly:piecewise-supp} gives the sign rules.

Finite sequences of these decisions satisfy (S) when their possible branches use a bounded collection of such comparisons. Online Appendix Section~\SuppRef{supp:replayable-selectors} gives the construction, including thresholds based on estimated standard errors, worst-period fit, and fixed-design regression screens.

\subsection{Choosing controls by both outcome levels and paths}\label{levelv62:main}
The centered criterion in \eqref{eq:selection} allows a persistent pre-treatment level difference. To match levels and paths, write the pre-treatment treated-minus-candidate gaps as the vector $d_j(y)=D_jy$. Its component at date $t\in P$ is $d_{jt}(y)$. Set $m_P=|P|$ and let $\mathbf1_{m_P}$ and $I_{m_P}$ denote the vector of ones and the identity matrix of dimension $m_P$. Write $\bar d_j(y)=\mathbf1_{m_P}'d_j(y)/m_P$, $H_{m_P}=\mathbf1_{m_P}\mathbf1_{m_P}'/m_P$, and $P_{m_P}=I_{m_P}-H_{m_P}$. For a fixed weight $\lambda\ge0$, consider
\begin{align}
 Q_j^{(\lambda)}(y)
 &=d_j(y)'(P_{m_P}+\lambda H_{m_P})d_j(y)\notag\\
 &=\|P_{m_P}d_j(y)\|^2+\lambda m_P\bar d_j(y)^2,
 \qquad S_\lambda(y)=\arg\min_jQ_j^{(\lambda)}(y).
 \label{levelv62:loss}
\end{align}
At $\lambda=0$ we use centered gaps. At $\lambda=1$ the loss is $\sum_{t\in P}d_{jt}(y)^2$. Values $0<\lambda<1$ downweight the average level gap relative to departures from that gap. Each rule fixes the candidate pool, periods, weight, and tie-breaking order.

Each rule uses the candidate-specific DiD contrast in \eqref{eq:target}. If two rules select the same control, they give the same estimate and conventional standard error at fixed covariance, but can condition on different selection events. We compare intervals within each selected comparison and report changes in the selected control separately. Under \eqref{sd:single-decomposition}, every selectable control that satisfies parallel trends identifies the same ATT. Parallel trends permits pre-treatment level gaps.

\begin{corollary}[Inference after selecting on levels and paths]
\label{levelv62:corollary}
Fix the finite candidate pool, periods, nonzero contrasts, and
tie-breaking order, and let $\lambda^{(n)}\ge0$ be known and deterministic
for the stated candidate pool.
The selector in \eqref{levelv62:loss} satisfies (S) with degree at most two. In the known-covariance normal model, Lemma~\ref{lem:gaussian} gives exact coverage conditional on any positive-probability choice with nonzero contrast variance. Under Assumption~\ref{ass:array},
Theorems~\ref{thm:plugin} and~\ref{rev:endpoints} apply when the variance,
conditional direction, residual, and complete selection set are
recomputed using the estimated covariance.
\end{corollary}
\noindent The comparison matrices are
$K_{n,j}=D_j'(P_{m_P}+\lambda^{(n)}H_{m_P})D_j$:
apart from boundaries, selecting $s$ requires
$y'(K_{n,j}-K_{n,s})y\ge0$ for every competitor $j$. Appendix~\MainRef{levelv62:fiber} gives the full selection-set construction and implementation; Section~\MainRef{levelv62:channel-proof} proves the covariance calculation below.
Coverage and endpoint equivalence inherit the selection-probability
conditions in Theorems~\ref{thm:plugin} and~\ref{rev:endpoints}.
The interval-length equivalence compares \emph{selective} intervals
computed using estimated and reference covariance.

The DiD estimate subtracts the same pre-treatment mean used to choose the control by level fit. The following calculation separates this dependence
from correlation between outcomes at different dates.

\begin{proposition}[What centering removes from selection]
\label{levelv62:channel}
For each candidate $j$, let $e_{jP}\in\R^{m_P}$ and $e_{jA}$ be the sampled pre-treatment gaps and post-treatment average gap minus their respective means. Assume joint normality and, for every candidate pair $j,k$, the following structure with scalar coefficients $a_{jk},b_{jk},\gamma_{jk}$:
\begin{equation}\label{levelv62:exchangeable}
 \Cov(e_{jP},e_{kP})=a_{jk}I_{m_P}+b_{jk}\mathbf1_{m_P}\mathbf1_{m_P}',
 \qquad \Cov(e_{jP},e_{kA})=\gamma_{jk}\mathbf1_{m_P}.
\end{equation}
Put $C_j=P_{m_P}e_{jP}$, $\bar e_{jP}=\frac{\mathbf1_{m_P}'e_{jP}}{m_P}$, and
$U_k=e_{kA}-\bar e_{kP}$. Then
\begin{equation}\label{levelv62:channel-identities}
 \Cov(C_j,U_k)=0,\qquad
 \Cov(\bar e_{jP},U_k)=\gamma_{jk}-\frac{a_{jk}}{m_P}-b_{jk}.
\end{equation}
All centered gap errors are jointly independent of all DiD errors.
Selection using only centered gaps therefore leaves the fixed-comparison
normal error distributions unchanged, while using levels can introduce
dependence through the second covariance.
\end{proposition}
The applications use their full joint covariances. In the additive model of Appendix~\MainRef{levelv62:channel-proof}, each gap error contains the same persistent component across dates and an independent transitory component. When transitory errors are uncorrelated across dates, $\gamma_{jk}=b_{jk}$ and the level covariance is $-a_{jk}/m_P$. Thus even strong positive dependence across dates can coexist with independence between centered fit and the DiD error. Subtracting the pre-treatment mean can couple level fit and the DiD error even with independent dates. For a joint outcome vector $Y$ with covariance $\Sigma$, let $T_s=c_s'Y$ be the effect estimate for comparison $s$.
\begin{equation}\label{levelv62:general-channel}
 \Cov(P_{m_P}d_j(Y),T_s)=P_{m_P}D_j\Sigma c_s,\qquad
 \Cov(\bar d_j(Y),T_s)=\frac{\mathbf1_{m_P}'D_j\Sigma c_s}{m_P}.
\end{equation}
Researchers evaluate these covariances for every candidate entering the
selection rule.

Adding levels can change the selected control. For a given selected control, it can also widen or shorten the selective interval by changing which contrast values are consistent with selection. We report separate 95\% intervals for the fixed weights $\lambda\in\{0,.25,.5,1\}$; the replication materials give the full grid. If researchers choose the weight after examining effects, the joint reporting event must include that choice (Section~\ref{simv3:search}).

\subsection{Reporting after screening, ranking, or both}\label{sec:reporting-design}
Table~\ref{tab:scope-rules} shows how familiar workflows determine a report. For a single control, $s$ denotes the reported candidate. Ties follow the recorded priority. A fit criterion can compare levels, changes, or both and can penalize large gaps or use the worst-fitting period. Specifying the comparisons, cutoffs, their order, and what happens at a tie lets us replay the complete rule for alternative outcomes.

\begin{table}[!htb]
\centering\setlength{\tabcolsep}{4pt}
\caption{From control-selection workflows to the event used for inference}
\label{tab:scope-rules}
\begin{tabularx}{\linewidth}{@{}p{.38\linewidth}>{\raggedright\arraybackslash}X@{}}
\toprule
Researcher's workflow & What must happen to make this report\\\midrule
Fix one control, then test & The prespecified control $s$ passes.\\
Rank the candidate controls & $s$ has the best fit in the candidate pool.\\
Rank, then screen the winner & $s$ has the best fit and passes the screen.\\
Screen, then rank the survivors & $s$ passes and has the best fit among those that pass.\\
Search over fit rules & The search reports the observed rule and its control, using every comparison that determined the result.\\
Screen, then choose the strongest effect & $s$ passes and has the largest absolute effect statistic among those that pass.\\
Choose controls across cohorts & The cohort choices produce the reported path or another path assigned to the same report and contrast.\\
\bottomrule
\end{tabularx}
\end{table}

Each report combines the decisions that can lead to it, including all paths assigned to the same label and contrast. If a rule can return no comparison, that outcome has its own label, and its formal contrast in (S) is never reported. Section~\ref{simv3:search} illustrates rule search. Section~\ref{sd:section} formalizes cohort paths. Online Appendix Section~\SuppRef{supp:replayable-selectors} shows how to check the sign representation of fit rankings, normalization, and test statistics.

\begin{proposition}[Conventional inference for a prespecified comparison]\label{prop:fixed-screen}
Fix the comparison before observing the outcomes and report it only upon passing its specified pre-tests. Suppose the standardized effect error and pre-test diagnostics converge jointly to a centered Gaussian pair \((Z,G)\) under correct specification, with \(Z\) standard normal in the dimension of the effect. Suppose the reporting indicator converges jointly to \(\mathbf 1\{G\in A_G\}\), where the pre-test acceptance region \(A_G\) is convex, symmetric about zero, and satisfies \(\Pr(G\in A_G)>0\). For any convex coverage region \(C\) symmetric about zero with a boundary of probability zero under \(Z\), the asymptotic conditional coverage equals \(\Pr(Z\in C\mid G\in A_G)\) and is at least \(\Pr(Z\in C)\). For a scalar effect, the corresponding conventional two-sided test has asymptotic conditional size at most its nominal level.
\end{proposition}
\begin{proof}
The joint convergence gives the conditional limit because the reporting probability is positive and the coverage boundary has zero probability. The Gaussian correlation inequality for \(C\times\mathbb R^{\dim G}\) and \(\mathbb R^{\dim Z}\times A_G\) gives \(\Pr(Z\in C,G\in A_G)\ge\Pr(Z\in C)\Pr(G\in A_G)\). Dividing by the reporting probability proves the result.
\end{proof}
This recovers the fixed-comparison result of \citet[Theorem~1 and Corollary~1]{deCDH2026}. Their Assumption~1 gives convex, origin-symmetric acceptance regions with positive probability and zero boundary probability. Online Appendix Section~\SuppRef{supp:fixed-inherited} gives the remaining fixed-comparison results. Ranking adds the requirement that the reported control beat its competitors. The resulting event can be nonconvex. In Figure~\ref{fig:matched-screen}, it raises conventional rejection of the true ATT null even when the winner passes the same screen. Online Appendix Section~\SuppRef{cmpv2:supp} illustrates the two-control reporting regions in Figure~\SuppRef{fig:fixed-versus-ranked} and their rejection rates in Table~\SuppRef{tab:two-control-coverage}. We use the full reporting event to construct the adjusted interval.

\subsection{From the joint covariance to a confidence interval}\label{sec:pivot}
In the reference model, $Y$ has a $d$-dimensional normal distribution with mean $\mu$ and known covariance $\Sigma$. We suppress the sampling index $n$. Let $S$ be the full reporting rule and $E_s=\{S(Y)=s\}$ the event that it reports comparison $s$. Define
\begin{equation}\label{eq:decomp}
 T_s=c_s'Y,\quad v_s=c_s'\Sigma c_s>0,\quad
 b_s=\frac{\Sigma c_s}{v_s},\quad R_s=Y-b_sT_s.
\end{equation}
The vector $b_s$ gives the regression coefficients of the outcome coordinates on the DiD estimate in the normal model. Since $T_s$ and $R_s$ are jointly normal and uncorrelated, they are independent. Conditional on $R_s=r$, a possible DiD estimate $t$ corresponds to the outcome vector $r+b_st$. We determine which values of $t$ would produce the observed report by evaluating the full rule along this line:
\begin{equation}\label{eq:fiber}
 \mathcal T_s(r;\Sigma)=\{t:S(r+b_st)=s\}.
\end{equation}
Restrict each defining polynomial to this line. Use the real roots of the nonzero restrictions to partition the line. Evaluate the rule on every resulting interval, retaining those labeled $s$, and use sign zero for any identically zero restriction. Under (S) and positive definite $\Sigma$, the union has at most $M_0D+1$ components up to sets of normal probability zero.

Let $\theta$ be a trial mean and $v>0$ a variance. For a measurable set $\mathcal T$ of positive probability under $N(\theta,v)$, its restricted cumulative distribution function is
\begin{equation}\label{eq:pivot}
 F_\theta(t;v,\mathcal T)=
 \frac{\int_{\mathcal T\cap(-\infty,t]}e^{-\frac{(x-\theta)^2}{2v}}\,dx}
      {\int_{\mathcal T}e^{-\frac{(x-\theta)^2}{2v}}\,dx}.
\end{equation}
\begin{lemma}[Normal reference distribution]\label{lem:gaussian}
In the normal reference model with known $\Sigma$, fix $0<\alpha<1$. For measurable selection with $\Pr(E_s)>0$ and $v_s>0$, $T_s\mid(E_s,R_s)$ is $N(c_s'\mu,v_s)$ restricted to $\mathcal T_s(R_s;\Sigma)$ for almost every selected residual. Consequently
\begin{equation}\label{eq:selci}
 C^{\mathrm{sel}}_\alpha(Y;\Sigma)=
 \{\theta:\frac{\alpha}{2}\le F_\theta(T_s;v_s,\mathcal T_s(R_s;\Sigma))
                      \le1-\frac{\alpha}{2}\}
\end{equation}
has conditional coverage $1-\alpha$ given $(E_s,R_s)$ and given $E_s$:
\begin{equation}\label{eq:coverage}
 \Pr\{c_s'\mu\in C^{\mathrm{sel}}_\alpha(Y;\Sigma)\mid E_s\}=1-\alpha.
\end{equation}
Its endpoints are finite and unique almost surely. Under (S) and $\Sigma\succ0$, $T_s$ is an interior point of the complete selection set $\mathcal T_s(R_s;\Sigma)$ almost surely.
\end{lemma}

The exact Gaussian result allows any measurable rule with a fixed contrast for each report. With estimated covariance, (S) lets us compute the intervals of estimate values that reproduce the report. Section~\ref{sec:asymptotics} establishes the large-sample guarantee.
\begin{enumerate}
\item Assemble the outcome vector, apply all rankings and screens in their recorded order, and record the label $s$, contrast vector $c_s$, and estimate $T_s$.
\item Estimate the joint covariance $\widehat\Sigma_n$ from the same observations, using the sampling units and dependence model in Section~\ref{sec:asymptotics}.
\item Use $\widehat\Sigma_n$ in \eqref{eq:decomp} to compute the variance, direction, and residual for the observed contrast. Fix that covariance and residual when constructing the conditional reference distribution.
\item Use the polynomial roots to split the line. Evaluate the full rule on each interval, including any unbounded intervals, and retain those that reproduce the report. If the rule estimates a screening variance, recompute it from the summaries on the line (Online Appendix Section~\SuppRef{supp:replayable-selectors}).
\item At the observed $T_s$, solve \eqref{eq:pivot} for the trial mean $\theta$, using the union from step~4. The probabilities $1-\alpha/2$ and $\alpha/2$ give the lower and upper endpoints, respectively.
\end{enumerate}
To test the null mean $\theta_0$, reject if the interval excludes it. The selective two-sided $p$-value is twice the smaller tail probability in \eqref{eq:pivot}, evaluated at $T_s$ under $\theta_0$. Conventional inference uses $\widehat\sigma_s=(c_s'\widehat\Sigma_nc_s)^{\frac{1}{2}}$ and selection set $\R$. Appendix~\MainRef{supp:full-fiber-inversion} details the calculation.

The Buenos Aires worked example reproduces the centered-fit and rank-then-screen rules from their recorded inputs (Section~\ref{appv3:police}). For quadratic level-and-path fit, the replication package also supplies a health-center example.\footnote{See \path{analysis/minimal_chc_inference_v63.py} and the input-and-output guides in the replication package.} Its inputs are the primitive outcome means, their full joint covariance, the maps defining candidate paths, the pre- and post-treatment positions, and the fixed candidate order and level weight. The code returns the conventional and adjusted intervals, resolving identical score functions by recorded priority. If distinct scores tie in the sample, it returns the whole real line for the adjusted interval and flags the boundary. Section~\ref{sec:asymptotics} states the conditions, with proofs in Appendix~\MainRef{app:proofs}.

\section{Feasible inference at perfect fit and population ties}\label{sec:asymptotics}
\subsection{Why the fit comparisons must be retained}
When several controls fit equally well in the population, sampling variation can decide which control the researcher reports. We retain the fit comparisons that determine that choice, including at perfect population fit.

For this illustration, let $d_n$ be a population pre-treatment discrepancy vector, $\widehat d_n$ its sample counterpart, and $Z_{d,n}=\sqrt n(\widehat d_n-d_n)$ its scaled estimation error. Assume that $Z_{d,n}$ is bounded in probability, written $O_p(1)$. With the Euclidean norm,
\[
 \|\widehat d_n\|^2=\|d_n\|^2+2n^{-\frac{1}{2}}d_n'Z_{d,n}+n^{-1}\|Z_{d,n}\|^2.
\]
At perfect fit ($d_n=0$), the quadratic sampling term is the only random term and is $O_p(n^{-1})$. At $d_n=h/\sqrt n$ for a fixed vector $h$, the linear and quadratic sampling terms are both of order at most $n^{-1}$ and can affect ranking. Retaining the original comparisons also handles ties at positive population losses.

\subsection{Estimating covariance from the observations used for selection}
We estimate covariance from the same observations used for control selection and effect estimation, accounting for sampling scales through the known matrix $L_n$.

\begin{assumption}[Joint normal approximation and covariance consistency]\label{ass:array}
Let $L_n$ be known, deterministic, and nonsingular, and put $Z_n=L_n^{-1}(Y_n-\mu_n)$. Let $I_d$ be the $d$-dimensional identity and $N_d(0,V_n(P))$ the centered normal distribution with reference covariance $V_n(P)$. For all sampling laws $P\in\mathcal P_n$ and fixed bounds $0<\underline v<\overline v<\infty$, suppose
\begin{align}
 &\underline v I_d\preceq V_n(P)\preceq\overline v I_d,\qquad
 \sup_{P\in\mathcal P_n}d_{\mathrm{BL}}\{\mathcal L_P(Z_n),N_d(0,V_n(P))\}\to0,
       \label{eq:uniform-clt}\\
 &\sup_{P\in\mathcal P_n}P\{\|\widehat V_n-V_n(P)\|>\epsilon\}\to0
       \quad(\epsilon>0).\label{eq:uniform-cov}
\end{align}
Here $\mathcal L_P$ denotes law under $P$, $d_{\mathrm{BL}}$ is the bounded-Lipschitz distance, $\widehat V_n$ is symmetric, $\preceq$ is the positive-semidefinite order, and the matrix norm is the operator norm. The normal reference covariance is $V_n(P)$ on the standardized scale and $\Sigma_n=L_nV_n(P)L_n'$ on the outcome-mean scale. The same observations may determine $Y_n$ and $\widehat V_n$.
\end{assumption}
Set $\widehat\Sigma_n=L_n\widehat V_nL_n'$. If necessary, replace each eigenvalue of $\widehat V_n$ below $\varepsilon_n\downarrow0$ by $\varepsilon_n$ before rescaling. Assumption~\ref{ass:array} makes the probability of an adjustment vanish uniformly. On the sample boundary $\mathcal B_n$ of the defining polynomial inequalities, or if the covariance used is not positive definite, define the pivot as $\frac{1}{2}$ and the interval as $\R$. Apply the recorded priority to structurally identical scores first, as in (S). Population fit ties remain part of the regular inference problem. Lemma~\MainRef{rev:boundary} and covariance consistency make the sample exceptions uniformly negligible.

\begin{theorem}[Feasible inference after polynomial selection]\label{thm:plugin}
Under (S) and Assumption~\ref{ass:array}, fix $0<\alpha<1$. Let $U_{n,s}$ be the pivot in \eqref{eq:pivot}, evaluated at $\theta_{n,s}=c_{n,s}'\mu_n$, with the variance, direction, residual, and complete conditional selection set recomputed using $\widehat\Sigma_n$. Then
\begin{equation}\label{eq:pluginjoint}
 D_n:=\sup_{P\in\mathcal P_n}\max_{s\in\mathcal S}\sup_{u\in[0,1]}
 |P(U_{n,s}\le u,E_{n,s})-uP(E_{n,s})|\to0.
\end{equation}
Write $p_{n,s}=P(E_{n,s})$ for the probability of label $s$. The conditional bounds in \eqref{eq:pluginpivot}--\eqref{eq:plugincoverage} apply when $p_{n,s}>0$. The marginal bound in \eqref{eq:pluginmarginal} holds for all label probabilities:
\begin{align}
 \sup_{u\in[0,1]}|P(U_{n,s}\le u\mid E_{n,s})-u|
     &\le \frac{D_n}{p_{n,s}},\label{eq:pluginpivot}\\
 |P\{\theta_{n,s}\in C^{\mathrm{sel}}_\alpha(Y_n;\widehat\Sigma_n)
        \mid E_{n,s}\}-(1-\alpha)|&\le\frac{2D_n}{p_{n,s}},\label{eq:plugincoverage}\\
 |P\{\theta_{n,S_n}\in C^{\mathrm{sel}}_\alpha(Y_n;\widehat\Sigma_n)\}
          -(1-\alpha)|&\le2|\mathcal S|D_n.\label{eq:pluginmarginal}
\end{align}
Conditional coverage converges uniformly to $1-\alpha$ over $p_{n,s}\ge p_0>0$. The marginal bound in \eqref{eq:pluginmarginal} averages over the full label set and requires no lower bound on individual probabilities. Conditional calibration extends to vanishing probabilities when $D_n/p_{n,s}\to0$, uniformly over $p_{n,s}\ge\rho_n$ if $D_n/\rho_n\to0$ for a positive threshold sequence $\rho_n$. With independent Gaussian calibration contrasts, Online Appendix Section~\SuppRef{oa99:calibration} gives a finite-sample alternative for every positive-probability report.
\end{theorem}

The result allows drifting population fit and coefficients with a fixed outcome dimension, label set, and bounds on polynomial count and degree. When every sample yields a report, \eqref{eq:pluginmarginal} averages coverage across choices. Otherwise let $\mathcal S_{\rm rep}\subseteq\mathcal S$ collect the labels that report an effect. Conditional on any such report, the absolute coverage error is at most $2|\mathcal S_{\rm rep}|D_n/P(S_n\in\mathcal S_{\rm rep})$ when the denominator is positive.

For the proof, standardize the argument and normalize the coefficients of each nonzero polynomial in (S):
\begin{equation}\label{eq:normalized-comparison}
 q_{nk}(z)=p_{nk}(\mu_n+L_nz)=\sum_{|\beta|\le D}a_{nk,\beta}z^\beta,
 \qquad \widetilde q_{nk}=\frac{q_{nk}}{\|q_{nk}\|_{\mathrm{coef}}}.
\end{equation}
Here $a_{nk,\beta}$ is the coefficient on $z^\beta$, $\beta=(\beta_1,\ldots,\beta_d)$ is a nonnegative-integer multi-index, $|\beta|=\sum_i\beta_i$, and $z^\beta=\prod_i z_i^{\beta_i}$. The quantity $\|q_{nk}\|_{\mathrm{coef}}=(\sum_{|\beta|\le D}a_{nk,\beta}^2)^{\frac{1}{2}}$ is the Euclidean norm of the coefficient vector. Positive rescaling preserves the signs. The normalized coefficient vectors lie on a unit sphere, so every sequence admits a subsequence with a nonzero polynomial limit. The proof controls estimated covariance, polynomial boundaries, and normal integrals over the complete selection set, including at perfect fit, positive population ties, and repeated roots (Appendix~\MainRef{app:uniform-plugin-proof}).

\subsection{Accuracy of the reported interval endpoints}
Researchers use confidence limits to assess economically relevant effect sizes. For the same data and selected comparison, we compute two selective intervals, one with estimated covariance and one with reference covariance. We recompute the complete selection set in each case. We measure endpoint differences in standard-error units and compare interval lengths by their ratio.

\begin{theorem}[Plug-in equivalence of interval endpoints]\label{rev:endpoints}
Assume the conditions of Theorem~\ref{thm:plugin} and fix $0<\alpha<1$.
For the normal reference covariance $\Sigma_n=L_nV_nL_n'$, set
\[
 \sigma_{n,s}=(c_{n,s}'\Sigma_nc_{n,s})^{\frac{1}{2}},\qquad
 \widehat\sigma_{n,s}=(c_{n,s}'\widehat\Sigma_nc_{n,s})^{\frac{1}{2}}.
\]
For the selected $s$ and the same observed $Y_n$, compare the interval using reference covariance $\Sigma_n$ with its feasible counterpart using $\widehat\Sigma_n$:
\[
C^\circ_{n,s}=[\ell^\circ_{n,s},u^\circ_{n,s}],\qquad
 \widehat C_{n,s}=[\widehat\ell_{n,s},\widehat u_{n,s}].
\]
Each uses its respective covariance for the variance, conditional line, residual, and truncation set. Define the endpoint difference in standard-error units by
\[
A_{n,s}=\frac{
  \max\{|\widehat\ell_{n,s}-\ell^\circ_{n,s}|,
        |\widehat u_{n,s}-u^\circ_{n,s}|\}}{\sigma_{n,s}},
\]
setting $A_{n,s}=\infty$ on the boundary and covariance exceptions defined above.
For any $\epsilon>0$,
\begin{equation}\label{rev:eq:endpoint-equivalence}
 \sup_{P\in\mathcal P_n}\max_s
P\{E_{n,s},\ A_{n,s}>\epsilon\}\longrightarrow0.
\end{equation}
Writing $|C|$ for interval length, we obtain
\begin{equation}\label{rev:eq:length-ratio}
 \frac{|\widehat C_{n,S_n}|}{|C^\circ_{n,S_n}|}\to_p1
\end{equation}
uniformly over the class of distributions, where $\to_p$ denotes convergence in probability. On the exceptional samples where the ratio is undefined, assign it any value.

Both conclusions hold conditional on any label with $P(E_{n,s})\ge p_0>0$ and unconditionally at $S_n$ without a label-probability lower bound. The endpoint conclusion also holds using $\widehat\sigma_{n,s}$ in place of $\sigma_{n,s}$.
\end{theorem}
Appendices~\MainRef{supp:lemmas} and~\MainRef{app:proof-endpoints} establish continuity and inversion over the complete conditional selection set. Lemma~\MainRef{v60:full-endpoint-bounds} bounds endpoint errors and the probability of a very short standardized interval.

\subsection{Unequal regional samples and interval contraction}
Treated and control samples can grow at different rates, even when one region's share of the total sample tends to zero. The following result shows how the two regional sample sizes govern interval width.

\begin{corollary}[Regional implementation and effective sample size]\label{cor:regional}\label{rev:regional-rate}
For a selector satisfying (S) and fixed $0<\alpha<1$, use the independent-region implementation \eqref{eq:regional-means}--\eqref{eq:regional-cov-est}. Suppose that for some $\delta>0$ and finite $M$, with $\|\cdot\|$ the Euclidean norm,
\[
 \sup_{n,j}\E\|X_{nji}-\mu_{nj}\|^{2+\delta}\le M,\quad
 \underline v I_T\preceq\Gamma_{nj}\preceq\overline v I_T,\quad
 \min_jn_j\to\infty.
\]
Assumption~\ref{ass:array} then holds uniformly over distributions with these constants, including regional sample shares tending to zero. For the DiD contrast, let $w\in\R^T$ equal $|A|^{-1}$ on $A$, $-|P|^{-1}$ on $P$, and zero elsewhere. The effective sample size and contrast variance are
\[
 n^{\mathrm{eff}}_{n,s}=(n_0^{-1}+n_s^{-1})^{-1},\qquad
 \sigma_{n,s}^2=\frac{w'\Gamma_{n0}w}{n_0}+\frac{w'\Gamma_{ns}w}{n_s}.
\]
At the selected control, the feasible interval satisfies uniformly
\begin{align}
 |\widehat C_{n,S_n}|&=\Theta_p\{(n^{\mathrm{eff}}_{n,S_n})^{-\frac{1}{2}}\},
       \label{rev:eq:regional-width}\\
 \max\{|\widehat\ell_{n,S_n}-\theta_{n,S_n}|,
          |\widehat u_{n,S_n}-\theta_{n,S_n}|\}
       &=O_p\{(n^{\mathrm{eff}}_{n,S_n})^{-\frac{1}{2}}\}\to_p0.
       \label{rev:eq:regional-shrink}
\end{align}
Here $\Theta_p$ means that the length scaled by $\sqrt{n^{\mathrm{eff}}_{n,S_n}}$ and the reciprocal of that scaled length are both uniformly bounded in probability. For the endpoint errors, $O_p$ gives the upper bound in probability after the same scaling. Both rates hold conditional on labels with $p_{n,s}\ge p_0>0$. Under parallel trends for every selectable candidate, replace $\theta_{n,S_n}$ by the ATT $\tau_n$.
\end{corollary}

The interval can contract in outcome units while the difference between the selective and conventional intervals persists in standard-error units. If $n_0=r$ and $n_s=r^2$ as $r\to\infty$, the effective sample size is asymptotic to $r$, so the treated sample governs the rate (Appendix~\MainRef{supp:proof-regional}).

We summarize interval lengths by medians and upper quantiles because expected length can be infinite even when both endpoints are finite \citep{KL2021}. Online Appendix Section~\SuppRef{supp:length} explains boundary effects and gives conditions under which selective and conventional inference coincide.

\section{Staggered adoption: selecting controls across cohorts}\label{sd:section}
When adoption dates differ, researchers can choose a control for each cohort and post-treatment period and combine the estimates into a policy average or event-time profile. The same region-period outcome can help choose a later cohort's control and estimate an earlier cohort's effect. We use the joint covariance to account for this reuse.

We condition on the reported comparisons, combining selection paths when they give the same coefficient vector for each reported effect.

\subsection{From cohort comparisons to reported effects}
Include each region--period mean once in $Y_n$. Region $j$ adopts at fixed $G_j\in\{2,\ldots,T,\infty\}$, where $G_j=\infty$ denotes a never-treated region. Cohort weights $q_{gj}\ge0$ are positive only for regions with $G_j=g$ and sum to one within that cohort. Let $\kappa\ge0$ be the length of the allowed anticipation window. For each cohort--period cell $a=(g,t)$, $t\ge g$, fix a set of pre-treatment comparison periods $P_a\subset\{u:u<g-\kappa\}$, $|P_a|\ge2$, and period-average weights $\pi_{a,u}\ge0$ summing to one over $u\in P_a$. Assume consistency, no spillovers, and no anticipation more than $\kappa$ periods before adoption.

Let $\mathcal C_a$ be the finite candidate set for cell $a$. Candidate $k\in\mathcal C_a$ combines eligible control regions with fixed weights. Its weights $v_{akj}\ge0$ sum to one, with $v_{akj}>0$ only if $G_j>t+\kappa$, so every contributing region remains untreated and outside its anticipation window at $t$. At date $u\in P_a\cup\{t\}$, let $D_{n,ak,u}=\sum_jq_{gj}\bar X_{nj,u}-\sum_jv_{akj}\bar X_{nj,u}$ be the observed treated-minus-control gap. Write $d_{n,ak}$ for the coefficient vector of the DiD contrast:
\begin{equation}\label{sd:cell-contrast}
 T_{n,ak}=D_{n,ak,t}-\sum_{u\in P_a}\pi_{a,u}D_{n,ak,u}
          =d_{n,ak}'Y_n.
\end{equation}
Let $\mu^g_{nj,t}$ denote the potential outcome mean in region $j$ at $t$ under adoption at $g$, and let $\mu^0_{nj,t}$ denote its untreated mean. The cohort--period ATT is $\tau_{n,g,t}=\sum_jq_{gj}(\mu^g_{nj,t}-\mu^0_{nj,t})$. Define the weighted gap under no treatment as $m^0_{n,ak,u}=\sum_jq_{gj}\mu^0_{nj,u}-\sum_jv_{akj}\mu^0_{nj,u}$, and put $\Delta_{n,ak}=m^0_{n,ak,t}-\sum_{u\in P_a}\pi_{a,u}m^0_{n,ak,u}$. Thus $d_{n,ak}'\mu_n=\tau_{n,g,t}+\Delta_{n,ak}$. Under candidate-specific parallel trends, $\Delta_{n,ak}=0$, so candidate $k$ identifies the ATT for cell $a$. These effects may differ across cohorts and dates \citep{CS2021,SA2021}.

The control choices $S_{n,a}$ satisfy (S). Record these choices for all cells used in the reported estimators as $S_n=(S_{n,a})_a$. For a realized path $s$, let $s_a$ denote the candidate chosen in cell $a$. For each reported effect $\ell$, fixed weights $w_{\ell,a}$ aggregate the cohort--period effects into a policy average or a time-specific effect. Put
\begin{equation}\label{sd:aggregate}
 c_{n,\ell,s}=\sum_aw_{\ell,a}d_{n,a,s_a},\qquad
 c_{n,\ell,s}'\mu_n=\tau_{n,\ell}+\sum_aw_{\ell,a}\Delta_{n,a,s_a},
 \quad \tau_{n,\ell}=\sum_aw_{\ell,a}\tau_{n,g,t}.
\end{equation}
For a nonnegative integer event time $e$, fix reporting cohorts $\mathcal G_e$ with $g+e\le T$ and weights $\omega_{e,g}\ge0$ summing to one. Define
\begin{equation}\label{sd:eventtime-example}
 \tau_{n,e}^{\mathrm{evt}}=\sum_{g\in\mathcal G_e}\omega_{e,g}\tau_{n,g,g+e},
 \qquad
 c_{n,e,s}=\sum_{g\in\mathcal G_e}\omega_{e,g}
             d_{n,(g,g+e),s_{(g,g+e)}}.
\end{equation}
Nonnegative event-time weights define policy averages. Signed weights in \eqref{sd:aggregate} also permit differences between event times; comparisons of dynamics use the same cohorts and weights at each event time.

\subsection{The selection event for a collection of effects}
Choose a map $\rho$ in advance from selection paths to report labels in $\mathcal H$, and set $H_n=\rho(S_n)$. Combine paths only if
\begin{equation}\label{sd:compatibility}
 \rho(s)=\rho(s')=h\quad\Longrightarrow\quad
 c_{n,\ell,s}=c_{n,\ell,s'}=:c_{n,\ell,h}\ne0
 \quad\text{for every reported }\ell.
\end{equation}
The event $E_{n,h}=\{H_n=h\}$ is the union of the selection events for paths assigned to $h$. Equality of coefficient vectors guarantees identical estimates for every $Y_n$. With fixed numbers of reported effects, cohort--period cells, candidates, and defining polynomials, the resulting selection rule satisfies (S). Apply \eqref{eq:decomp}--\eqref{eq:selci} to $c_{n,\ell,h}$ with the same $\widehat\Sigma_n$, using this entire union to construct each conditional selection set. Denote the resulting interval at error probability $\alpha_\ell$ by $C_{n,\ell,h}(\alpha_\ell;\widehat\Sigma_n)$, abbreviated $C_{n,\ell,h}$. The multiplicity correction counts the effects reported together, such as several event-time effects. All intervals condition on the same label $h$.

\begin{corollary}[Joint inference on the effects reported after staggered selection]\label{sd:theorem}\label{sd:causal-coverage}
Suppose Assumption~\ref{ass:array} holds for the joint means and covariance and the combined selection rule satisfies (S). For each reported effect, Theorems~\ref{thm:plugin} and \ref{rev:endpoints} apply to the expected contrast $c_{n,\ell,H_n}'\mu_n$ in \eqref{sd:aggregate}, conditioning on the entire event $E_{n,h}$. For $q$ reported effects and $0<\alpha<1$, choose $\alpha_\ell>0$ with $\sum_\ell\alpha_\ell\le\alpha$ (Bonferroni). Simultaneous coverage of the expected contrasts is at least $1-\alpha-o(1)$ uniformly, both marginally and conditional on any report with $P(E_{n,h})\ge p_0$ for fixed $p_0>0$.

Let $D_n$ be the largest joint pivot error in \eqref{eq:pluginjoint} over reported effects and labels. The event $\mathcal I_n=\{\exists\ell:\sum_aw_{\ell,a}\Delta_{n,a,S_{n,a}}\ne0\}$ records that at least one selected aggregate contrast differs in expectation from its ATT. Bound the probability of this event, overall and jointly with report $h$, by
\begin{equation}\label{sd:bad-selection}
 \eta_n=\sup_{P\in\mathcal P_n}P(\mathcal I_n),\qquad
 \eta_{n,h}=\sup_{P\in\mathcal P_n}P(\mathcal I_n\cap E_{n,h}).
\end{equation}
Marginal simultaneous ATT coverage satisfies
\begin{equation}\label{sd:causal-bound}
 P\{\tau_{n,\ell}\in C_{n,\ell,H_n}(\alpha_\ell;\widehat\Sigma_n)
                \text{ for all }\ell\}
 \ge1-\alpha-2q|\mathcal H|D_n-\eta_n.
\end{equation}
For $p_{n,h}=P(E_{n,h})>0$, conditional simultaneous ATT coverage satisfies
\begin{equation}\label{sd:conditional-causal-bound}
 P\{\tau_{n,\ell}\in C_{n,\ell,h}\ \forall\ell\mid E_{n,h}\}
 \ge1-\alpha-\frac{2qD_n+\eta_{n,h}}{p_{n,h}}.
\end{equation}
If every reported aggregate has zero discrepancy on label $h$ throughout $\mathcal P_n$, including through cancellation of weighted discrepancies, then $\eta_{n,h}=0$. If an aggregate discrepancy on label $h$ is nonzero under $P$, it is the same for every path assigned to $h$ by \eqref{sd:compatibility}; hence $\eta_{n,h}\ge p_{n,h}$ and the bound is uninformative. For any positive threshold sequence $\rho_n$, if $(2qD_n+\eta_{n,h})/\rho_n\to0$, the uniform lower limit of conditional coverage on $p_{n,h}\ge\rho_n$ is at least $1-\alpha$. Use $\max_h\eta_{n,h}$ for all labels. For one reported effect ($q=1$), absolute marginal ATT coverage error relative to $1-\alpha_1$ is at most $2|\mathcal H|D_n+\eta_n$.
\end{corollary}
For example, a researcher may report two event-time ATTs using common cohorts and weights, together with their
difference, with nonzero contrasts $c_{n,e_0,h}$, $c_{n,e_1,h}$, and
$c_{n,e_1,h}-c_{n,e_0,h}$. Compute all three intervals from the same
joint event $H_n=h$ and estimated full covariance, allocating $\frac{\alpha}{3}$
to each. Combine paths when all three coefficient vectors remain unchanged.
Candidate-specific parallel trends give $\eta_n=0$. Appendix~\MainRef{app:staggered-core} proves
the result.

\subsection{How another cohort's choice affects inference}
The same observations can determine one cohort's control and enter another cohort's effect estimate, even when region--period outcomes are independent.

\begin{proposition}[Joint conditioning across cohorts]\label{sd:orthogonality}\label{sd:counterexample}
In the normal reference model with known covariance $\Sigma$, let the rows of a fixed matrix $B$ extract from $Y$ the pre-treatment gaps that determine the report. For a fixed contrast $c$ with $c'\Sigma c>0$, independence of the estimate and the gap vector is equivalent to zero covariance:
\begin{equation}\label{sd:orthogonality-eq}
 c'Y\perp BY\quad\Longleftrightarrow\quad B\Sigma c=0.
\end{equation}
If $B\Sigma c=0$, conditioning on any positive-probability label determined by $BY$, including a combination of compatible paths, gives the conventional normal interval.

For regions $j\in\{E,L,A,B\}$ and dates $u=1,\ldots,4$, let the outcomes $Y_{j,u}$ be independent standard normals. Regions $E$ and $L$ adopt at dates 3 and 4, respectively. Regions $A$ and $B$ remain untreated. With $A$ fixed as the early control, the standardized early DiD estimate is
\[
 W=\frac{Y_{E,3}-Y_{A,3}
       -\tfrac12\sum_{u=1}^2(Y_{E,u}-Y_{A,u})}{\sqrt3}.
\]
The late cohort selects $S_L\in\{A,B\}$ by centered squared fit over periods 2 and 3, minimizing $\frac{\{(Y_{L,3}-Y_{L,2})-(Y_{k,3}-Y_{k,2})\}^2}{2}$ over $k\in\{A,B\}$. Although $W\sim N(0,1)$,
\begin{equation}\label{sd:counterexample-moments}
 \E(W^2\mid S_L=A)=1-\frac{\sqrt3}{4\pi},\qquad
 \E(W^2\mid S_L=B)=1+\frac{\sqrt3}{4\pi}.
\end{equation}
\end{proposition}
The late ranking and the early estimate both use $Y_{A,2}$ and $Y_{A,3}$. If only the early effect is reported, both possible late choices reproduce that report under \eqref{sd:compatibility}. Pooling both paths thus recovers the unconditional distribution of $W$ and the conventional interval in this example. Keeping the late choice in the conditioning event requires the conditional distribution, whose second moments differ in \eqref{sd:counterexample-moments} (Appendix~\MainRef{supp:sd-restored}).

In the four-cohort experiment, nine pooled reports describe the first two cohorts' aggregate effect at adoption. At zero effect, conditioning separately on all 81 paths loses 25.9 percent of the mean conditional Fisher information under pooling, averaged over 200,000 draws. Among the 2,047 draws yielding finite intervals under both procedures in the first 2,048 draws, pooling reduces median length from 7.89 to 6.25. For a fixed residual and report, the information gain equals the probability-weighted variance across paths of the conditional means of the effect estimate, divided by the squared contrast variance. The same normal integrals give the path probabilities and means (Online Appendix Sections~\SuppRef{supp:staggered-simulation} and~\SuppRef{oa99:information}).

\section{Simulations: false rejection of a true zero ATT}\label{sec:simulation}

\subsection{False rejection after ranking and screening candidate controls}\label{size92:main}
Every candidate in Figure~\ref{fig:matched-screen} identifies the same zero ATT. Each candidate's pre-treatment gaps are independent of its effect estimate, but its rivals' gaps can move with that estimate. We rank candidates by the sum of their squared pre-treatment gaps and report the winner if it passes the joint 5 percent pre-trend test. With five distinct candidates, 98,359 of 100,000 draws yield a report. On these same reports, the usual test rejects in 20.87 percent, our adjustment for the screen alone in 20.70 percent, and our adjustment for ranking and screening in 5.06 percent. All three use the same control, estimate, and estimated effect covariance. The usual and fully adjusted rates differ by 15.81 percentage points (paired Monte Carlo standard error, 0.12 points). With ten distinct candidates, the corresponding two rates are 35.57 and 4.96 percent.

Of the 95,092 draws in which candidate 1 passes its screen, it wins against four rivals in 19,529. Ranking retains 20.54 percent of its reports but 85.20 percent of its conventional false rejections. Its usual test rejects in 20.78 percent of the retained draws and 0.93 percent of the 75,563 draws in which it passes but loses. Replaying ranking and screening gives 5.15 percent on its retained draws. For an incumbent tested by the same criterion, if adding rivals only removes reports, its rejection rate rises exactly when the lost reports reject less often than the retained reports, provided both groups occur. Online Appendix Table~\SuppRef{tab:matched-screen} pools the winners.

Keeping the five-candidate rankings and screens fixed, we shift the common post-treatment effect from zero to two true standard errors. On the same 98,359 reports, the fully adjusted test rejects zero in \SelectiveFive\ percent under the null and 55.04 percent under this alternative. Online Appendix Section~\SuppRef{v93:matched-oa} reports results by winner and with an estimated pre-test covariance.

\subsection{Coverage for the control that is reported}
In a separate zero-effect design, the usual 5 percent test rejects in 5.57\% of all 40,000 draws and 23.00\% among the 1,478 draws that select candidate 2. The selective test using the same estimate and sample covariance rejects 6.16\% in those draws (Monte Carlo standard error, or MCSE, 0.63 percentage point). Table~\ref{tab:selection-evidence}, panel A, gives the corresponding coverage rates. In panel B, standardized Gamma innovations yield skewed outcomes and the two candidates have equal positive population losses. The selective intervals both use a normal reference law. Conditional on candidate 2, coverage is 92.97\% with estimated covariance and 95.49\% with true covariance (MCSEs 0.18 and 0.15 percentage point). This comparison holds the reference law fixed while varying the covariance input. In a separate grid their coverage rates move closer at larger sample sizes (Online Appendix Table~\SuppRef{tab:gamma-positive-tie}).

\begin{table}[tbp]\centering\setlength{\tabcolsep}{3pt}
\caption{Confidence interval coverage for all selections and for each reported control}\label{tab:selection-evidence}\label{v63:full-selection-evidence}
\begin{tabular}{lrccc}\toprule
Evaluation & Draws & Conventional & \shortstack{Feasible\\selective} & \shortstack{Known\\selective}\\\midrule
\multicolumn{5}{l}{\textit{A. Normal: separated population fit, regional sizes $(25,50,75)$}}\\
All selections & 40,000 & 94.43 (0.11) & 94.59 (0.11) & 95.10 (0.11)\\
Candidate 1 & 38,522 & 95.10 (0.11) & 94.62 (0.11) & 95.11 (0.11)\\
Candidate 2 & 1,478 & 77.00 (1.09) & 93.84 (0.63) & 94.79 (0.58)\\
\addlinespace
\multicolumn{5}{l}{\textit{B. Gamma: positive population-fit tie, regional sizes $(25,50,75)$}}\\
All selections & 40,000 & 93.65 (0.12) & 94.12 (0.12) & 94.93 (0.11)\\
Candidate 1 & 19,532 & 95.01 (0.16) & 95.32 (0.15) & 94.35 (0.17)\\
Candidate 2 & 20,468 & 92.35 (0.19) & 92.97 (0.18) & 95.49 (0.15)\\\bottomrule
\end{tabular}
\note{Coverage of nominal 95\% intervals in percent. Parentheses give Monte Carlo standard errors in percentage points. Each candidate row uses the draws that selected that control. Feasible and known selective intervals use estimated and true covariance, respectively. Each method uses the same observations and choices. Selective intervals with known covariance are exact in the normal design and use the normal reference law in the Gamma design.}
\end{table}

Across 24 designs with 5,000 draws each, estimated-covariance selective coverage across all choices ranges from 94.06\% to 95.60\%. We vary population fit, relative regional sizes, and the innovation distribution. The treated-region sample size $r$ is 25, 100, or 400. A separate local-mean grid uses 100,000 draws per design (Online Appendix Section~\SuppRef{sim:full-results}).

Online Appendix Figure~\SuppRef{fig:fixed-pool-size} reports the regional and ranking-only simulations. In Gaussian zero-effect designs fixing the observed pre-treatment means and archived covariances, the usual and adjusted tests reject in 3.46 and 5.05 percent of the 39,673 screened Buenos Aires reports, and 4.95 and 5.11 percent of the 50,000 screened Washington reports (Online Appendix Section~\SuppRef{v96:policy-reference}).

\subsection{More data improve covariance estimation while selection adjustment persists}\label{sec:endpoint-sim}
In the normal perfect-fit designs, we compute two selective intervals on the same 1,000 draws per sample size, one with estimated and one with true covariance. The larger absolute difference between corresponding endpoints, in true standard errors of the selected contrast, has median 0.799, 0.415, and 0.204 at $r=25,100,400$. We then compare the selective and conventional intervals using the same estimated covariance. Their median length ratios are 1.536, 1.563, and 1.556, respectively. More data bring the feasible endpoints closer to their reference values while the selection adjustment persists. The 95th-percentile length ratios across the three designs range from 9.16 to 12.31 (Online Appendix Section~\SuppRef{app:endpoint-details}).

\subsection{Selecting controls after seeing effect estimates}\label{simv3:search}
Suppose the researcher screens every candidate and, if any pass, reports the largest absolute effect statistic among them. With a true zero ATT, the usual test rejects in 21.97 percent of reports at five candidates and 39.41 percent at ten. Replaying eligibility and the comparison of effect statistics gives 5.00 and 5.05 percent on the same reports (Online Appendix Section~\SuppRef{cherry93:att-screen}).
A researcher may select the fitting rule that yields the largest or smallest effect estimate. In a fixed menu, rule $k$ selects control $S_k(Y)$ and yields estimate $T_k(Y)=c_{S_k(Y)}'Y$, where $Y$ stacks the outcome means and $c_s$ is the contrast vector for control $s$. Let $A(Y)$ choose the rule maximizing $\varsigma T_k(Y)$, with $\varsigma=1$ for the largest estimate or $-1$ for the smallest. The experimental menu uses squared-error (L2), absolute-error (L1), and Huber losses with thresholds $\frac{.25}{\sqrt{50}}$, $\frac{.5}{\sqrt{50}}$, and $\frac{1}{\sqrt{50}}$. Break control ties by candidate order and rule ties by this menu order. For report $(k,s)$, we condition on
\[
 E_{k,s}=\{A(Y)=k,\ S_k(Y)=s\}.
\]
Polynomial inequalities representing both choices determine this event.

In a normal local design with known covariance and regional sample sizes $(50,100,150)$, the search for the largest estimate reports L1 with candidate 2 in 506 of 40,000 draws. On these same 506 draws, coverage is 94.66\% for the conventional interval, 76.68\% for conditioning on the control choice within L1, and 96.05\% for conditioning on both choices (MCSEs 1.00, 1.88, and .87 percentage points). Across all draws, coverage for the two selective intervals is 95.21\% and 95.32\%, respectively. A near-nominal marginal rate can therefore conceal undercoverage for the final report when the interval accounts only for the control choice. Online Appendix Section~\SuppRef{appv3:rule-search-full} gives the joint event and independent checks.

\section{Discussion}\label{sec:discussion}
Two researchers can report the same control and estimate but need different intervals if only one chose that control after comparing pre-treatment outcomes. A reader assessing the policy effect needs to know which other controls could have won.

To make such comparisons reproducible, report three elements alongside the effect estimate:
\begin{description}
\item[Eligible comparisons.] List the candidates, outcome, periods, sampling units, and how the joint covariance was estimated across candidates and dates.
\item[Choice of comparison.] State the fit criterion and weights, screens and cutoffs, their order, the tie rule, and when no effect is reported. Include any use of effect estimates to choose a specification.
\item[Policy inference.] Name the selected control and report its DiD estimate and selection-adjusted interval. State which effect sizes of policy interest the interval rules out. If paths are pooled, specify which ones give the same coefficient vector for every reported effect.
\end{description}
For rules such as those in Table~\ref{tab:scope-rules} that satisfy (S) in Section~\ref{poly:setup}, Section~\ref{sec:asymptotics} gives conditions for valid inference with estimated covariance. Researchers keep the comparison they selected and replay the decisions that led to it. If every possible estimate on the conditional line gives the same report, the adjusted interval coincides with the conventional one. For quadratic fit comparisons, Online Appendix Section~\SuppRef{diagv63:full-line} gives the check in Proposition~\SuppRef{diagv63:certificate}.

\section*{Acknowledgements}
We used Claude (Anthropic) and ChatGPT (OpenAI) for proof checks with the Lean proof assistant, manuscript revision, English proofreading, consistency checks, and preparation of analysis code. We take responsibility for the manuscript and analysis.

\section*{Funding}
The Japan Science and Technology Agency supported this work through Support for Pioneering Research Initiated by the Next Generation (SPRING, JPMJSP2123) and Core Research for Evolutionary Science and Technology (CREST, JPMJCR2665). We also received support from the Japan Society for the Promotion of Science (JSPS) through its Grants-in-Aid for Scientific Research (JSPS KAKENHI Grant Number JP25K21651).

\section*{Conflict of interest}
The authors declare no known competing financial interests or personal
relationships relevant to this work.

\section*{Data availability}
The replication companion provides the outcome summaries and joint covariance matrices used in the six reanalyses. Its README gives the commands for reproducing the application tables and simulations. The data availability guide lists the source data and construction steps for each application and identifies which inputs require separate access. Third-party materials remain subject to their providers' terms.

\FloatBarrier
\clearpage
\appendix
\renewcommand{\thesubsection}{\thesection.\arabic{subsection}}
\renewcommand{\theequation}{A\arabic{equation}}
\renewcommand{\theHequation}{A\arabic{equation}}
\setcounter{equation}{0}
\renewcommand{\thetheorem}{A\arabic{theorem}}
\setcounter{theorem}{0}
\renewcommand{\theHtheorem}{A\arabic{theorem}}
\section{Proofs of feasible control-selection inference}\label{app:proofs}

\subsection{Boundary control and continuity of the full fiber}\label{supp:lemmas}
For a polynomial $q(z)=\sum_{|\beta|\le D}q_\beta z^\beta$, where $\beta$ is a multi-index of nonnegative integers, write
$\|q\|_{\mathrm{coef}}^2=\sum_\beta q_\beta^2$, and let $\phi_v$ denote the $N(0,v)$ density.
\begin{lemma}[Polynomial boundaries]\label{rev:boundary}
For fixed $d,D$ and $G_V\sim N_d(0,V)$,
\begin{equation}\label{rev:eq:boundary-gaussian}
 \lim_{\eta\downarrow0}\sup_{\substack{\|q\|_{\mathrm{coef}}=1\\
 \underline v I_d\preceq V\preceq\overline v I_d}}
 P\{|q(G_V)|\le\eta\}=0.
\end{equation}
Under (S) and Assumption~\MainRef{ass:array}, the probability of the boundary event $\mathcal B_n$, where at least one nonzero defining polynomial vanishes, tends to zero uniformly.
\end{lemma}
\begin{proof}
A nonzero polynomial has a Lebesgue-null zero set. In the induction on $d$, choose one nonzero coefficient in its expansion in the last coordinate. Off that coefficient's null zero set, the last-coordinate polynomial has finitely many roots. Apply Fubini on bounded boxes. If \eqref{rev:eq:boundary-gaussian} failed, there would be a sequence with $\eta_j\downarrow0$ and probabilities bounded below by a positive constant. Compactness of the unit coefficient sphere and the covariance class would give a subsequence with $q_j\to q$ coefficientwise, $\|q\|_{\mathrm{coef}}=1$, and $V_j\to V\succ0$. For $G\sim N_d(0,I_d)$, $q_j(V_j^{\frac{1}{2}}G)\to q(V^{\frac{1}{2}}G)\ne0$ almost surely, contradicting dominated convergence.

For the second assertion, take any $P_n\in\mathcal P_n$ and normalize as in \MainEqRef{eq:normalized-comparison}. Nonsingular $L_n$ preserves nonzero atoms. A subsequence fixes their number and gives nonzero coefficient limits $q_k$ and $V_n\to V\succ0$. Coefficient convergence gives uniform convergence on compact sets, which together with Assumption~\MainRef{ass:array} gives $\widetilde q_{nk}(Z_n)\Rightarrow q_k(G_V)$. Portmanteau applied to the closed set $\{0\}$ gives vanishing zero probabilities. Sum over at most $M_0$ defining polynomials. A violating sequence admits the same argument, establishing uniformity.
\end{proof}

\begin{lemma}[Full-fiber continuity and inversion]\label{rev:inverse}
Fix an integer $M\ge0$ and a map $\lambda:\{-1,0,1\}^M\to\mathcal S$. For $1\le k\le M$, let $q_{nk}$ have degree at most $D$ and converge coefficientwise to nonzero $q_k$. Suppose $a_n\to a$, $\|a_n\|=\|a\|=1$, $W_n\to W\succ0$, and $z_n\to z$, where every $q_k(z)\ne0$ and their sign vector gives label $s$. Set
\begin{gather*}
 v_n=a_n'W_na_n,\quad \beta_n=\frac{W_na_n}{v_n},\quad t_n=a_n'z_n,\\
 \mathcal K_n=\{u:\lambda((\operatorname{sgn}q_{nk}(z_n+\beta_n(u-t_n)))_k)=s\},\\
 \mathcal F_n(h)=
 \frac{\int_{\mathcal K_n\cap(-\infty,t_n]}\phi_{v_n}(u-h)\,du}
      {\int_{\mathcal K_n}\phi_{v_n}(u-h)\,du}.
\end{gather*}
Define $v,\beta,t,\mathcal K,\mathcal F$ by the limiting expressions. For every finite $H$,
\begin{equation}\label{rev:eq:local-cdf}
 \sup_{|h|\le H}|\mathcal F_n(h)-\mathcal F(h)|\to0.
\end{equation}
For each $u\in(0,1)$ the inverse equations $\mathcal F_n(h_n(u))=u$ and $\mathcal F(h(u))=u$ have finite unique solutions, eventually for $\mathcal F_n$, and $h_n(u)\to h(u)$. For fixed $0<\alpha<1$, the equal-tail interval lengths $h_n(\frac{\alpha}{2})-h_n(1-\frac{\alpha}{2})$ converge to the finite positive limit $h(\frac{\alpha}{2})-h(1-\frac{\alpha}{2})$.
\end{lemma}
\begin{proof}
Each limiting polynomial restricted to the line, $u\mapsto q_k(z+\beta(u-t))$, is nonzero at $t$ and has finitely many roots. Away from their union, signs stabilize, giving $1_{\mathcal K_n}\to1_{\mathcal K}$ almost everywhere. Atom continuity gives neighborhoods of $t$ and eventually $t_n$ inside their fibers. With $1\{x\le t_n\}\to1\{x\le t\}$, an envelope $C_He^{-cx^2}$ dominates the supremum over $|h|\le H$ of each integrand difference. Dominated convergence yields uniform integral convergence. The positive continuous limiting denominator has a positive minimum on $[-H,H]$, so the ratios converge uniformly, proving \eqref{rev:eq:local-cdf}.

The neighborhood of $t$ inside $\mathcal K$ gives positive mass on both sides of $t$. Let $U$ have the $N(h,v)$ distribution restricted to $\mathcal K$, and let $\operatorname{Cov}_h$ denote covariance under that law. Normal-envelope differentiation therefore gives $\partial_h\mathcal F=v^{-1}\operatorname{Cov}_h(1\{U\le t\},U)<0$. Bounded positive-measure subsets above $t+\delta$ and below $t-\delta$ give exponential-tilting bounds $\mathcal F(h)\le Ce^{-\frac{h\delta}{v}}$ as $h\to\infty$ and $1-\mathcal F(h)\to0$ as $h\to-\infty$. The same holds eventually for $\mathcal F_n$. Strict monotonicity and compact uniform convergence bracket $h_n(u)$ between $h(u)\pm\epsilon$, proving inverse convergence and positive limiting length. The integrals throughout the argument are over the union of all sign cells assigned to $s$.
\end{proof}

\subsection{Normal pivot}\label{app:proof-gaussian}
\begin{proof}[Proof of Lemma~\MainRef{lem:gaussian}]
For $T=c_s'Y$, $v=c_s'\Sigma c_s$, $b=\frac{\Sigma c_s}{v}$, and $R=Y-bT$, $\Cov(R,T)=0$ gives normal independence. At $R=r$, selection is $T\in\mathcal T_s(r;\Sigma)$, giving the restricted normal density with cumulative distribution function (CDF) \MainEqRef{eq:pivot} at $\theta=c_s'\mu$. Zero-mass fibers contribute zero probability. The probability integral transform gives conditional uniformity and coverage; integration over $R$ gives \MainEqRef{eq:coverage}.

The restricted law is continuous, so its CDF at the observation lies
in $(0,1)$ almost surely. Let $U$ have the $N(\theta,v)$ distribution restricted to the fixed fiber $\mathcal T_s$, and let $\operatorname{Cov}_\theta$ denote covariance under this law. Both sides have positive Lebesgue measure; normal-envelope
differentiation gives
\begin{equation}\label{eq:cdfderivative}
 \partial_\theta F_\theta(T;v,\mathcal T_s)
 =v^{-1}\operatorname{Cov}_\theta(\mathbf 1\{U\leq T\},U)<0.
\end{equation}
The exponential-tilting argument in the proof of Lemma~\ref{rev:inverse} applies to
this fixed measurable fiber and gives $F_\theta(T;v,\mathcal T_s)\to1$ as $\theta\to-\infty$ and $F_\theta(T;v,\mathcal T_s)\to0$ as $\theta\to\infty$, yielding finite unique endpoints. Under (S) and $\Sigma\succ0$, the observed $Y$ is almost surely off every defining polynomial's zero set. Each restriction is then nonzero at $T$, so $T$ is interior to the selection set. Enumerating the real roots on the conditioning line and retaining every accepted interval computes the set up to endpoints of zero normal mass.
\end{proof}

\subsection{Uniform inference with estimated covariance}\label{app:uniform-plugin-proof}
\begin{proof}[Proof of Theorem~\MainRef{thm:plugin}]
For a violation of \MainEqRef{eq:pluginjoint}, the finite label set, bounded atom count, and finite set of sign-to-label maps let us fix a label $s$, atom count, and map $\lambda$ along a subsequence. The violating joint CDF errors are bounded away from zero and at most $P(E_{n,s})$, so pass to a subsequence with $P(E_{n,s})\to p>0$. We divide by $P(E_{n,s})$ for the conditional claim after establishing the joint bound.
Put $r_n=\|L_n'c_{n,s}\|>0$ and $a_n=\frac{L_n'c_{n,s}}{r_n}$. Normalize the atoms as in \MainEqRef{eq:normalized-comparison}. Compactness and Assumption~\MainRef{ass:array} give a further subsequence such that
\[
 \begin{aligned}
 a_n&\to a,\quad V_n\to V\succ0,\quad
 \widetilde q_{nk}\to q_k\text{ coefficientwise},\\
 (Z_n,\widehat V_n)&\Rightarrow(Z,V),\qquad Z\sim N_d(0,V).
 \end{aligned}
\]
Here every $q_k$ is nonzero. Covariance consistency and Slutsky's theorem give the joint convergence for the mean and covariance estimates computed from the same observations. Set
$E^*=\{z:\lambda((\operatorname{sgn}q_k(z))_k)=s\}$.
Lemma~\ref{rev:boundary} and local sign stability give $P(Z\in E^*)=p$ and, in an almost surely convergent Skorokhod coupling, $1_{E_{n,s}}\to1_{E^*}(Z)$.

For positive definite $W$, write
\begin{equation}\label{eq:plugin-line-normalized}
 \mathcal K_n(z,W)=\left\{u:\lambda\left(
  \left(\operatorname{sgn}\widetilde q_{nk}
  \left[z+\frac{Wa_n}{a_n'Wa_n}(u-a_n'z)\right]\right)_k\right)=s\right\}.
\end{equation}
For a trial displacement $h$ from the true target in units of $r_n$, changing the original integration variable to $\theta_{n,s}+r_nu$ gives the pivot at $\theta_{n,s}+r_nh$ as
\[
 \frac{\int_{\mathcal K_n(Z_n,W)\cap(-\infty,a_n'Z_n]}
       \phi_{a_n'Wa_n}(u-h)\,du}
      {\int_{\mathcal K_n(Z_n,W)}\phi_{a_n'Wa_n}(u-h)\,du},
 \qquad W=\widehat V_n.
\]
At the true target $h=0$, Lemma~\ref{rev:inverse} gives pathwise convergence on $E^*$ to the $(Z,V,a)$ full-fiber pivot. By Lemma~\MainRef{lem:gaussian}, this limiting pivot has the uniform distribution on $[0,1]$ conditional on $Z\in E^*$. Elsewhere the selection indicator vanishes eventually. The polynomial-boundary event $\mathcal B_n$ and the event that the covariance used in the pivot is not positive definite have probabilities tending to zero uniformly, so $P(U_{n,s}\le u,E_{n,s})\to up$. A finite grid and monotonicity give uniformity in $u$, contradicting the violation.

Dividing the joint CDF error by $p_{n,s}>0$ bounds the conditional CDF error by $D_n/p_{n,s}$. Evaluating at $1-\frac{\alpha}{2}$ and taking the left limit at $\frac{\alpha}{2}$ bounds the joint coverage error by $2D_n$, allowing point masses. The exceptional convention $U=\frac{1}{2}$, $C=\R$ also satisfies the acceptance condition $\frac{\alpha}{2}\le U\le1-\frac{\alpha}{2}$ for every trial value. Dividing by $p_{n,s}\ge p_0$ or $\rho_n$, and summing over the exact partition, gives all conditional and marginal assertions.
\end{proof}

\subsection{Endpoint equivalence and interval length}\label{simv3:endpoints}\label{app:proof-endpoints}
\begin{lemma}[Complete endpoint and interval-length bounds]\label{v60:full-endpoint-bounds}
Assume the conditions of Theorem~\MainRef{thm:plugin} and fix $0<\alpha<1$.
For the normal reference covariance $\Sigma_n=L_nV_nL_n'$, set
\[
r_{n,s}=\|L_n'c_{n,s}\|,\qquad
 \sigma_{n,s}=(c_{n,s}'\Sigma_nc_{n,s})^{\frac{1}{2}},\qquad
 \widehat\sigma_{n,s}=(c_{n,s}'\widehat\Sigma_nc_{n,s})^{\frac{1}{2}}.
\]
For the selected $s$ and the same observed $Y_n$, compare the interval using reference covariance $\Sigma_n$ with its feasible counterpart using $\widehat\Sigma_n$:
\[
C^\circ_{n,s}=[\ell^\circ_{n,s},u^\circ_{n,s}],\qquad
 \widehat C_{n,s}=[\widehat\ell_{n,s},\widehat u_{n,s}].
\]
Each uses its respective covariance for the variance, conditional line, residual, and truncation set. Outside the boundary and covariance exceptions defined in Section~\MainRef{sec:asymptotics}, define
\begin{align*}
A_{n,s}&=\frac{
  \max\{|\widehat\ell_{n,s}-\ell^\circ_{n,s}|,
        |\widehat u_{n,s}-u^\circ_{n,s}|\}}{\sigma_{n,s}}
 ,\\
R_{n,s}&=\frac{\max\{|\widehat\ell_{n,s}-\theta_{n,s}|,
             |\widehat u_{n,s}-\theta_{n,s}|\}}{\sigma_{n,s}}.
\end{align*}
Thus $A_{n,s}$ measures the endpoint difference between the two intervals, and $R_{n,s}$ measures the distance of the feasible endpoints from the target, both in standard-error units. Set both to $\infty$ on exceptions.
For any $\epsilon>0$, the endpoint equivalence in equation~\MainEqRef{rev:eq:endpoint-equivalence} is
\[
 \sup_{P\in\mathcal P_n}\max_s
P\{E_{n,s},\ A_{n,s}>\epsilon\}\longrightarrow0.
\]
The following bounds control large endpoint errors and small interval widths in standard-error units:
\begin{align}
 &\lim_{M\to\infty}\limsup_{n\to\infty}
   \sup_{P\in\mathcal P_n}\max_s
P\{E_{n,s},R_{n,s}>M\}=0,\label{rev:eq:endpoint-tight}\\
 &\lim_{\eta\downarrow0}\limsup_{n\to\infty}
   \sup_{P\in\mathcal P_n}\max_s
P\{E_{n,s},\frac{|\widehat C_{n,s}|}{\sigma_{n,s}}<\eta\}=0.
                                                        \label{rev:eq:width-lower}
\end{align}
The tightness and lower-length statements also hold for the reference interval $C^\circ_{n,s}$. Together with endpoint equivalence, they give, for the selected label, equation~\MainEqRef{rev:eq:length-ratio}:
\[
 \frac{|\widehat C_{n,S_n}|}{|C^\circ_{n,S_n}|}\to_p1
\]
uniformly over the class of distributions. On exceptions where the ratio is undefined, any value may be substituted.

All statements hold conditionally for $P(E_{n,s})\ge p_0>0$ and, by summing over the fixed label set, unconditionally at $S_n$ without a probability lower bound. The standard-error-normalized bounds also hold using $\widehat\sigma_{n,s}$ on nonexceptional samples.
\end{lemma}
\begin{proof}[Proof of Theorem~\MainRef{rev:endpoints} and Lemma~\ref{v60:full-endpoint-bounds}]
For a fixed label $s$ in the preceding subsequence argument, put $r_n=r_{n,s}$ and $a_n=L_n'c_{n,s}/r_n$. Then
\[
 \left|\frac{\widehat\sigma_{n,s}^2}{\sigma_{n,s}^2}-1\right|
 \le\frac{\|\widehat V_n-V_n\|}{\underline v},\qquad
 \sqrt{\underline v}\,r_n\le\sigma_{n,s}\le\sqrt{\overline v}\,r_n.
\]
The reference and estimated scales are uniformly equivalent. For any putative violating sequence, a vanishing selection probability makes the joint violation probability vanish. Otherwise use the same normalization and coupling as in the preceding proof, with $P(E^*)>0$.
The normalized reference CDF uses $W=V_n$ in \eqref{eq:plugin-line-normalized}; the estimated CDF uses $W=\widehat V_n$. On $E^*$, Lemma~\ref{rev:inverse} gives convergence of both endpoint pairs, centered at $\theta_{n,s}$ and divided by $r_n$, to the same $(h_-,h_+)$ with $-\infty<h_-<h_+<\infty$. The endpoint differences divided by $r_n$ vanish, and $\frac{\sigma_{n,s}}{r_n}$ is bounded above and below by positive constants. Thus the endpoint differences divided by $\sigma_{n,s}$ vanish on the selected event. The polynomial-boundary event $\mathcal B_n$ and the event that the covariance used is not positive definite have probability $o(1)$, proving \MainEqRef{rev:eq:endpoint-equivalence} by the subsequence argument.

For any violating sequence, thresholds $M_j\uparrow\infty$ contradict the finite limiting endpoints on $E^*$; thresholds $\eta_j\downarrow0$ contradict the positive limiting length. This proves tightness and the lower-length bound for both intervals. The absolute length difference is bounded by the sum of the absolute lower- and upper-endpoint differences, hence by $2\sigma_{n,s}A_{n,s}$. On reference lengths at least $\eta\sigma_{n,s}$, the absolute relative length error is therefore at most $\frac{2A_{n,s}}{\eta}$. Apply endpoint equivalence and then let $\eta\downarrow0$, using the lower-length bound, to obtain \MainEqRef{rev:eq:length-ratio}. Summation gives selected-target statements, division by $p_{n,s}\ge p_0$ gives conditional statements, and scale equivalence permits $\widehat\sigma_{n,s}$.
\end{proof}

\subsection{Regional sampling and effective sample size}\label{supp:proof-regional}
\begin{proof}[Proof of Corollary~\MainRef{cor:regional}]
Write $\varepsilon_{nji}=X_{nji}-\mu_{nj}$ and reduce the moment exponent to $0<\delta\le2$, using Lyapunov's inequality if necessary. The resulting moment bound remains uniform. For a fixed nonzero block vector $a=(a_0',\ldots,a_J')'$, independence and the moment bound give
\begin{gather*}
 \sum_{j,i}\E|\frac{a_j'\varepsilon_{nji}}{\sqrt{n_j}}|^{2+\delta}
 \le M\sum_j\|a_j\|^{2+\delta}n_j^{-\frac{\delta}{2}}\to0,\\
 \operatorname{Var}(a'Z_n)=\sum_ja_j'\Gamma_{nj}a_j
 \ge\underline v\|a\|^2.
\end{gather*}
For any sequence with $V_n(P_n)\to V$, Lyapunov's central limit theorem (CLT) and the Cram\'er--Wold device give $Z_n\Rightarrow N_d(0,V)$. A violation of \MainEqRef{eq:uniform-clt} admits such a subsequence by covariance compactness. Along this subsequence, weak convergence makes the distance to $N_d(0,V)$ vanish, while $d_{\mathrm{BL}}\{N_d(0,V_n),N_d(0,V)\}\le\E\|(V_n^{\frac{1}{2}}-V^{\frac{1}{2}})G\|\to0$, $G\sim N_d(0,I_d)$. The triangle inequality therefore makes the distance in \MainEqRef{eq:uniform-clt} vanish, contradicting the assumed violation.

For covariance estimation put $p=1+\frac{\delta}{2}\in(1,2]$. The independent centered products
$\xi_{nji,kl}=\varepsilon_{nji,k}\varepsilon_{nji,l}-(\Gamma_{nj})_{kl}$ have uniformly bounded $p$th moments. Theorem~2 of \citet{vonBahrEsseen1965} gives
\[
 \sup_P\E_P\left|n_j^{-1}\sum_i\xi_{nji,kl}\right|^p
 \le C n_j^{1-p}\to0.
\]
Markov's inequality makes each covariance entry uniformly consistent. Subtracting the sample mean contributes $\bar\varepsilon_{nj}\bar\varepsilon_{nj}'$, whose operator norm has expectation at most $\frac{T\overline v}{n_j}$; the factor $\frac{n_j}{n_j-1}$ also tends to one. Finitely many entries and regions give \MainEqRef{eq:uniform-cov}. The population eigenvalue lower bound and uniform covariance consistency ensure that raising eigenvalues below the prescribed threshold $\varepsilon_n\downarrow0$ occurs with uniformly vanishing probability. All covariances here are computed from the same observations as $Y_n$.

\smallskip\noindent\textit{Interval-length and endpoint rates.}\label{supp:proof-rate}
The contrast has block $w$ in region 0, block $-w$ in region $s$, and zero elsewhere. Independence across regions gives
$\sigma_{n,s}^2=\frac{w'\Gamma_{n0}w}{n_0}+\frac{w'\Gamma_{ns}w}{n_s}$ and
\[
 \sqrt{\underline v}\,\|w\|(n^{\mathrm{eff}}_{n,s})^{-\frac{1}{2}}
 \le\sigma_{n,s}\le
 \sqrt{\overline v}\,\|w\|(n^{\mathrm{eff}}_{n,s})^{-\frac{1}{2}}.
\]
The preceding CLT and covariance-consistency arguments verify Assumption~\MainRef{ass:array}. Since $w\ne0$, Theorem~\MainRef{rev:endpoints} and Lemma~\ref{v60:full-endpoint-bounds} apply. The uniform tightness of $R_{n,s}$ in \eqref{rev:eq:endpoint-tight}, together with $|\widehat C_{n,s}|\le2\sigma_{n,s}R_{n,s}$, proves the upper interval-length order and the endpoint order. The lower-length bound \eqref{rev:eq:width-lower} gives the corresponding tightness bound for $\frac{\sigma_{n,s}}{|\widehat C_{n,s}|}$, the reciprocal of interval length measured in standard-error units, on the joint selected events. In each case summing joint selected-event probabilities over finitely many $s$ proves the assertion at $S_n$.
Moreover $\max_s\sigma_{n,s}^2\le\frac{2\overline v\|w\|^2}{\min_jn_j}\to0$. For any $M$, the probability that a selected endpoint error exceeds $\epsilon$ is at most $P(R_{n,S_n}>M)+1\{M\max_s\sigma_{n,s}>\epsilon\}$. Let $n\to\infty$ and then $M\to\infty$ to obtain \MainEqRef{rev:eq:regional-shrink}. Conditional bounds follow by dividing by $p_{n,s}\ge p_0>0$. Under parallel trends for every selectable candidate, $\Delta_{n,s}=0$ in $\theta_{n,s}=\tau_n+\Delta_{n,s}$, so $\theta_{n,s}=\tau_n$ for the reported comparison.
\end{proof}

\subsection{The three-period selection example}\label{app:selection-example}\label{cmpv2:did-counterexample}\label{cmpv2:correct-pivot}
The moment formula in Proposition~\MainRef{amp:two-main} follows by setting the fit correlation to zero in Online Appendix Proposition~\SuppRef{diagv63:correlated-prop}.

To compare conventional coverage in Proposition~\MainRef{amp:two-main}, put
$D=(Z_1+Z_2)/\sqrt{2}$, $F=(Z_1-Z_2)/\sqrt{2}$,
$u=(a+b)/\sqrt{2}$, $w=(a-b)/\sqrt{2}$, and $\sigma=\sqrt{1-a^2-b^2}$.
Then $D,F$ are independent standard normals and $U_1=uD+wF+\sigma e_1$.
Candidate 1 wins when $DF<0$. Conditional on either sign of $DF$, the
magnitudes $|D|,|F|$ have the same joint law. Given both magnitudes and both signs, $U_1$ is normal. Its absolute mean when candidate 1 wins is
$\bigl|\,|u||D|-|w||F|\,\bigr|$ if $uw>0$, whereas when it loses the
absolute mean is $|u||D|+|w||F|$. The probability of a normal variate
lying in a symmetric interval decreases strictly with its absolute mean.
Thus coverage when candidate 1 wins exceeds coverage when it loses
if $uw>0$. The order reverses for $uw<0$, and equality holds for $uw=0$.
Unconditional coverage is nominal and averages the two cases.
Since $2uw=a^2-b^2$, this gives the stated signs. Interchanging the
candidate labels gives the same result for candidate 2.

For $a=0$, $0<|b|<1$, and the own-candidate test cutoff $c>0$, write $\Phi$ for the standard normal
cumulative distribution function. The probability that candidate 1 wins
and passes, given $U_1=x$, is
$\E[2\Phi(\min\{c,|Z_2|\})-1\mid U_1=x]$.
This follows because $Z_1$ is independent of $(Z_2,U_1)$.
As $|x|$ increases, the absolute value of
$Z_2\mid U_1=x\sim N(bx,1-b^2)$ increases strictly in stochastic order.
The reporting probability therefore increases strictly with $|x|$.
The indicator that the usual interval covers decreases with $|x|$,
so its covariance with that reporting probability is strictly negative.
Dividing their joint probability by the positive report probability
gives strict undercoverage.

Use the special case $a=0$, $b=\rho\in(0,1)$ of equation~\MainEqRef{cmpv2:latent}, and put $V_j=Z_j$, $R_j=U_j$, and $T_j=\tau+R_j$. This is the cross-loading example in Section~\MainRef{cmpv2:failure}. To realize it as regional outcomes, take five independent standard normals $A_{-1},A_0,A_1,L_1,L_2$, independent of $(V_1,V_2,R_1,R_2)$, and, for $j=1,2$, set
\[
\begin{aligned}
 (Y_{0,-1},Y_{0,0},Y_{0,1})&=(A_{-1},A_0,A_1+\tau),\\
 (Y_{j,-1},Y_{j,0},Y_{j,1})
 &=(A_{-1}-L_j-V_j,\ A_0-L_j,\ A_1-L_j-\frac{V_j}{2}-R_j).
\end{aligned}
\]
The untreated mean is zero in every group and period, and treatment adds $\tau$ only to group 0 after period 0. Both comparisons therefore satisfy parallel trends and identify the same ATT. Their centered squared-fit scores are $\frac{V_j^2}{2}$, and their post-treatment gaps minus the two-period pre-treatment average equal $T_j=\tau+R_j$. The covariance of $(V_1,V_2,R_1,R_2)$ has eigenvalues $1-\rho$ and $1+\rho$, each twice. The displayed transformation is invertible: recover $(A_{-1},A_0,A_1)$ from the treated outcomes after subtracting the fixed treatment shift, $L_j$ from the period-0 gaps, $V_j$ from the pre-treatment gap changes, and $R_j$ from DiD. Thus the outcome covariance is positive definite for $0<\rho<1$.

Each fixed $R_j$ is standard normal. Let $\phi$ and $\Phi$ denote the standard normal density and CDF, write $z_q=\Phi^{-1}(q)$ for $0<q<1$, and put $c_\alpha=z_{1-\frac{\alpha}{2}}$. If $X=\max(|Z_1|,|Z_2|)$, its density on $x>0$ is $4\phi(x)\{2\Phi(x)-1\}$. Conditional on the two pre-treatment errors, the selected $R_S$ has mean $\rho V_{3-S}$ and variance $1-\rho^2$. Here $V_{3-S}$ is the other candidate's pre-treatment error and has the larger absolute value, by the rule $S=\arg\min_jV_j^2$. Symmetry gives the selected conventional coverage
\begin{equation}\label{cmpv2:quadrature}
 \int_0^\infty4\phi(x)\{2\Phi(x)-1\}
 \left\{\Phi\!\left(\frac{c_\alpha-\rho x}{\sqrt{1-\rho^2}}\right)
       -\Phi\!\left(\frac{-c_\alpha-\rho x}{\sqrt{1-\rho^2}}\right)\right\}dx.
\end{equation}
At $\rho=.9$ and $\alpha=.05$ this is $0.9061725185$. As $\rho\uparrow1$, $|R_S|\to X$ almost surely, giving limiting coverage $(1-\alpha)^2$.
For the selective calculation with candidate 1, the normal residual fixes $V_1$ and $r=V_2-\rho T_1$. For trial values $T_1=t$, solving the selection inequality $|V_1|\le|r+\rho t|$ gives the full fiber
\[
 \left(-\infty,\frac{-|V_1|-r}{\rho}\right]
 \ \cup\ \left[\frac{|V_1|-r}{\rho},\infty\right).
\]
Lemma~\MainRef{lem:gaussian} gives exact $1-\alpha$ coverage conditional on this residual and candidate 1; symmetry gives the same for candidate 2. Integrating over residuals and summing over the two selection events proves exact selected-ATT coverage.
For a reported candidate passing $|V_j|\le c$ with fixed $c>0$, the own-candidate pretest depends only on its fixed residual coordinate. The same full fiber therefore gives exact coverage conditional on both selection and pretest acceptance.

\subsection{Staggered adoption and joint reports}\label{sd:appendix}\label{app:staggered-core}\label{supp:proof-staggered-sampling}
\begin{proof}[Proof of Corollary~\MainRef{sd:theorem}]
Collect the nonzero defining polynomials of all cohort--period rules into one finite list and apply their predetermined tie priorities. The full path is a deterministic function of this sign vector. A declared coarsening map merges full paths into reported labels; composing the path with this map preserves (S), with fixed complexity. Equality of report coefficient vectors on each coarsened label supplies a deterministic nonzero contrast $c_{n,\ell,h}$ for every report. Theorem~\MainRef{thm:plugin} and Theorem~\MainRef{rev:endpoints} therefore apply to the full union of sign cells assigned to $h$, simultaneously over finitely many $(\ell,h)$. Their maximal joint pivot error is $D_n^{\mathrm{sd}}=D_n\to0$. For reportwise levels $\alpha_\ell$ summing to at most $\alpha$, the left-limit CDF calculation in the proof of Theorem~\MainRef{thm:plugin} gives
\[
 P\{c_{n,\ell,h}^{\prime}\mu_n\notin C_{n,\ell,h},E_{n,h}\}
 \le\alpha_\ell P(E_{n,h})+2D_n^{\mathrm{sd}}.
\]
The union bound over $q$ reports and summation over $h$ give simultaneous failure probability at most $\alpha+2q|\mathcal H|D_n^{\mathrm{sd}}$.
Consistency, eligibility, and the anticipation restriction give the cell identity
$d_{n,ak}'\mu_n=\tau_{n,g,t}+\Delta_{n,ak}$ by adding and subtracting the treated cohort's untreated mean. Aggregation gives the reported ATT plus its weighted parallel-trends discrepancy.

Let $\mathcal I_n=\{\exists\ell:\sum_aw_{\ell,a}\Delta_{n,a,S_{n,a}}\ne0\}$. Outside $\mathcal I_n$, every reported contrast equals its ATT. Thus simultaneous ATT failure is contained in the union of statistical-target failure and $\mathcal I_n$. Its probability is bounded by $\alpha+2q|\mathcal H|D_n^{\mathrm{sd}}+\eta_n$.

Intersecting that inclusion with $E_{n,h}$ gives joint failure at most $\alpha p_{n,h}+2qD_n^{\mathrm{sd}}+P(\mathcal I_n\cap E_{n,h})\le\alpha p_{n,h}+2qD_n^{\mathrm{sd}}+\eta_{n,h}$; division by $p_{n,h}$ proves \MainEqRef{sd:conditional-causal-bound}. For a single report the coverage indicators for the statistical target and ATT differ only on $\mathcal I_n$. Theorem~\MainRef{thm:plugin} bounds the absolute error relative to $1-\alpha_1$ for the statistical target, and $\eta_n$ bounds the additional error for the ATT. For each $h$, the threshold condition follows by dividing by $p_{n,h}\ge\rho_n$; uniformity over labels uses $\max_h\eta_{n,h}$.

Independent regional panels satisfying the conditions of Corollary~\MainRef{cor:regional} satisfy Assumption~\MainRef{ass:array} by its proof. Cohort means and overlapping comparisons are deterministic transformations of the vector containing each region--period mean once, so their reuse changes only polynomial and contrast coefficients. For independent repeated cross-sections, apply the same sampling proof with each region--period cell as a scalar block; the lower cell-variance bound, uniform $2+\delta$ moments, and diverging minimum cell size give the joint normal approximation in \MainEqRef{eq:uniform-clt} and covariance consistency in \MainEqRef{eq:uniform-cov}. The causal interpretation uses the regional populations specified by the cohort effects.
\end{proof}

The finite-sample experiment for joint and combined reporting events appears in Section~\SuppRef{supp:staggered-simulation}.

\subsection{Joint orthogonality and the cross-cohort example}\label{supp:sd-restored}
\begin{proof}[Proof of Proposition~\MainRef{sd:orthogonality}]
Joint normality makes $c'Y$ independent of $BY$ exactly when $B\Sigma c=0$, including singular $\Var(BY)$. Put $b=\frac{\Sigma c}{c'\Sigma c}$ and let $r$ be the observed residual $Y-bc'Y$. Since $Bb=0$, $B(r+bt)=Br$ for every trial contrast value $t$. A label determined by $BY$ is therefore constant along the conditional line. The conditional fiber for the observed label is $\R$, giving the conventional normal pivot and interval.

In the example let $U_k=(Y_{L,3}-Y_{L,2})-(Y_{k,3}-Y_{k,2})$, $k=A,B$. The discrepancy vector $(U_A,U_B)$ has variances four and covariance two, and the selection scores are $\frac{U_k^2}{2}$. Hence $D=\frac{U_A-U_B}{2}$ and $H=\frac{U_A+U_B}{\sqrt{12}}$ are independent standard normals, and $S_L=A$ iff $DH\le0$, apart from a null tie. Direct covariance calculation gives $\Cov(W,D)=\frac{\sqrt3}{4}$, $\Cov(W,H)=\frac{1}{4}$, and the normal projection
\[
 W=(\frac{\sqrt3}{4})D+(\frac{1}{4})H+\xi,\qquad
 \xi\perp(D,H),\quad \xi\sim N(0,\frac{3}{4}).
\]
Each sign event has probability $\frac{1}{2}$. On either, the conditional second moments of $D,H$ are one. Their conditional product moment is $-\frac{2}{\pi}$ on $DH\le0$ and $\frac{2}{\pi}$ on $DH>0$. Squaring the projection yields $\E(W^2\mid DH\le0)=1-\frac{\sqrt3}{4\pi}$ and the opposite sign gives $1+\frac{\sqrt3}{4\pi}$, proving \MainEqRef{sd:counterexample-moments}. The unstandardized early contrast is marginally normal with variance three, so $W$ has marginal variance one.
\end{proof}

\FloatBarrier
\renewcommand{\theequation}{B\arabic{equation}}
\renewcommand{\theHequation}{B\arabic{equation}}
\setcounter{equation}{0}
\renewcommand{\thetheorem}{B\arabic{theorem}}
\renewcommand{\theHtheorem}{B\arabic{theorem}}
\setcounter{theorem}{0}
\section{Expected variances under general normal outcomes}
\label{v58:expected-var-general-section}\label{v57:normal-characterization-section}
Let $X$ collect the outcome information used for selection, before applying
the possibly nonlinear fitting rule. Let $S=S(X)$ choose a control from a fixed candidate pool, with
$p_s=\Pr(S=s)$. Each fixed-comparison error $U_s$ has mean zero and variance
$\sigma_s^2>0$. Independently of the data, draw $\widetilde S$ with
$\Pr(\widetilde S=s)=p_s$. Define
\begin{equation}\label{v58:expected-variance-ratio-definition}
 \mathcal R_E=
 \frac{\E_S[\Var(U_S\mid S)]}
      {\E_{\widetilde S}[\Var(U_{\widetilde S}\mid\widetilde S)]}
 =\frac{\sum_s p_s\Var(U_s\mid S=s)}{B},\qquad
 B=\sum_s p_s\sigma_s^2.
\end{equation}
The numerator averages uncertainty conditional on the reported control.
In probability-weighted sums, terms with $p_s=0$ are defined as zero.
For the denominator, we draw controls independently of the outcome realization
using the same selection probabilities. We hold the candidate pool fixed
throughout this comparison. Section~\SuppRef{disp:comparison-distribution} distinguishes this benchmark from reported-error variance and mean squared error.

\begin{theorem}[Expected variances under normal outcomes]
\label{v58:expected-variance-normal}\label{v57:normal-curvature}
Suppose $(X,U_1,\ldots,U_J)$ is jointly normal,
where $J$ is the candidate count and $X\sim N_k(\mu,\Omega)$ has dimension $k$ and $\Omega\succ0$. Let
$a_s=\Cov(X,U_s)$. The measurable regions
$\mathcal A_s=\{x:S(x)=s\}$ form a partition fixed before sampling.
Write $p_s(\mu)=\Pr_\mu(X\in\mathcal A_s)$ and let derivatives vary
$\mu$ while holding the regions, $\Omega$, $a_s$ and $\sigma_s^2$ fixed,
with each $U_s$ centered throughout this location family. Write $\nabla$ and $\nabla^2$ for the gradient and Hessian with respect to $\mu$. Then
\begin{equation}\label{v58:expected-variance-normal-function}
 \begin{aligned}
 \mathcal R_E
 &=1+\frac{1}{B}\sum_{s:p_s>0}
 \left\{a_s'\nabla^2p_s(\mu)a_s
       -\frac{[a_s'\nabla p_s(\mu)]^2}{p_s(\mu)}\right\}\\
 &=1+\frac{1}{B}\sum_{s:p_s>0}
        p_s a_s'\nabla^2\log p_s(\mu)a_s.
 \end{aligned}
\end{equation}
Thus, for any $q>0$, $\mathcal R_E>q$ if and only if the sum in braces
exceeds $(q-1)B$. This is a functional of the population outcome means,
their covariance, their covariances with the ATT errors, the fixed-error
variances, and the complete selection regions.

Let $\nu_s=\E(U_s\mid S=s)$ for $p_s>0$, set $\nu_s=0$ for $p_s=0$, and
write $\bar\nu=\sum_s p_s\nu_s$. The unconditional reported-error variance
satisfies
\begin{equation}\label{v58:variance-ratio-distinction}
 \frac{\Var(U_S)}{B}
 =\mathcal R_E+
   \frac{\sum_s p_s(\nu_s-\bar\nu)^2}{B}.
\end{equation}
For $p_s(\mu)>0$, the first two identities give the candidate-conditional moments. The last two sum over all labels:
\begin{align}
 \nu_s&=a_s'\nabla\log p_s(\mu),\label{v57:normal-condmean}\\
 \Var(U_s\mid S=s)&=\sigma_s^2+a_s'\nabla^2\log p_s(\mu)a_s,\label{v57:normal-condvar}\\
 \E U_S&=\sum_s a_s'\nabla p_s(\mu),\label{v57:normal-selectedmean}\\
 \Var(U_S)&=B+\sum_s a_s'\nabla^2p_s(\mu)a_s
             -\left\{\sum_s a_s'\nabla p_s(\mu)\right\}^2.
             \label{v57:normal-selectedvar}
\end{align}
The two ratios coincide if and only if all positive-probability
selected-label means are equal. In particular, they coincide if
$\mu=0$ and each selection region is symmetric about the origin ($\mathcal A_s=-\mathcal A_s$). Then every $\nu_s=0$. If every $a_s=0$, $\mathcal R_E=1$.
\end{theorem}
\begin{proof}
Differentiating the normal density under the integral gives
\[
 \nabla p_s=\Omega^{-1}\E[(X-\mu)\mathbf1_{\mathcal A_s}],\qquad
 \nabla^2p_s=\Omega^{-1}
 \E[(X-\mu)(X-\mu)'\mathbf1_{\mathcal A_s}]\Omega^{-1}
             -p_s\Omega^{-1}.
\]
All finite-order derivatives exist: locally in $\mu$, each density
derivative is dominated by an integrable polynomial times a normal
density with larger covariance. Normal projection writes
$U_s=a_s'\Omega^{-1}(X-\mu)+\eta_s$, with $\eta_s$ independent of
$X$, mean zero, and variance $\sigma_s^2-a_s'\Omega^{-1}a_s$.
Therefore
\[
 \E[U_s\mathbf1_{\mathcal A_s}]=a_s'\nabla p_s,\qquad
 \E[U_s^2\mathbf1_{\mathcal A_s}]
          =p_s\sigma_s^2+a_s'\nabla^2p_s a_s.
\]
Subtracting $p_s\nu_s^2=\frac{[a_s'\nabla p_s]^2}{p_s}$ from the second
quantity and summing proves the first functional. The Hessian identity
for $\log p_s$ proves the second. A zero-probability normal region is
Lebesgue null and has identically zero probability and derivatives, so
omitting those terms is legitimate. Equation
\eqref{v58:variance-ratio-distinction} is the law of total variance.
For origin-symmetric regions, $p_s(-\mu)=p_s(\mu)$ and $\nabla p_s(0)=0$. The conditional means therefore vanish. The
last assertion follows directly from the functional.
\end{proof}
\renewcommand{\theequation}{C\arabic{equation}}
\renewcommand{\theHequation}{C\arabic{equation}}
\setcounter{equation}{0}
\section{Selection on outcome levels and paths}\label{levelv62:oa}
\subsection{Selection event and interval computation}\label{levelv62:fiber}
Use the primitive outcome vector, gap maps $D_j$, and projections $P_{m_P},H_{m_P}$ defined in Section~\ref{levelv62:main}. Set $G_\lambda=P_{m_P}+\lambda H_{m_P}$ and $K_j^{(\lambda)}=D_j'G_\lambda D_j$. The criterion is $Q_j^{(\lambda)}(y)=y'K_j^{(\lambda)}y$; $G_0$ has rank $m_P-1$, whereas $G_\lambda$ has rank $m_P$ for $\lambda>0$. Away from ties,
\begin{equation}\label{levelv62:oa-event}
 E_{s,\lambda}=\bigcap_{j\ne s}\{y:y'A_{js}^{(\lambda)}y\ge0\},\qquad
 A_{js}^{(\lambda)}=K_j^{(\lambda)}-K_s^{(\lambda)}.
\end{equation}
Recorded priorities resolve exact ties, and identically zero comparisons are removed before locating boundaries. The finite pool therefore satisfies (S) with degree at most two, including $\lambda=0$.
\begin{proof}[Proof of Corollary~\ref{levelv62:corollary}]
For $\Sigma\succ0$ and $c_s'\Sigma c_s>0$, Lemma~\ref{lem:gaussian} applied to \eqref{levelv62:oa-event} gives exact coverage conditional on any positive-probability label and almost every selected residual. Deterministic sequences $\lambda^{(n)}$ preserve the bounded degree and comparison count. Under Assumption~\ref{ass:array}, Theorem~\ref{thm:plugin} gives uniform joint pivot error $D_n\to0$, conditional coverage error at most $2D_n/P(E_{n,s,\lambda^{(n)}})$, and uniform marginal coverage after summing over labels. Conditional error vanishes for probabilities bounded below, or whenever this ratio vanishes. Theorem~\ref{rev:endpoints} gives uniform endpoint equivalence in standard-error units and a length ratio tending to one between feasible and reference selective intervals, each using its own covariance, direction, residual and complete set. Equation~\eqref{sd:single-decomposition} makes the target the common ATT when every selectable candidate has $\Delta_{n,s}=0$.
\end{proof}
For observed residual $r$ and trial contrast value $t$, substitute the conditional line $r+b_st$ from Section~\ref{sec:pivot}:
\begin{equation}\label{levelv62:quadratic-coefficients}
 \begin{gathered}
 (r+b_st)'A_{js}^{(\lambda)}(r+b_st)=\alpha_{js}t^2+\beta_{js}t+\gamma_{js},\\
 \alpha_{js}=b_s'A_{js}^{(\lambda)}b_s,\quad
 \beta_{js}=2r'A_{js}^{(\lambda)}b_s,\quad
 \gamma_{js}=r'A_{js}^{(\lambda)}r.
 \end{gathered}
\end{equation}
Nonconstant comparisons give at most $2(J-1)$ distinct roots and $2J-1$ open intervals. Use the recorded priority for ties, and retain every accepted component. Off the polynomial boundaries, the observation is interior. For a displacement of $u$ standard errors, write
\begin{equation}\label{levelv62:displacement}
 y(u)=y+\frac{\Sigma c_s}{\sigma_s}u,\quad
 \sigma_s^2=c_s'\Sigma c_s,\quad c_s'y(u)=T_s+\sigma_su.
\end{equation}
Invert \eqref{eq:pivot} on the complete accepted union, fixing covariance and the union while varying the trial mean. For affine discrepancies $B_jy-a_j$ and positive-semidefinite weights $W_j$, put $h_j=B_jr-a_j$, $k_j=B_jb_s$. The corresponding coefficients are
\[
 \begin{gathered}
 \alpha_{js}=k_j'W_jk_j-k_s'W_sk_s,\quad
 \gamma_{js}=h_j'W_jh_j-h_s'W_sh_s,\\
 \beta_{js}=2(h_j'W_jk_j-h_s'W_sk_s).
 \end{gathered}
\]
\subsection{Dependence removed by centering}\label{levelv62:channel-proof}
\begin{proof}[Proof of Proposition~\ref{levelv62:channel}]
Equation~\eqref{levelv62:exchangeable} gives
\[
 \Cov(e_{jP},\bar e_{kP})=(a_{jk}/m_P+b_{jk})\mathbf1_{m_P},
\]
and hence
\[
 \begin{gathered}
 \Cov(C_j,U_k)=P_{m_P}(\gamma_{jk}-a_{jk}/m_P-b_{jk})\mathbf1_{m_P}=0,\\
 \Cov(\bar e_{jP},U_k)=\gamma_{jk}-a_{jk}/m_P-b_{jk}.
 \end{gathered}
\]
Stacking over all $j,k$ and joint normality give independence between all centered gaps and all DiD errors, including singular centered-gap covariance. Deterministic population gaps preserve independence, so centered-gap selection preserves each fixed-comparison error law. If $B$ stacks the centered-gap maps, $B\Sigma c_s=0$ gives $B(r+b_st)=Br$ for every $t$. The selector is constant along that line, its full accepted set is $\R$, and its selective interval equals the conventional interval. Adding level maps can break this equality through the second covariance.
\end{proof}
For $e_{jt}=A_j+\varepsilon_{jt}$ and $e_{jA}=A_j+\bar\varepsilon_{jA}$, with independent persistent and transitory components, transitory errors uncorrelated across dates, and common pre-treatment same-date covariances $\Cov(\varepsilon_{jt},\varepsilon_{kt})=a_{jk}$, $b_{jk}=\Cov(A_j,A_k)$ and $\gamma_{jk}=b_{jk}$, yielding
\begin{equation}\label{levelv62:persistent}
 \Cov(\bar e_{jP},U_k)=-a_{jk}/m_P.
\end{equation}
Thus positive temporal correlation can coexist with centered-fit independence, while the transitory pre-treatment average remains in level fit. Without exchangeability, use \eqref{levelv62:general-channel} and the complete selection set.
\renewcommand{\theequation}{D\arabic{equation}}
\renewcommand{\theHequation}{D\arabic{equation}}
\setcounter{equation}{0}
\renewcommand{\thetheorem}{D\arabic{theorem}}
\renewcommand{\theHtheorem}{D\arabic{theorem}}
\setcounter{theorem}{0}
\section{Loss representations and joint covariance implementation}\label{supp:reduction}
\subsection{Polynomial atoms and structural ties}\label{poly:piecewise-supp}
Use the nonzero polynomials defining loss regions and within-region score differences. Evaluate a loss contribution directly if its discrepancy lies identically on a piece boundary. Resolve identically tied scores by priority and evaluate the original loss on boundaries. Fixed counts of candidates, discrepancies, and pieces bound the number $M_0$ of polynomials of degree at most $D$.

For L1, signs $\epsilon_{jt}\in\{-1,1\}$ use guards $\epsilon_{jt}r_{jt}>0$ and score $Q_{j,\epsilon}=\sum_t\omega_t\epsilon_{jt}r_{jt}$. Under priority $1<\cdots<J$, candidate $s$ requires $Q_j-Q_s>0$ for $j<s$ and $\ge0$ for $j>s$. An identically zero comparison empties a strict branch and leaves a weak branch unrestricted. Huber pieces at $r=-\kappa_H,\kappa_H$ are $-\kappa_Hr-\frac{\kappa_H^2}{2}$, $\frac{r^2}{2}$, and $\kappa_Hr-\frac{\kappa_H^2}{2}$; their comparisons have degree at most two. The boundary event in \MainEqRef{poly:boundary} records zeros of all nonzero defining polynomials.

\subsection{Full-fiber inversion and recorded inputs}\label{supp:full-fiber-inversion}
Apply the exceptions in Section~\MainRef{sec:asymptotics} first. Otherwise, compute $\widehat b_s,\widehat R_s$ with the estimated covariance as in \MainEqRef{eq:decomp}. Each defining polynomial is nonzero at the observed contrast. Along the conditional line, find its real roots and retain all intervals giving label $s$, including unbounded intervals. With no roots, evaluate the rule on the entire line. At the observed contrast, invert \MainEqRef{eq:pivot} over their union at $1-\alpha/2$ for the lower endpoint and $\alpha/2$ for the upper endpoint.

Pointwise equivalent selectors give the same interval at a common covariance and contrast outside both representations' exceptions.

\subsection{Exact invariance under linear reduction}\label{supp:linear-reduction}
\begin{proposition}[A known linear reduction]\label{rev:linear-reduction}
Let $D\in\R^{q\times d}$ be fixed and full row rank. Suppose $S(y)=S^*(Dy)$, with corresponding atoms $p_k(y)=p_k^*(Dy)$ and the same sign-to-label rule, and $c_s=D'\gamma_s\ne0$. For every $W\succ0$, the nonrandomized full-fiber pivot and interval computed from $(Y,W,c_s)$ equal those computed from $(DY,DWD',\gamma_s)$. The identity holds when the same $D$ transforms an estimated positive definite covariance.
\end{proposition}
\begin{proof}
Write $v=c_s'Wc_s$, $b=\frac{Wc_s}{v}$, and $R=Y-bc_s'Y$. Reduction gives $v^*=v$, $b^*=Db$, and $R^*=DR$, hence $R^*+b^*t=D(R+bt)$. The contrast, variance, complete fiber and atom-boundary exceptions coincide, while full row rank makes $DWD'$ positive definite. Thus the pivots and their inversions agree identically.
\end{proof}

\subsection{Independent influence units and the joint covariance}\label{supp:influence-units}
The independent-regional special case is verified in Section~\MainRef{supp:proof-regional}; the influence-unit representation below allows the corresponding joint covariance calculations.
The independent sampling units may be individuals or clusters.

\begin{proposition}[Independent influence-unit implementation]\label{rev:influence-units}
With known deterministic nonsingular $L_n$ as in
Assumption~\MainRef{ass:array}, suppose in fixed dimension
\[
 L_n^{-1}(Y_n-\mu_n)=\sum_{g=1}^{G_n}\xi_{ng}+e_n,
 \qquad
 V_n=\sum_g\E(\xi_{ng}\xi_{ng}'),
\]
where $\xi_{ng}$ is the independent centered contribution of unit $g$ on the standardized scale, $\widehat\xi_{ng}$ estimates it, and $e_n$ is the remainder. With $\|\cdot\|$ the Euclidean norm, assume uniformly over the distribution class that $\underline v I_d\preceq V_n\preceq\overline v I_d$, $e_n=o_P(1)$, and
\[
 \sum_g\E\|\xi_{ng}\|^4\longrightarrow0,
 \qquad
 \sum_g\|\widehat\xi_{ng}-\xi_{ng}\|^2=o_P(1).
\]
Then Assumption~\MainRef{ass:array} holds with $\widehat V_n=\sum_g\widehat\xi_{ng}\widehat\xi_{ng}'$. For a recorded selector satisfying (S), use $L_n\widehat V_nL_n'$ to recompute the variance, direction, residual, and complete selection set for Theorem~\MainRef{thm:plugin}.
\end{proposition}

\begin{proof}
The uniform fourth-moment condition gives the Lindeberg condition, since $\sum_g\E[\|\xi_{ng}\|^2\ind\{\|\xi_{ng}\|>\epsilon\}]\le\epsilon^{-2}\sum_g\E\|\xi_{ng}\|^4$. Covariance compactness and the multivariate CLT give uniform bounded-Lipschitz convergence by a subsequence argument; $e_n=o_P(1)$ preserves it. Write $\|A\|_F^2=\sum_{k,l}A_{kl}^2$ for the squared Frobenius norm. Independence gives
\[
 \E\left\|\sum_g\{\xi_{ng}\xi_{ng}'-\E(\xi_{ng}\xi_{ng}')\}\right\|_F^2
 \le\sum_g\E\|\xi_{ng}\|^4\to0.
\]
Since $\sum_g\|\xi_{ng}\|^2=O_P(1)$, set $a_{ng}=\widehat\xi_{ng}-\xi_{ng}$. The outer-product estimation difference is bounded uniformly by $2(\sum_g\|\xi_{ng}\|^2)^{\frac{1}{2}}(\sum_g\|a_{ng}\|^2)^{\frac{1}{2}}+\sum_g\|a_{ng}\|^2=o_P(1)$.
\end{proof}

For ratio estimators, retain numerator, denominator, and normalization terms in the influence contributions. Sum contributions sharing an independent cluster before taking outer products. For dependence across units, estimate a joint covariance justified under Assumption~\MainRef{ass:array} and recompute the selective reference law.
%

%

\clearpage
\pagenumbering{arabic}
\setcounter{section}{0}
\setcounter{subsection}{0}
\setcounter{subsubsection}{0}
\setcounter{equation}{0}
\setcounter{theorem}{0}
\setcounter{assumption}{0}
\setcounter{remark}{0}
\setcounter{figure}{0}
\setcounter{table}{0}
\setcounter{footnote}{0}
\pdfbookmark[0]{Online Appendix}{online-appendix}
\renewcommand{\thesection}{S\arabic{section}}
\renewcommand{\thesubsection}{\thesection.\arabic{subsection}}
\renewcommand{\theequation}{S\arabic{equation}}
\renewcommand{\thetheorem}{S\arabic{theorem}}
\renewcommand{\theassumption}{S\arabic{assumption}}
\renewcommand{\theremark}{S\arabic{remark}}
\renewcommand{\thetable}{S\arabic{table}}
\renewcommand{\thefigure}{S\arabic{figure}}
\renewcommand{\theHsection}{S\arabic{section}}

\begin{center}
{\normalsize\bfseries ONLINE APPENDIX TO\par}
{\normalsize\bfseries Valid Tests and Confidence Intervals in Difference-in-Difference\par}
\medskip Takahiro Hoshino\qquad Ryoya Nakano
\end{center}
\noindent This appendix gives the data constructions and proofs behind the policy reanalyses and simulations. Section~\ref{oa99:extensions} quantifies the information gained by pooling paths and gives a finite-sample calibration. Appendix~\MainRef{app:proofs} of the main paper proves feasible inference and the joint-report results.
\par\medskip

\section{Diagnostics for interval length and selection dependence}
\label{diagv63:oa}\label{supp:length}
We use the complete conditional selection set to explain interval
length, determine when researchers recover conventional inference, and
calculate the two-candidate variance ratio with correlated fit
information. Section~\ref{disp:oa} extends the normal-panel construction to multiple candidates, and Section~\ref{amp:proofs} relates amplification to candidate dependence. Researchers can use these diagnostics alongside the
feasible-inference guarantees in main Theorems~\MainRef{thm:plugin}
and~\MainRef{rev:endpoints}.

\subsection{Distance to a truncation boundary and interval length}
\label{diagv63:boundary}
For a fixed interval and reference variance, the boundary limit below
explains main \MainEqRef{diagv63:boundary-main}; see also
\citet[Lemma A.4 and Appendices B and D]{oa-KL2021}.

\begin{lemma}[A boundary-distance diagnostic]\label{diagv63:boundary-lemma}
Fix $\sigma>0$ and $-\infty\le a<b<\infty$. Let $F_\theta$ be the cumulative distribution function of
$N(\theta,\sigma^2)$ restricted to $(a,b)$. For $t=b-d$, $d\downarrow0$,
and fixed $r\in(0,1)$, let $\theta_r(t)$ solve $F_{\theta_r(t)}(t)=r$.
Then
\begin{equation}\label{diagv63:boundary-endpoint}
 \theta_r(b-d)=b-\frac{\sigma^2\log r}{d}+O(d).
\end{equation}
For fixed $\alpha\in(0,1)$, the equal-tail interval has the following
length ratio relative to the conventional normal interval:
\begin{equation}\label{diagv63:boundary-width}
 \frac{\theta_{\frac{\alpha}{2}}(t)-\theta_{1-\frac{\alpha}{2}}(t)}
      {2z_{1-\frac{\alpha}{2}}\sigma}
 =\frac{\log\{\frac{2-\alpha}{\alpha}\}}
        {2z_{1-\frac{\alpha}{2}}}\frac1{\delta_{\partial}}
        +O(\delta_{\partial}),
 \qquad \delta_{\partial}=\frac{d}{\sigma}.
\end{equation}
For $\alpha=.05$, the coefficient is $.9345992261$.
\end{lemma}
\begin{proof}
Put $h=(\theta-b)d/\sigma^2$, $\delta=d/\sigma$ and $v=(b-x)/d$.
The conditional CDF is the ratio of the integrals of
$e^{-hv-\delta^2v^2/2}$ over $[1,(b-a)/d]$ and $[0,(b-a)/d]$.
Uniformly for $h$ in a positive compact neighborhood of $-\log r$,
replacing the upper limit by infinity incurs $O(e^{-c/d})$ when $a$ is
finite, and $|e^{-\delta^2v^2/2}-1|\le\delta^2v^2/2$ makes the ratio
$e^{-h}+O(\delta^2)$. Monotonicity brackets its inverse in this
neighborhood, and the nonzero derivative $-r$ gives
$h=-\log r+O(\delta^2)$. Substitution gives \eqref{diagv63:boundary-endpoint}, and subtracting the two endpoint expansions gives \eqref{diagv63:boundary-width}.
\end{proof}

For level-inclusive selection under the primary TennCare historical-gap model,
the full standardized selection set, centered at the observed contrast, is
$[-4.9520816243,.009475057381]$. Equation~\eqref{diagv63:boundary-width}
gives a leading ratio of 98.6378, compared with 98.7234 by inversion.
The expansion describes the limit as the observation approaches the
boundary. Its remainder need not bound the finite-sample error at every
distance. We verified the full selection set and both confidence limits
from the primitive means, covariance, and candidate maps using 100-digit
arithmetic. At the upper confidence limit, the base-ten log probability
of the accepted set is $-32{,}915.67$ for the trial mean and fixed residual.
This small probability requires stable tail integration. The reporting
probability in main Theorem~\MainRef{thm:plugin} averages over samples.
For a disconnected set, the finite-union result of \citet{oa-KL2021} uses
the outermost finite boundary. We retain both components in the inversion
for centered TennCare, whose nearest boundary is distant.

\subsection{When conventional intervals are recovered}
\label{diagv63:full-line}\label{rev:no-adjustment}
For an event defined by $p_k(y)=y'A_ky+d_k'y+e_k\ge0$, $A_k=A_k'$, the
sufficient invariance condition
\begin{equation}\label{rev:eq:invariance}
 A_kb_s=0,\qquad d_k'b_s=0\quad\text{for every }k
\end{equation}
gives $p_k(r+b_st)=p_k(r)$ for all $t$. Every selected residual therefore
has full line $\R$, and the normal selective and conventional intervals
coincide. For the canonical regional gap maps in main
\MainEqRef{eq:selection}, with at least two candidates,
$A_{j,s}\Sigma c_s=0$ for every $j\ne s$ is equivalent to
$B_j\Sigma c_s=0$ for every candidate: writing
$\zeta_j=B_j\Sigma c_s$, each $\zeta_j$ is centered across pre-treatment periods; the separate region-$j$ and region-$s$ blocks of
$B_j'\zeta_j-B_s'\zeta_s=0$ force $\zeta_j=\zeta_s=0$; the converse is immediate.
For general fixed maps, cancellation between scores can also yield \eqref{rev:eq:invariance}.

For fixed covariance and residual, write each remaining comparison as
$q_{j,s}(u)=\alpha_{j,s}u^2+\beta_{j,s}u+\gamma_{j,s}\ge0$,
resolving identically tied comparisons by the recorded priority first.
\begin{proposition}[Conventional inference on the conditional line]
\label{diagv63:certificate}
The complete selection set equals $\R$ up to a Lebesgue-null set
if and only if every comparison satisfies
\begin{equation}\label{diagv63:certificate-positive}
 \big\{\alpha_{j,s}>0,\ \beta_{j,s}^2-4\alpha_{j,s}\gamma_{j,s}\le0\big\}
 \quad\text{or}\quad
 \big\{\alpha_{j,s}=\beta_{j,s}=0,\ \gamma_{j,s}\ge0\big\}.
\end{equation}
Its normal reference law and interval then coincide with their
conventional counterparts.
\end{proposition}
\begin{proof}
These are precisely the quadratics nonnegative everywhere. Every other
quadratic or nonconstant linear comparison excludes an open interval,
which the intersection also excludes. When all comparisons are
nonnegative everywhere, priority can remove only finitely many isolated
tie points, of zero normal mass; the restricted and unrestricted normal
laws, pivots and intervals therefore agree.
\end{proof}
The quadratic-root calculation verifies \eqref{diagv63:certificate-positive} for the
state-clustered TVA rules with $\lambda=.5,1$, establishing equality of
the complete conditional law.

\subsection{Two candidates with correlated fit information}
\label{diagv63:correlated}
We extend main Proposition~\MainRef{amp:two-main} to account for a
common source of candidate dependence: researchers compare both
controls with the same treated path. For a symmetric model, we derive
a closed-form variance ratio that complements the general covariance
characterization in Theorem~\MainRef{v58:expected-variance-normal}.

\begin{proposition}[Own and competing fit with correlated candidates]
\label{diagv63:correlated-prop}
Let $(Z_1,Z_2)$ be centered bivariate normal with unit variances and
correlation $r\in(-1,1)$. Let $e_1,e_2$ be independent standard normals,
independent of $(Z_1,Z_2)$. For $q=a^2+b^2+2rab<1$, put
\begin{align}
 U_1&=aZ_1+bZ_2+\sqrt{1-q}\,e_1,\notag\\
 U_2&=bZ_1+aZ_2+\sqrt{1-q}\,e_2,
 \qquad S=\arg\min_{j=1,2}Z_j^2.
 \label{diagv63:correlated-model}
\end{align}
Each fixed error has mean zero and variance one. Conditional on either
label, the selected error has mean zero and variance
\begin{equation}\label{diagv63:correlated-ratio}
 \mathcal R_E=\Var(U_S)
 =1+\frac2\pi\sqrt{1-r^2}(b^2-a^2).
\end{equation}
Hence the variance increases if and only if $|b|>|a|$.
If $a=\pm b$, the error conditional on either label is standard normal.
If $a=0$ and $0<|b|<1$, the conventional normal interval strictly
undercovers at every nominal level $1-\alpha\in(0,1)$.
\end{proposition}
\begin{proof}
For independent standard normals $X,Y$, write
\[
 Z_1=\sqrt{\frac{1+r}{2}}\,X+\sqrt{\frac{1-r}{2}}\,Y,
 \qquad
 Z_2=\sqrt{\frac{1+r}{2}}\,X-\sqrt{\frac{1-r}{2}}\,Y.
\]
Since $Z_1^2-Z_2^2=2\sqrt{1-r^2}XY$, the event $S=1$ is $XY\le0$
up to a null set and has probability one half. Independent normal signs
give
\[
 \E(X^2\mid S=1)=\E(Y^2\mid S=1)=1,
 \qquad \E(XY\mid S=1)=-\frac{2}{\pi}.
\]
For $k_r=(\frac{2}{\pi})\sqrt{1-r^2}$, it follows that
\[
 \E(Z_1^2\mid S=1)=1-k_r,\quad
 \E(Z_2^2\mid S=1)=1+k_r,\quad
 \E(Z_1Z_2\mid S=1)=r.
\]
Simultaneous sign reversal preserves the event and reverses each $Z_j$,
so the conditional first moments are zero. Independence of $e_1$ gives
\[
 \Var(U_1\mid S=1)
 =a^2(1-k_r)+b^2(1+k_r)+2abr+1-q
 =1+k_r(b^2-a^2).
\]
Exchange symmetry gives the same result for label two. Both conditional
means vanish, so averaging yields \eqref{diagv63:correlated-ratio}.

If $a=b$, only the sum coordinate $X$ and independent noise enter each
error; if $a=-b$, only the difference coordinate $Y$ and independent
noise enter. Conditioning on $XY\le0$ or $XY\ge0$ leaves the marginal
normal law of either $X$ or $Y$ unchanged, proving the standard-normal
claims after the unit-variance normalization.

For $a=0$, let $M=\max(|Z_1|,|Z_2|)$ and
$c=\Phi^{-1}(1-\frac{\alpha}{2})$. The conventional coverage is $\E h(M)$,
where $h(x)=P\{|bx+\sqrt{1-b^2}e|\le c\}$ for $x\ge0$.
Normal symmetry makes this expression depend only on the magnitude of
the competing discrepancy. The function $h$ strictly decreases for
$x>0$, while $M\ge|Z_1|$ with strict inequality on a positive-probability
set. Hence $\E h(M)<\E h(|Z_1|)=1-\alpha$. As $|b|\uparrow1$,
coverage tends to $P(|Z_1|\le c,|Z_2|\le c)$, which equals
$(1-\alpha)^2$ when $r=0$.
\end{proof}

\paragraph{Three-period difference-in-differences realization.}\label{diagv70:two-candidate-did}
Let $B_0,B_1,B_2,Z_0,F,e_1,e_2$ be independent standard normal variables,
independent of the bivariate fit shocks $(Z_1,Z_2)$ in the proposition.
Put $W_0=Z_0$, $W_s=Z_0-Z_s$ for $s=1,2$, and define
\[
 Y_{j,1}=B_j-\frac{W_j}{2},\quad Y_{j,2}=B_j+\frac{W_j}{2},\qquad
 Y_{0,3}=B_0+F+\tau,\quad Y_{s,3}=B_s+F-U_s.
\]
All untreated means are zero and treatment adds $\tau$. Centered squared
pre-treatment discrepancies are $\frac{Z_s^2}{2}$, and the DiD estimate using
the mean of periods 1 and 2 has error $U_s$. The nine observed means
recover $B_j,W_j,F,U_s$, then $Z_0,Z_s,e_s$. This linear transformation
is invertible because $1-q>0$. Its input covariance is positive
definite for $|r|<1$, so the full outcome covariance is positive definite.
Both controls satisfy parallel trends, selection uses pre-treatment
observations alone, and temporal dependence is retained. At $r=0$,
this construction and the proposition give all claims of main
Proposition~\MainRef{amp:two-main}, including the coverage limit
$(1-\alpha)^2$.

\paragraph{Sharp range within the two-candidate model.}
For every fixed $r\in(-1,1)$, the attainable variance ratios in
Proposition~\ref{diagv63:correlated-prop} form the open interval
\[
 1-\frac2\pi<\mathcal R_E<1+\frac2\pi.
\]
To see this, write $d_r=\sqrt{1-r^2}$. Since $q\ge0$ and
\[
 q^2-d_r^2(b^2-a^2)^2
 =\{2ab+r(a^2+b^2)\}^2\ge0,
\]
we have $d_r|b^2-a^2|\le q<1$. Both boundaries are sharp while
the full covariance remains positive definite. For
$\varepsilon\in(0,1)$, choose
\[
 b_\varepsilon^2
   =(1-\varepsilon)\frac{1+d_r}{2d_r^2},\qquad
 a_\varepsilon=-\frac{r}{1+d_r}b_\varepsilon.
\]
Then $q=1-\varepsilon$,
$d_r(b_\varepsilon^2-a_\varepsilon^2)=1-\varepsilon$, and
$\mathcal R_E=1+(\frac{2}{\pi})(1-\varepsilon)$; exchanging the coefficients
gives the lower half of the range. Along these sequences,
$\operatorname{Cov}((U_1,U_2)'\mid Z_1,Z_2)=\varepsilon I_2\succ0$.
The choice $a=b=0$ gives the midpoint of the range. This range
concerns the symmetric normal two-candidate model; the realized
interval-length ratio in Section~\ref{diagv63:boundary} is a
different quantity.

In model~\eqref{diagv63:correlated-model}, writing $k_o=\Cov(Z_1,U_1)=a+rb$ and
$k_c=\Cov(Z_2,U_1)=b+ra$ gives the equivalent expression
\begin{equation}\label{diagv63:correlated-covariances}
 \mathcal R_E-1
 =\frac2\pi\frac{k_c^2-k_o^2}{\sqrt{1-r^2}}.
\end{equation}
When $a,b$ are fixed and $q<1$ is maintained, increasing $|r|$ reduces
the magnitude of the variance change in
\eqref{diagv63:correlated-ratio}. Holding $k_o,k_c$ fixed instead gives
a different comparison in \eqref{diagv63:correlated-covariances}.
Comparative statics with respect to candidate dependence therefore
require specifying whether the projection coefficients or the
fit--error covariances are held fixed.

\section{Data construction and covariance estimation}\label{app:empirical-details}
\subsection{Sampling units and the joint approximation}\label{app:empirical-repetition-v15}
The police and Washington analyses condition on regional calendar paths and fixed auxiliary inputs $H$, estimate $\operatorname{Var}(Y_n\mid H)$, and target $c_s'\mu_n(H)$ under the article's ATT interpretation. Repeated-path inference additionally requires $\operatorname{Var}\{\E(Y_n\mid H)\}$. For police, independent complete block vectors, uniformly bounded $2+\delta$ moments, uniformly bounded and positive group covariance spectra, and increasing minimum group size give Corollary~\MainRef{cor:regional} after fixed pooling. For Washington, Proposition~\MainRef{rev:influence-units} applies to independent standardized household contributions satisfying its moment conditions; population-normalized denominators bounded away from zero and negligible maximum contributions control the ratio remainder.

\subsection{Police outcomes and covariance}\label{app:police-construction}
The \citet{oa-DS2004} data contain 876 street segments over nine recorded months (7,884 cells), retaining quarter-allocation of corner incidents to their four adjoining segments. Disjoint groups contain 37 protected blocks and 161, 226 and 452 unprotected blocks by distance; the nearest institution may lie outside the sampled neighborhoods. Candidates A, B and C contain all 839 unprotected blocks, the 678 at least two blocks away, and the 452 more than two blocks away. Fit and baseline use April--July, with July restricted to July 1--17; post-treatment outcomes average August--December. Fixed A/B/C priority and scale $.1$ thefts per block per recorded month give Huber thresholds $.025,.05,.1$; every independent-block loss selects A.

For complete group-$g$ vectors, estimate the unbiased covariance $\widehat\Gamma_g$ and pool the disjoint means using their observed group proportions. If $M$ is this fixed invertible nested-pool map, then
\begin{equation}\label{app:overlap}
Y=(M\otimes I_9)\bar X^b,\qquad
\widehat\Sigma=(M\otimes I_9)\operatorname{blockdiag}_g(\widehat\Gamma_g/n_g)(M'\otimes I_9).
\end{equation}
This model permits temporal dependence within blocks and retains covariance between overlapping pools. Quarter-allocation can create adjoining-block dependence.

\paragraph*{Spatial inputs.}
Apply product-triangular kernels to the full group--time block contributions before pooling, with bandwidths 100, 200 and 400 metres in local tangent-plane coordinates. The 23 unresolved blocks retain diagonal contributions and have zero unresolved off-diagonal weights. Each kernel recomputes the full selective line. All 20 cells across five losses have positive-definite pooled inputs; shared-endpoint and two-hop graph inputs are indefinite, and historical geography remains unresolved. Asymptotic inference requires Assumption~\MainRef{ass:array} for the stated dependence model. Table~\ref{tab:police-full-grid-v15} retains the L2 comparisons and the negative 200m L1 selective interval.

\begin{table}[tbp]\centering\submissionTableStyle\setlength{\tabcolsep}{3pt}
\caption{Police: covariance sensitivity and the L1 exception}\label{tab:police-full-grid-v15}
\begin{tabular}{llrllr}\toprule
Covariance & Loss & SE & Conventional 95\% & Selective 95\% & $p_{\rm sel}$\\\midrule
Independent & L2 & $0.0238$ & $[-0.1242,-0.0309]$ & $[-0.1227,0.0427]$ & $0.1810$\\
HAC 100m & L2 & $0.0239$ & $[-0.1244,-0.0307]$ & $[-0.1229,0.0437]$ & $0.1847$\\
HAC 200m & L2 & $0.0244$ & $[-0.1253,-0.0298]$ & $[-0.1239,0.0386]$ & $0.1784$\\
HAC 400m & L2 & $0.0248$ & $[-0.1262,-0.0289]$ & $[-0.1236,0.0119]$ & $0.2544$\\
Independent & L1 & $0.0238$ & $[-0.1242,-0.0309]$ & $[-0.1167,0.2115]$ & $0.7120$\\
HAC 200m & L1 & $0.0244$ & $[-0.1253,-0.0298]$ & $[-0.1183,-0.0205]$ & $0.0046$\\\bottomrule
\end{tabular}
\note{All rows select A, with estimate $-0.07753$ thefts per block per recorded month and loss scale $.1$. The minimum eigenvalues of the four pooled 36-coordinate inputs are $3.169,3.215,2.821,1.763$ times $10^{-6}$. Spatial inputs use current partial geometry. HAC denotes spatial
heteroskedasticity- and autocorrelation-consistent covariance; distances
give the kernel bandwidth. L2 and L1 denote squared-error and absolute-error
loss. SE denotes the standard error of the reported contrast, and
$p_{\rm sel}$ is the two-sided selective $p$-value for a zero ATT.}
\end{table}

\begin{table}[tbp]\centering\submissionTableStyle\setlength{\tabcolsep}{3pt}
\caption{Buenos Aires: restricting eligible controls before ranking}\label{tab:police-eligibility-v97}
\begin{tabular}{lrrll}\toprule
Eligible pools & Report & Estimate & Conventional 95\% & Selective 95\%\\\midrule
A, B, C & A & $-0.07753$ & $[-0.12416,-0.03090]$ & $[-0.12267,0.04266]$\\
B, C & B & $-0.08007$ & $[-0.12690,-0.03324]$ & $[-0.12434,0.07083]$\\
C alone & C & $-0.08080$ & $[-0.12826,-0.03334]$ & $[-0.12826,-0.03334]$\\
\bottomrule\end{tabular}
\note{Unit: car thefts per block per recorded month. A includes all 839 untreated blocks, B the 678 at least two blocks away, and C the 452 more than two blocks away. All rows use the same 37 treated blocks, April through July 17 for fitting, centered squared loss, and the archived complete-block covariance. Di Tella and Schargrodsky (2004, Table 3) use all untreated blocks in A, model one-block proximity in B, and use the farthest blocks as the reference in C. Our reported ranks follow the rule stated here.}
\end{table}

\subsection{Washington employment and household contributions}\label{appv3:cps-full}
Following Online Appendix~C of \citet{oa-CDLZ2019}, use Washington and 39 comparison jurisdictions, including the District of Columbia, in 1996--1998 and 2000--2004. Employment in the original $\$5$--$\$13$ wage bins is divided by adult population, retaining the unnormalized population-share target. CPS denotes the Current Population Survey and QCEW the Quarterly Census of Employment and Wages. Administrative/CPS uses Washington administrative counts and CPS elsewhere; all CPS uses CPS everywhere. Administrative counts, QCEW employment totals and wage-bin mapping are fixed. Jurisdiction-code priority and scale $.01$ in proportion units give Huber thresholds $.25,.5,1$ percentage point.

For $\widehat m_{st}=\widehat N_{st}/\widehat P_{st}$, linearize the quarterly weighted wage-band ratios and annual population denominator jointly at household level. Center contributions within state--year--month cells, apply the finite-cell correction, and sum complete time vectors within each validated linked-household component $k$. With these contributions $u_{sk}$, covariance is $\sum_ku_{sk}u_{sk}'$ within state and zero across jurisdictions. Washington's administrative numerator is fixed; its estimated population still contributes uncertainty. Links use household/person identifiers, rotation, sex, coarse race and age consistency, dropping ambiguous links; the primary age change is $-1$ to 3 and the strict variant requires one. The working model treats linked components and jurisdictions as independent. Strata, primary sampling unit identifiers and replicate weights are unavailable; repeated-path inference also requires the component described in Section~\ref{app:empirical-repetition-v15}.

Table~\ref{tab:cps-covariance-v15} preserves the linked, strict-link and unlinked results. Detailed ratio derivatives, linkage diagnostics and accepted-set coordinates are recorded with the computation files.

\begin{table}[tbp]\centering\submissionTableStyle\setlength{\tabcolsep}{3pt}
\caption{All-CPS Washington: household linkage under L2}\label{tab:cps-covariance-v15}
\begin{tabular}{lrll}\toprule
Covariance & SE & Conventional 95\% & Selective 95\%\\\midrule
Linked & $0.873$ & $[-0.807,2.613]$ & $[-7.234,2.424]$\\
Strict link & $0.872$ & $[-0.805,2.611]$ & $[-7.211,2.423]$\\
Unlinked & $0.777$ & $[-0.620,2.426]$ & $[-0.620,2.426]$\\\bottomrule
\end{tabular}
\note{Units are percentage points in wage-band employment divided by adult population. SE denotes the standard error of the reported contrast; CPS is the Current Population Survey. All rows select Utah with estimate $.903$. Fixed population gives SE $.883$ and selective interval $[-7.554,2.435]$. Administrative Washington with fixed population has structural rank 312/320, so its nondegenerate 320-coordinate calculation is excluded. The complete five-loss covariance grid reports 35 computed cells and five structural exclusions, with no numerical ridge for excluded cells.}
\end{table}

\section{Six policy reanalyses: design and covariance}\label{appv61:cross-details}\label{appv62:details}
\subsection{Recorded designs and reporting events}\label{appv3:nested}
Table~\ref{tab:application-provenance} distinguishes the published comparisons from our candidate menus. We fixed the source-based periods, pools and outcomes. We recorded the five-loss menu after the first squared-loss results and before calculating the other losses, and the level weights $\lambda\in\{0,.25,.5,1\}$ after the centered-fit analyses and before the level-weight calculations. Each length ratio divides the selective interval length by the conventional interval length at 95 percent confidence, using the same comparison, estimate and joint covariance. For fixed rule $r$, we condition on $S_r(Y)=s$; choosing $r$ by the effect additionally requires $A(Y)=r$. We recomputed spatial and linkage sensitivities without choosing the covariance by its effect.

\begin{table}[tbp]
\centering\setlength{\tabcolsep}{4pt}
\caption{Published comparisons and control selection in our reanalyses}
\label{tab:application-provenance}
\begin{tabularx}{\linewidth}{@{}p{.14\linewidth}>{\raggedright\arraybackslash}p{.24\linewidth}>{\raggedright\arraybackslash}p{.21\linewidth}>{\raggedright\arraybackslash}X@{}}
\toprule
Study & Published comparison & Documented choice & Our candidate menu and selected control \\\midrule
Police & \citet[Table~3]{oa-DS2004}: protected and unprotected blocks. & Main reference group is more than two blocks away. & Three nested distance pools. Both fit rules select all unprotected blocks. \\
Washington & \citet[Online Appendix~C]{oa-CDLZ2019}: normalized wage-bin employment. & Pools 39 eligible jurisdictions. & One candidate per jurisdiction. Centered fit selects Utah and level-inclusive fit selects New York. \\
Health centers & \citet[eq.~(1)]{oa-BGB2015}: event study using first-grant timing. & Allows different time effects by urban share. & Five never-grant pools for the 1967 cohort. Both rules select urban share of at least $.5$. \\
TVA & \citet[Tables~III--IV]{oa-KM2014}: manufacturing growth comparisons. & Reports three comparisons using national, Southern, and proposed-authority controls. & Uses the three comparison types with eligibility for our outcome. Both rules select the South. \\
Licensing & \citet[Table~2]{oa-MV2012}: licensed and other subclasses. & Studies enemy-patent availability as a separate exposure. & All unlicensed and two pools by availability. Both rules select available but unlicensed. \\
TennCare & \citet[Section~IV and Table~II]{oa-GGN2014}: state and childless-status contrasts. & Baseline compares Tennessee with other Southern jurisdictions. & All Southern donors or three divisions. Centered fit selects South Atlantic and level-inclusive fit selects all. \\\bottomrule
\end{tabularx}
\note{The source columns describe documented comparisons and design choices. Any earlier unreported search is unknown. We rank each listed menu with centered or level-inclusive squared fit as in main equation~\MainEqRef{appv62:loss}. This section gives the timing of our rules and weights.}
\end{table}

\paragraph*{TVA.} \citet{oa-KM2014}; 1,421 counties (87 TVA, 1,334 non-TVA), 47 states, retaining border exclusion and area-change restriction and positive employment at all five dates. Equal-weight county means of log manufacturing employment use 1900/1910/1920/1930 for fit and baseline, and 1960 afterward. Candidates are all eligible non-TVA, 432 southern non-TVA, and 473 proposed-but-unrealized-authority counties; the latter pools may overlap. Every weight selects South. These dates and unadjusted contrast differ from the original 1940--2000 covariate-adjusted growth estimand.

\paragraph*{Health centers.} \citet{oa-BGB2015}; age-adjusted deaths per 100,000 aged 50+. Use 24 first-grant-1967 counties and 2,440 counties with no grant in 1965--1980, the original New York, Los Angeles and Chicago exclusions, complete 1959--1971 outcomes and positive fixed 1960 age-50+ population weights. Fit: 1959--1966; post mean: 1967--1971. Disjoint controls: urban share zero (826), $(0,.5)$ (1,030), and $[.5,1]$ (584); candidates are all, each stratum and any positive urban share. Population totals set mixture weights. All four weights select the pool with urban share at least $.5$. The 1965 and 1966 cohorts contain only two and three counties. Matching county outcomes, urban shares and weights by county code reproduces the 36 means and their state-clustered covariance.

\paragraph*{Licensing.} \citet{oa-MV2012}; US-inventor patents/subclass-year. Retain 336 licensed and 6,912 unlicensed subclasses, 65 years (471,120 cells). Fixed bins: 1877--1890, 1891--1904, 1905--1918, and 1919--1939; first three form fit/baseline. Availability menu: all unlicensed, 1,057 with a positive subclass maximum of \texttt{count\_cl\_itt}, and the other 5,855. Every weight selects availability. This attribute is not pre-1919; 16 licensed subclasses have zero recorded availability, and treatment remains recorded licensing.

\paragraph*{TennCare.} \citet{oa-GGN2014}; supplement-weighted employment percentage points. Identify households with no member younger than 18 from all members first, then restrict to ages 21--64, education no higher than a bachelor's degree, non-armed-forces status and positive weights. In 17 southern jurisdictions: 101,781 person-years/64,470 household-years (Tennessee, TN: 4,210/2,656). Fit: 2000--2004; omit 2005; post mean: 2006--2007. Four candidates: all 16 donors, East South Central (ESC: AL/KY/MS), West South Central (WSC: AR/LA/OK/TX), and South Atlantic (SA: DE/DC/FL/GA/MD/NC/SC/VA/WV), with fixed equal state weights; state labels use postal abbreviations. Centered fit selects SA; positive weights select all. This childless-adult contrast differs from the original all-adult difference-in-differences and triple-difference specifications.

\paragraph*{Licensing menu chronology.}
The initial menu contains all unlicensed subclasses, parent codes at least 500 (source codes 528--570; 3,146 controls), and lower codes (3,766); centered fit selects all and positive weights select higher codes. The availability menu was recorded after initial-menu results but before its own calculation. At centered fit the initial menu selects all unlicensed subclasses, giving $.25895$ patents per subclass-year; the availability menu selects its available-but-unlicensed subset, giving $.03182$. Each calculation conditions on the ranking within its own menu. Main Table~\MainRef{tab:cross-applications} identifies the exploratory menu.

\subsection{Covariance of overlapping pools and clustered contributions}\label{appv62:covariance}
In disjoint primitive group $g$, use fixed normalized weights $a_i$, mean $\widehat\mu_g=\sum_{i:g(i)=g}a_iX_i$, and embedded contributions $h_i=a_i(X_i-\widehat\mu_g)$. Independent-group covariance uses $\{n_g/(n_g-1)\}\sum_{i:g(i)=g}h_ih_i'$. For cluster sums $H_k=\sum_{i:k(i)=k}h_i$,
\begin{equation}\label{appv62:cluster}
 \widehat W_{\rm cluster}=\frac{K}{K-1}\sum_{k=1}^KH_kH_k'.
\end{equation}
Apply the same fixed pooling map to means and covariance. Remove fixed linear identities on positive-definite primitive coordinates as in Proposition~\MainRef{rev:linear-reduction}; no ridge or eigenvalue clipping is used. Assumption~\MainRef{ass:array} governs the stated sampling sequence.

TVA has 20 coordinates for four overlapping paths at five dates. For the pool indicator $I_{ig}=\ind\{i\text{ belongs to pool }g\}$, county contributions use $h_{ig}=I_{ig}(X_i-\widehat\mu_g)/n_g$; the independent input applies each group's $n_g/(n_g-1)$ correction, and \eqref{appv62:cluster} clusters the uncorrected contributions over 47 states with factor $47/46$. Six states contain treated counties. Health centers have 36 coordinates for four weighted groups, eight pre-means and one post-average, fixed 1960 age-50+ population weights and 48 state clusters (factor $48/47$). The 24 treated counties span 16 states. The inverse sum of squared normalized treated weights is 11.49 across counties and 7.89 after aggregating weights by state. Both numbers describe weight concentration. Licensing uses three groups at four bins and either independent-subclass covariance or 19 parent classes (factor $19/18$); the largest parent supplies 36.2\% of the availability-comparison variance. The reported conventional comparisons use normal quantiles with these covariance estimates. The cluster counts and concentration matter for the joint normal approximation in Assumption~\MainRef{ass:array}. Police and Washington use the covariance models in Sections~\ref{app:police-construction} and \ref{appv3:cps-full}. TennCare uses the historical process below.

\subsection{TennCare: historical regional gaps and parameter support}\label{appv62:ggn-macro}
Let $G_t=(m_{TN,t}-m_{ESC,t},m_{TN,t}-m_{WSC,t},m_{TN,t}-m_{SA,t})'$ collect Tennessee's employment gaps against the three disjoint donor divisions. The all-southern gap is $(3G_{ESC,t}+4G_{WSC,t}+9G_{SA,t})/16$, preserving Tennessee's shared variation in every comparison. The primary history is 1995--2004 ($N=10$); 1986--2004 ($N=19$) checks inclusion of pre-expansion years. Centering or removing constant and linear trend, using the stationary bootstrap \citep{oa-PolitisRomano1994} with mean block lengths $L=2,3,4$, gives twelve covariance models. Detrending changes covariance while preserving the contrast and identifying assumptions.

For centered residuals $e_r\in\mathbb R^3$ and $a=1-1/L$, the covariance of the joint cyclic stationary-bootstrap process is
\begin{equation}\label{appv62:stationary-kernel}
 \Gamma_h^{\rm SB}=a^h\widehat\Gamma_h^{\rm circ},\qquad
 \widehat\Gamma_h^{\rm circ}=N^{-1}\sum_{r=1}^Ne_re_{1+((r+h-1)\bmod N)}'.
\end{equation}
Stack the three gaps in calendar-year order over 2000--2007 to form $W_8\in\mathbb R^{24\times24}$, using \eqref{appv62:stationary-kernel} and its transposes. Let $R$ select the five pre-treatment gaps and the 2006--2007 average by division, omitting 2005. Then
\begin{equation}\label{appv62:calendar-map}
 W=RW_8R'\in\mathbb R^{18\times18}.
\end{equation}
Equation~\eqref{appv62:calendar-map} retains the two-year 2004--2006 lag. The principal input uses $L=3$; all inputs are positive definite without a ridge. Historical gaps already include survey noise, so no extra person-sampling variance is added.

The Gaussian plug-in calculation assumes that this residual process describes the analysis-period gap covariance. Its estimation uses only 10 or 19 annual observations; covariance-estimation uncertainty, stationarity and the joint Gaussian approximation qualify the reported 95\% intervals. Across the twelve models, centered length ratios are 1--1.004907 and level-inclusive ratios 71.631999--102.252056; the adjusted centered-fit intervals exclude zero for all six 1995--2004 cases and include zero for all six 1986--2004 cases. For level-inclusive fit all twelve adjusted lower limits are positive.

For $\lambda=1$, the unrestricted interval is $[8.056695,698.212757]$ percentage points. Under ATT identification, intersecting with $[-100,100]$ gives $[8.056695,100]$ and length ratio 13.152 relative to $[.390544,7.381346]$. A known parameter set preserves coverage since $\{\theta\in C\cap\Theta\}=\{\theta\in C\}$ for $\theta\in\Theta$. The algebraic range of a difference of two percentage-point changes is instead $[-200,200]$; the support display uses the causal ATT interpretation.

\subsection{Selection boundaries and covariance sensitivity}\label{appv62:fibers}\label{appv63:chc-cov-sensitivity}
For level-inclusive fit, the accepted set for health centers is $(-\infty,5.634281]\cup[16.273324,\infty)$ in standard-error units centered at the observed estimate. Recomputing the line at six mixtures of state-clustered and independent-county covariance, with the same 36 means and five candidates, always selects the same urban pool; the state and county endpoint models both give a length ratio of 1.000 to three decimal places. In TennCare the upper boundary is only $.009475$ standard errors away (Lemma~\ref{diagv63:boundary-lemma}). We invert over every accepted component in both applications.

\FloatBarrier

\section{Simulation design and evaluation}\label{sec:simulation-supp}
\subsection{Regional designs and evaluation}\label{v63:simulation-embedding}\label{sim:full-results}
For independent individual vectors within and across regions, generate
\[
 (V_{ji},H_{ji})'=m_j+C_je_{ji},\qquad C_jC_j'=G_j,
\]
with lower-triangular Cholesky factors and independent coordinates of $e_{ji}$ that are standard normal or $\frac{G-2}{\sqrt2}$, $G\sim\mathrm{Gamma}(2,1)$ with scale 1. Set $\E H_j=0$, select $S=\arg\min_{j=1,2}|\bar V_0-\bar V_j|$, and report $T_S=\bar H_0-\bar H_S$. The panel $X_1=B-\frac{V}{2}$, $X_2=B+\frac{V}{2}$, $X_3=B+H$ gives every candidate parallel trends and ATT zero.
Here $B$ is an auxiliary baseline that cancels from both the pre-treatment
change and the DiD contrast. The regional covariances are
\begin{equation}\label{eq:sim-cov}
 G_0=\begin{pmatrix}.9&.65\\.65&1\end{pmatrix},\quad
 G_1=\begin{pmatrix}1.2&.30\\.30&1.4\end{pmatrix},\quad
 G_2=\begin{pmatrix}.7&.50\\.50&.8\end{pmatrix}.
\end{equation}
In the following list, $\E V$ abbreviates the regional mean vector
$(\E V_0,\E V_1,\E V_2)'$. The four settings are
\[
\begin{array}{lll}
\mathrm{T0}:&\E V=(0,0,0),&(n_0,n_1,n_2)=(r,2r,3r);\\
\mathrm{T+}:&\E V=(0,.5,-.5),&(n_0,n_1,n_2)=(r,2r,3r);\\
\mathrm{S}:&\E V=(0,.12,.65),&(n_0,n_1,n_2)=(r,2r,3r);\\
\mathrm{U}:&\E V=(0,.4r^{-\frac{1}{4}},-.4r^{-\frac{1}{4}}-.5r^{-\frac{1}{2}}),
 &(n_0,n_1,n_2)=(r,\operatorname{round}(r^{1.3}),2r).
\end{array}
\]
The initial 24 cells combine \eqref{eq:sim-cov} and these settings, $r=25,100,400$, and both innovations, with 5,000 draws each; eight independent $r=25$ confirmations use 40,000 each. The two designs in Main Table~\MainRef{tab:selection-evidence}, panels A--B, use $(25,50,75)$ and 40,000 draws each. An additional grid uses means $\frac{(0,.25,1)}{\sqrt r}$, $\frac{(0,.5,-.5)}{\sqrt r}$, and U, both innovations, and $r=25,50,100,200$: 24 cells with 100,000 draws each. Independent separated-normal and positive-fit-Gamma confirmations use 200,000 each.

The estimated-covariance calculation uses unbiased same-sample regional covariances divided by $n_j$. We compare conventional, estimated-covariance selective, and true-covariance selective 95\% intervals. The last is exact under normal observations and a normal-reference calculation under Gamma observations. All selective calculations retain the complete accepted union and recompute its direction for the covariance used, following Section~\MainRef{supp:full-fiber-inversion}; the accompanying design-and-results record preserves the expanded quadratic formulas. Coverage means that the true-target pivot lies in $[.025,.975]$. Marginal coverage uses all $R$ draws and conditional coverage the $N_s$ draws selecting $s$. Let $\widehat p$ and $\widehat p_s$ denote the corresponding empirical coverage proportions; their MCSEs are $\sqrt{\frac{\widehat p(1-\widehat p)}{R}}$ and $\sqrt{\frac{\widehat p_s(1-\widehat p_s)}{N_s}}$, and coverage is undefined for $N_s=0$.

For the equal-pre-change-variance baseline in Table~\ref{tab:fixed-pool-coverage-details}, every region has 200 normal observations, $\Var(V)=\Var(H)=1$, and $\Cov(V,H)=.9$ in the treated region and $-.9$ in controls. Each fixed candidate count has 100,000 draws; all ATTs are zero. Exact normal/Wishart sampling gives means and covariances from the same observations, with every fit comparison retained. Candidate counts define separate fixed-dimensional designs. With $n=200$ observations per region, the covariance of the contrast $\bar H_0-\bar H_s$ with its own fit discrepancy $\bar V_0-\bar V_s$ is zero. Its covariance with $\bar V_0-\bar V_j$ for a rival $j\ne s$ is $.9/n$ from the shared treated region. Additional checks use $n=50,200,800$ and correlation $.99$.

\begin{table}[tbp]\centering\setlength{\tabcolsep}{3pt}
\caption{Coverage by fixed candidate-pool size}\label{tab:fixed-pool-coverage-details}
\begin{tabular}{lrccc}\toprule
Pool size & Draws & Conventional & \shortstack{Feasible\\selective} & \shortstack{Known\\selective}\\\midrule
2 candidates & 100,000 & 93.27 (0.08) & 94.89 (0.07) & 94.91 (0.07)\\
5 candidates & 100,000 & 90.30 (0.09) & 95.07 (0.07) & 95.11 (0.07)\\
10 candidates & 100,000 & 88.01 (0.10) & 94.95 (0.07) & 94.99 (0.07)\\
39 candidates & 100,000 & 85.91 (0.11) & 95.06 (0.07) & 95.08 (0.07)\\\bottomrule
\end{tabular}
\note{Coverage in percent; parentheses contain Monte Carlo standard errors in percentage points. Each row is a separate normal design with 200 observations per region and 100,000 draws. The related regional experiment with varied pre-change variance is in Figure~\ref{fig:fixed-pool-size}, panel A.}
\end{table}

Main Table~\MainRef{tab:selection-evidence} uses the regional designs below; each selected-control row uses its own selected count.

\begin{table}[tbp]
\centering
\caption{Positive population-fit tie with Gamma innovations, conditional on candidate 2}\label{tab:gamma-positive-tie}
\begin{tabular}{rrrr}\toprule
Treated sample size & Reports & Estimated covariance & True covariance \\\midrule
25  & 2,584 & 93.38 (0.49) & 95.70 (0.40) \\
100 & 2,525 & 94.57 (0.45) & 95.41 (0.42) \\
400 & 2,489 & 94.78 (0.45) & 95.22 (0.43) \\\bottomrule
\end{tabular}
\note{Each design uses 5,000 draws and regional sizes in the ratio $1:2:3$. Entries are coverage in percent and Monte Carlo standard errors in percentage points. Both intervals use the normal reference distribution. Main Table~\MainRef{tab:selection-evidence}, panel B, reports a separate 40,000-draw run at treated sample size 25.}
\end{table}

\subsection{Combining paths in a staggered design}\label{supp:staggered-simulation}
Four normal known-covariance designs use two or four cohorts, three never-treated candidates, and 200,000 draws each. Candidates are separate across cohorts or reused. Adoption occurs at $3,4,\ldots$; centered squared fit over the two preceding periods selects a control retained at event times $0,1,3$. The five reports are cohort 1's time-0 effect, equal-weight cohort averages at those three event times, and their post-treatment average, with a fixed cohort subset. Independent region--period means have standard errors 1 in treated regions and $.5,1,2.5$ in controls; all ATTs and population contrasts are zero.

In Table~\ref{tab:sd-path-validation}, four cohorts share controls but reports concern the first two. The other two choices leave all five contrasts unchanged, so 81 full paths combine into nine events of nine compatible paths. Each procedure's conditional coverage is evaluated at the event specified in the table note. Joint coverage concerns the five effects reported under that event and is averaged over selections in the table.

\begin{table}[htbp]
\centering\small\setlength{\tabcolsep}{2pt}
\caption{Four-cohort selection: conditioning on full or combined paths}
\label{tab:sd-path-validation}
\begin{tabular}{>{\raggedright\arraybackslash}p{.22\linewidth}rrrrr}
\toprule
Procedure & \shortstack{Selection\\events} & \shortstack{Min.\\$N_s$} & \shortstack{Conditional\\coverage\\range} & \shortstack{Maximum\\MCSE} & \shortstack{Simultaneous\\coverage} \\
\midrule
Conventional & 81 & 282 & 92.31--99.38 & 1.12 & 96.12 \\
Cohort 1 choice & 81 & 282 & 92.12--99.26 & 1.15 & 96.23 \\
Joint path & 81 & 282 & 93.06--96.30 & 1.20 & 96.58 \\
Combined compatible paths & 9 & 7,082 & 94.85--95.40 & 0.25 & 96.50 \\
\bottomrule
\end{tabular}
\par\smallskip
\begin{minipage}{\textwidth}\normalsize\setstretch{1}
Coverage is in percent. The conditional range concerns nominal 95\% intervals for the first two cohorts' aggregate time-0 effect, evaluated on each of 81 full paths in the first three rows and each of nine combined events in the last. ``Cohort 1 choice'' constructs intervals using only that choice. Min. $N_s$ and maximum MCSE are the smallest selected count and largest conditional Monte Carlo standard error (percentage points). Simultaneous coverage is marginal for all five Bonferroni 99\% intervals, with total error 5\%.
\end{minipage}
\end{table}

For the first two cohorts' aggregate time-0 effect, both full-path and combined-event procedures give 2,047 finite intervals among the first 2,048 draws, with one numerical failure. Among those intervals, median/99th-percentile lengths are $7.89/453.32$ for full paths and $6.25/208.40$ for combined events. The least frequent event increases from 282 to 7,082 draws and maximum MCSE falls from 1.20 to .25 percentage point. These precision gains retain the reported estimators and require compatibility for every reported effect.

\subsection{Endpoint errors and interval-length dispersion}\label{app:endpoint-details}
We invert the full selection set in 1,000 draws per design for T0 and U, both innovations, and $r=25,100,400$. In normal T0 the corresponding 95th percentiles of the length ratio are 9.16, 12.31, and 11.13; all 3,000 intervals have finite endpoints. Main Section~\MainRef{sec:endpoint-sim} reports the endpoint errors and median ratios.

\subsection{Four-period fit and effect-directed rule search}\label{simv3:acd-full}\label{appv3:rule-search-full}
The normal design has three regional sample sizes $(50,100,150)$, four pre-treatment periods, and a common zero ATT. The two candidate mean paths relative to the treated path are $(1.5,-1.5,0,0)'/\sqrt{50}$ and $(.8,.8,-.8,-.8)'/\sqrt{50}$. Squared and absolute loss favor different candidates in the population. The search compares the two losses and three Huber rules with thresholds $.25/\sqrt{50},.5/\sqrt{50},1/\sqrt{50}$.

Among $R$ fitting rules, rule $r$ selects control $S_r(Y)$ and gives effect estimate $T_r(Y)=c_{S_r(Y)}'Y$. Give smaller rule indices priority at ties. For $\varsigma=1$ when choosing the largest estimate and $\varsigma=-1$ when choosing the smallest, the event for reporting $(r,s)$ is
\[
 E_{r,s}=\{S_r(Y)=s\}\cap
 \bigcap_{u<r}\{\varsigma T_r(Y)>\varsigma T_u(Y)\}\cap
 \bigcap_{u>r}\{\varsigma T_r(Y)\ge\varsigma T_u(Y)\}.
\]
For each tuple $(s_1,\ldots,s_R)$ with $s_r=s$, intersect $\{S_u(Y)=s_u\}$ over all rules and substitute $T_u(Y)=c_{s_u}'Y$ in the comparisons. Each intersection is defined by polynomial signs because effect comparisons are linear. The finite union over tuples and accepted loss pieces gives the complete event, with fixed complexity for a fixed menu. Retain every union component on the reported contrast's conditioning line.

Each search direction uses 40,000 paired draws with known covariance. For the larger-effect absolute-loss rule and candidate 2, accounting for rule search increases coverage by 19.37 percentage points (paired Monte Carlo standard error, 1.97 points). In an independent repetition, 468 draws reported this rule and candidate; within-rule and full-event coverage were 76.28 and 95.94 percent, respectively.

\section{Comparisons of inferential guarantees}\label{cmpv2:supp}\label{sim:related-procedures}
In the example in Appendix~\MainRef{app:selection-example} of the main paper, let $V_1$ and $V_2$ denote standardized treated-minus-control pre-treatment changes. A fixed two-sided test accepts the first comparison when $|V_1|\le c$ for a cutoff $c>0$. Choosing the first comparison by fit requires $|V_1|\le |V_2|$. This event contains $(1,2)$ and $(1,-2)$ but excludes their midpoint $(1,0)$, so it is nonconvex. If the winner must also pass, the full event is $\{|V_1|\le |V_2|\}\cap\{|V_1|\le c\}$. Conditioning on this intersection accounts for both decisions.

\citet{oa-AKM2024} study options selected by estimated performance under a joint normal approximation. A DiD adaptation using negative squared-fit scores must justify that approximation. For fixed $B_j$ with $B_jVB_j'\ne0$, if $\sqrt n(Y_n-\mu_n)\Rightarrow Z\sim N(0,V)$ and $\sqrt nB_j\mu_n\to h_j$, then Slutsky's theorem gives
\[
 n\|B_jY_n\|^2\Rightarrow\|h_j+B_jZ\|^2.
\]
The limit is nonconstant and nonnormal because it is a nonnegative quadratic form of a normal vector. We apply the normal approximation to outcome means and retain the polynomial fit comparisons. Main Section~\MainRef{sec:intro} describes the complementary candidate-family guarantee and its relation to inference conditional on the recorded report.

\subsection{Further fixed-comparison results}\label{supp:fixed-inherited}
The fixed-comparison specialization in Proposition~\MainRef{prop:fixed-screen}
uses the reporting event of \citet{oa-deCDH2026}. We keep their assumptions and
recover the following properties of conventional inference. In this
section \(X\) denotes the standardized limiting effect error, \(G\) the
limiting pre-test diagnostics, and \(C\) the conventional coverage region.

\paragraph*{Tests indexed by a continuum.}
Under their Assumption 2, the empirical diagnostics converge jointly with
the effect estimator to a centered Gaussian process and vector. Write
\(A=\{g:H(g)\leq0\}\), where \(H(g)=\max_j\{h_j(g)-q_j\}\),
\(h_j\) are the continuous, even, convex test functionals, and \(q_j>0\)
are their limiting critical values. The centered Gaussian law is
concentrated on a separable closed linear support. On that support \(A\)
is closed and convex, is symmetric about zero, and contains an open
neighborhood of zero. By linear separation and second countability,
\(A\) is an intersection of countably many symmetric strips. Applying
Gaussian correlation to each finite intersection and then passing to the
decreasing limit yields
\(\Pr(X\in C,G\in A)\geq\Pr(X\in C)\Pr(G\in A)\) for symmetric convex \(C\).
Gaussian log-concavity and convexity of \(H\) make
\(t\mapsto\Pr(H(G)\leq t)\) log-concave. Since \(H(0)<0\), this
function is positive and continuous near zero. Thus the reporting
event has positive probability and its threshold has zero mass.
The coverage boundary also has zero mass under the nonsingular
effect limit. The empirical-process limit, using outer probability
when required, gives their Theorem~2.

\paragraph*{Equality and one-sided inference.}
For their scalar effect and maximum absolute diagnostic test, the
conditional acceptance probability given \(X=x\) is even and
log-concave in \(x\). With nonzero cross-covariance it is positive near
zero and tends to zero at infinity, since the acceptance region is
bounded. Its covariance with the indicator of any nontrivial symmetric
interval is therefore strictly positive. For their whitened quadratic
test, condition on the diagnostic radius. The direction is independent
of the radius, and Gaussian regression leaves an independent centered
residual. For each direction, Anderson's inequality makes the coverage
probability nonincreasing in the radius. The direction average has
positive expectation and tends to zero for compact \(C\) and nonzero
cross-covariance, so its covariance with the pass indicator is strictly
positive. Zero cross-covariance instead gives independence. These are
their Proposition~1 cases.
With a symmetric screen, the conditional upper and lower tails of a
scalar effect error are equal. Dividing the two-sided tail bound equally
gives the one-sided bound for confidence levels at least one half
in their Proposition~2. Exchanging the roles of the scalar diagnostic
and the effect vector gives their corresponding one-sided screen result.

\paragraph*{Local misspecification.}
Under their Assumption 3, Gaussian means vary continuously at the null.
The strict coverage inequalities at the null therefore persist locally
for the tests in their Theorem~3. With zero cross-covariance, conditional
and unconditional coverage agree. For their Theorem~4 let
\(b=\operatorname{Cov}(G,X)\) and use their alignment of means to write
\(X=b'G+R\) with \(R\) centered Gaussian and independent of \(G\).
When \(b\ne0\), decompose the whitened \(G\) into its scalar projection
along \(b\) and the independent orthogonal coordinates. Given the
projection, acceptance and coverage are independent, and both
probabilities are nonincreasing in its absolute value. Their
nonnegative covariance gives conditional coverage at least as large
as unconditional coverage for every allowed mean shift.
\begin{figure}[!htbp]
\centering
\caption{Reporting regions for candidate 1}
\label{fig:fixed-versus-ranked}
\includegraphics[width=.96\linewidth]{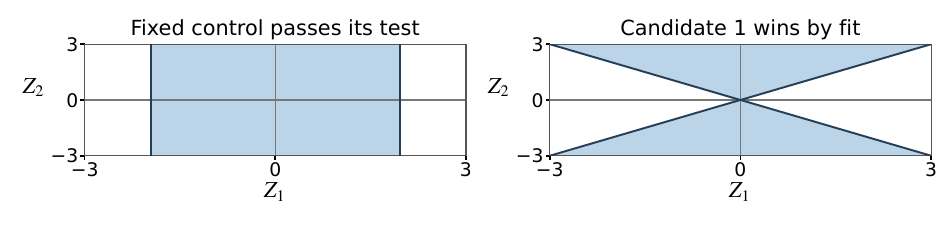}
\par\noindent\textit{Note.} $Z_1,Z_2$ are standardized pre-treatment discrepancies. Shading shows the fixed screen $|Z_1|\le1.96$ (rounded) and the winning region $|Z_1|\le|Z_2|$. Ranking and then screening requires both conditions.
\end{figure}

\begin{table}[!htbp]
\centering
\caption{False rejections of the true ATT null after control selection}
\label{tab:two-control-coverage}
\begin{tabular}{@{}p{.61\linewidth}rr@{}}
\toprule
Reporting rule & Usual test & Selective test \\
\midrule
Fix a control and report it if its pre-trend test passes & 5.00 & 5.00 \\
Choose the control with better pre-treatment fit & 9.38 & 5.00 \\
Choose by fit and require the winner to pass & 9.23 & 5.00 \\
\bottomrule
\end{tabular}

\smallskip
\begin{minipage}{0.94\linewidth}
\small
\textit{Note.} Rejection rates in percent among reports in the three-period Gaussian example. Both controls satisfy parallel trends and identify the same ATT. A test rejects when its 95 percent interval excludes the true ATT. Selective inference uses the true covariance and the full rule. Pre-tests use a two-sided 5 percent cutoff (Main Appendix~\MainRef{app:selection-example}).
\end{minipage}
\end{table}
\subsection{Checking recorded selection rules}\label{supp:replayable-selectors}
A recorded rule takes the joint data summaries and determines exactly which comparison is reported, including a no-report outcome when the rule allows one. Here $y$ is the fixed-dimensional vector of summaries in main Section~\MainRef{poly:setup}. To verify (S), write each possible path as a sequence of sign checks. Suppose each check on each possible path uses a polynomial in $y$, with a bounded total number and degree of distinct comparisons across all paths along the sampling sequence. List the nonzero polynomials appearing on any path. Their signs determine the first decision, the next check, and ultimately the report, so they define the map in main equation~\MainEqRef{poly:selector}. An identically zero polynomial gives sign zero. Use the recorded tie or failure convention. The event for a label is the union of all paths assigned to it, including paths that differ in intermediate decisions. A fixed contrast for each effect-reporting label gives the same estimate as a function of $y$ on all such paths. Conversely, a selector satisfying (S) can first check those signs and then apply its label map. These facts show why finite combinations of screens, rankings, and prespecified fallbacks use the same representation.

For example, suppose candidates A, B, and C are screened and then ranked by fit. Reporting B means B passes and outranks each of A and C that also passes, using the fixed priority at a tie. When replaying the rule along the conditional line, reevaluate whether A and C pass at each trial value. Reporting B can occur on several distinct parts of the line.

Let $r_{st}(y)$ denote the affine pre-treatment gap for candidate $s$ at fit coordinate $t$. Absolute loss splits at zero and Huber loss at its two thresholds. Main Appendix~\MainRef{poly:piecewise-supp} gives the sign checks, including structural ties.

\paragraph{Comparing fit criteria.}
A fixed-weight squared-error ranking compares polynomials directly. Centering the gaps, assigning a separate weight to the mean gap, or using fixed linear contrasts for pre-treatment slopes preserves the polynomial form. For absolute and Huber loss, use the pieces in Appendix~\MainRef{poly:piecewise-supp}. For a worst-period criterion, partition first by the signs of the gaps and then by comparisons between their weighted absolute values. On each piece, the maximum is one of the signed, weighted gaps. The same ordering can select a fixed quantile, such as either middle order statistic or the mean of the two middle order statistics. Comparisons between the resulting scores, with the recorded priority at ties, determine the report. For two polynomial numerators $a_1(y)$ and $a_2(y)$ with strictly positive polynomial denominators $b_1(y)$ and $b_2(y)$, the sign of $a_1(y)b_2(y)-a_2(y)b_1(y)$ compares the two normalized scores. Include signs for the denominators and record what the rule does when a score is undefined.

\paragraph{Screens based on estimated standard errors.}
For a two-sided screen with deterministic critical value $k>0$, suppose the test statistic uses a polynomial estimate $a(y)$ of its numerator and a polynomial estimated variance $v(y)$. On $v(y)>0$, passing the screen means $a(y)^2-k^2v(y)\le0$ for an inclusive cutoff. Record the rule for $v(y)\le0$ too. For example, for $N\ge2$ independent identically distributed units with measured vectors $X_i$, take $m=N^{-1}\sum_{i=1}^N X_i$ and $H=N^{-1}\sum_{i=1}^N X_iX_i'$. For a fixed direction $h$ and null value $a_0$, the statistic uses $a(y)=h'm-a_0$ and $v(y)=h'(H-mm')h/(N-1)$. The screen is a polynomial sign check in $y=(m,\operatorname{vech}(H))$, where $\operatorname{vech}$ stacks the unique lower-triangular entries. A reported contrast linear in $m$ has zero weights on the second moments. Apply Assumption~\MainRef{ass:array} to the enlarged vector, dropping known redundant linear coordinates where applicable, and estimate its full joint covariance. Thus the selection rule recomputes its own estimated standard errors at alternative values of $y$. The covariance estimate that constructs the conditional line and reference law is a separate input.

\paragraph{A multivariate Wald screen.}
Take coordinates in a fixed basis spanning the $q$ independent pre-treatment contrasts. Let $g(y)\in\mathbb R^q$ be the vector of estimated gaps in these coordinates and let $W(y)$ be its estimated covariance. For the independent-unit example above, a fixed rank-$q$ matrix $G$ gives $g(y)=Gm-g_0$ and $W(y)=G(H-mm')G'/(N-1)$, where $g_0$ is the null value. On the positive-definite branch, the screen $g(y)'W(y)^{-1}g(y)\le k_W$ for a deterministic cutoff $k_W$ is equivalent to
\[
 g(y)'\operatorname{adj}\{W(y)\}g(y)-k_W\det W(y)\le0.
\]
The adjugate is the transpose of the cofactor matrix. The inequality keeps its direction because $\det W(y)>0$. Positivity of the leading principal minors defines the positive-definite branch by polynomial signs. Record the decision on other branches of the rule. For fixed $q$, the entries of the adjugate and the determinant are polynomials of bounded degree in the enlarged summary $y$. Apply Assumption~\MainRef{ass:array} to that summary and use its joint covariance for inference. In the Gaussian sample-Wald experiment in Section~\ref{v93:matched-oa}, we instead use independence of the sample means and sample covariance to condition directly on the latter.

\paragraph{Ranking by the estimated effect.}
Suppose the researcher screens controls and then ranks the survivors by the absolute value of their DiD effect statistics. Write $a_s(y)$ for the effect numerator and $v_s(y)>0$ for its estimated variance. For each other survivor $j$, choosing $s$ means $a_s(y)^2v_j(y)-a_j(y)^2v_s(y)\ge0$, with the recorded tie convention. Candidate eligibility is rechecked at each trial outcome. When variances are estimated as in the preceding paragraph, use the enlarged summary for both the variance comparisons and the eligibility checks. Retain the union of all parts of the conditional line giving the same report.

\paragraph{Fixed-design regression rules.}
For a finite menu of full-rank, fixed-design least-squares specifications, the estimated coefficients and held-out prediction errors are fixed linear functions of the outcome vector. With fixed validation splits, squared prediction losses give quadratic comparisons. With usual fixed-design residual variance, the variance estimate in a coefficient test is a known multiple of a quadratic form in outcomes, so the same two-sided screen applies. For large samples, verify a fixed-dimensional set of summaries that preserves the decisions and the candidate contrasts and satisfies Assumption~\MainRef{ass:array}. The same construction allows selection among a finite menu of prespecified weighted control pools. With continuously refitted weights, check both the selected report's contrast and the sign representation.

\paragraph{The Gaussian benchmark and the feasible guarantee.}
In a jointly Gaussian model with known covariance and a fixed linear contrast for each possible report, Lemma~\MainRef{lem:gaussian} gives exact conditional inference for any measurable selector on the joint vector, at a positive-probability report with positive contrast variance. The conditional law uses the entire set of values producing the report. The sign representation supplies a finite interval calculation for that set. With covariance estimated from the same data, Theorem~\MainRef{thm:plugin} uses (S) and Assumption~\MainRef{ass:array}. Its conditional error depends on the probability of the report. When weights, periods, or thresholds are chosen using the outcomes, their effect on the report must enter the recorded rule. For example, a same-sample screening variance can enter through the summary vector above. When the decisions depend on an unrecoverable search or informal judgment, specifying how candidate controls are compared prospectively makes the rule replayable. The causal interpretation as an average treatment effect on the treated retains the assumptions in main equation~\MainEqRef{sd:single-decomposition}.

\section{When control selection increases estimation-error dispersion}\label{disp:oa}

We characterize how control selection affects estimation-error dispersion and develop the normal experiment used in the main article. Every fixed comparison identifies the same ATT and has the same variance, so we measure amplification against a common fixed-control variance. Section~\ref{v58:normal-functional-section} expresses the variance ratio through the candidate directions, and Section~\ref{v57:normal-geometry-section} bounds it as path noise changes. Shared factors allow dependence across regions under Assumption~\MainRef{ass:array}.

\subsection{Variance, mean squared error, and the comparison distribution}\label{disp:comparison-distribution}
The target is the expected conditional-variance ratio in main equation~\MainEqRef{v58:main-expected-ratio}; Proposition~\ref{amp:conditional-moments} in Section~\ref{amp:distribution-free-conditions} gives its general identity, and main Theorem~\MainRef{v58:expected-variance-normal} treats normal outcomes. The design below has unit fixed-comparison variances and zero selected-label error means, so the ratio equals reported-error variance and mean squared error. All candidates' fit information enters the joint dependence.

\subsection{A normal panel with large dispersion after fit-based selection}\label{disp:gaussian}

Fix $m$ pre-treatment directions and $J$ candidate controls. Let $u_1,\ldots,u_J\in\R^m$ be unit vectors, fixed before sampling. Let $Z,\zeta_0,\ldots,\zeta_J$ be independent $N_m(0,I_m)$ vectors, and let $B_0,\ldots,B_J,F,e_1,\ldots,e_J$ be independent standard normal scalars, independent of those vectors. Set $P=m+1$ and choose a $P\times m$ matrix $A$ satisfying
\begin{equation}\label{disp:basis}
 A'A=I_m,\qquad A'\mathbf1_P=0.
\end{equation}
This matrix expresses centered pre-treatment paths in orthonormal coordinates. For $\delta>0$ and $0<\rho<1$, define the observed regional pre-treatment paths and the single post-treatment observation by
\begin{align}
 Y_{0,P}&=B_0\mathbf1_P+A(Z+\delta\zeta_0),\nonumber\\
 Y_{s,P}&=B_s\mathbf1_P+A(u_su_s'Z+\delta\zeta_s),
                    &&s=1,\ldots,J,\label{disp:panel-pre}\\
 Y_{0,+}&=B_0+F+\tau,\nonumber\\
 Y_{s,+}&=B_s+F-\rho u_s'Z-\sqrt{1-\rho^2}e_s.
                    &&\label{disp:panel-post}
\end{align}
The untreated potential outcome for region 0 in the post-treatment period subtracts $\tau$ from $Y_{0,+}$; treatment adds the constant $\tau$. All untreated means are zero. Arbitrary deterministic regional levels and a common deterministic time path can be added without changing the comparisons. Consequently every candidate satisfies parallel trends and identifies ATT $\tau$.

By \eqref{disp:basis}, centered squared pre-treatment fit is exactly
\begin{equation}\label{disp:gaussian-fit}
 Q_s=\left\|(I_m-u_su_s')Z+\delta(\zeta_0-\zeta_s)\right\|^2,
 \qquad S=\arg\min_s Q_s.
\end{equation}
The DiD contrast using the pre-treatment mean is
\begin{equation}\label{disp:gaussian-error}
 \widehat\tau_s=\tau+U_s,\qquad
 U_s=\rho u_s'Z+\sqrt{1-\rho^2}e_s\sim N(0,1).
\end{equation}
Researchers select controls using the pre-treatment paths alone, independently of $e_s$. Candidate errors need not be independent: $\Cov(U_r,U_s)=\rho^2u_r'u_s$ for $r\ne s$.

\paragraph*{Nonsingularity.}
Transform the outcome vector into each region's pre-treatment mean, its $m$ centered pre-treatment coordinates, and its post-treatment outcome minus the pre-treatment mean. This transformation is invertible. Conditional on $Z$, the pre-treatment means contain the independent $B_j$, the centered coordinates contain the independent $\delta\zeta_j$, and the post-treatment differences contain $(F,F-\sqrt{1-\rho^2}e_1,\ldots,F-\sqrt{1-\rho^2}e_J)$, up to constants. The last vector is an invertible linear transformation of $(F,e_1,\ldots,e_J)$. These innovation blocks are independent and nonsingular. The conditional covariance, and hence the full outcome covariance, is therefore positive definite. The design allows dependence across regions under the full-covariance conditions in the main article.

\paragraph*{The amplification mechanism.}
At $\delta=0$, equation~\eqref{disp:gaussian-fit} gives $Q_s=\|Z\|^2-(u_s'Z)^2$: a candidate fits better when its path captures a larger projection of the common pre-treatment fluctuation. With positive $\delta$, the same comparison is perturbed by independent path noise, and the covariance remains nonsingular. The projection that improves pre-treatment fit also enters the ATT error in \eqref{disp:gaussian-error}. For any $\delta>0$, the simultaneous sign reversal of $Z$ and all $\zeta_j$ leaves $S$ unchanged and reverses $u_S'Z$. Thus $\E U_S=0$ and
\begin{equation}\label{disp:gaussian-ratio}
 R=\frac{\Var(U_S)}{\Var(U_s)}
   =1-\rho^2+\rho^2\E[(u_S'Z)^2].
\end{equation}
Variance can increase substantially without mean bias. Conditional on $Z$, all pre-treatment noises and $S$, the selected error is normal with mean $\rho u_S'Z$ and variance $1-\rho^2$. This gives both \eqref{disp:gaussian-ratio} and fixed-comparison interval coverage by averaging normal probabilities.

At 39 candidates the strongest loading in the experiment gives almost five times the variance for a fixed comparison. At $J=3{,}000$, $m=40$, $\rho=.95$ and $\delta=.1$, independent validation with 200,000 draws gives a variance ratio of \DispGaussianRatio{} (Monte Carlo standard error \DispGaussianMCSE{}) and conventional coverage of \DispGaussianCoverage{} percent. At zero loading, conventional coverage is 95 percent. Each candidate count is a separate fixed-dimensional design.

\section{Conditions for variance amplification}\label{amp:proofs}

\subsection{Variance amplification as a function of candidate dependence and path noise}
\label{v58:normal-functional-section}
Fix unit vectors $\mathcal U=(u_1,\ldots,u_J)$ in $\R^m$, $J\ge1$,
and write $C=(u_s'u_t)_{s,t\le J}$ for their Gram matrix. Let
$Z,H,H_1,\ldots,H_J$ be independent $N_m(0,I_m)$ vectors and let
$e_1,\ldots,e_J$ be independent standard normal variables, independent of
those vectors. For $0\le\rho<1$ and $\delta\ge0$, put
\[
 Q_s^\delta=\|(I_m-u_su_s')Z+\delta(H-H_s)\|^2,\qquad
 S_\delta=\arg\min_s Q_s^\delta,\qquad
 U_s=\rho u_s'Z+\sqrt{1-\rho^2}e_s,
\]
with a fixed rule for ties. Every fixed error has variance one. Define
\[
 R_\delta(\mathcal U,\rho)=\Var(U_{S_\delta}\mid\mathcal U),\qquad
 p_s=\Pr(S_\delta=s\mid\mathcal U).
\]
We condition on the candidate collection fixed before outcome sampling,
so $R_\delta$ measures the population variance ratio for that pool. For $p_s>0$, also define
$r_s=\Var(U_s\mid S_\delta=s,\mathcal U)$, the ratio conditional on reporting
candidate $s$. Since every fixed variance equals one, $r_s$ also measures the conditional variance ratio.
The sign symmetry established below gives $\nu_s=0$, so that
$R_\delta=\mathcal R_E$ in this model.

Let $F_{m,\lambda}$ and $f_{m,\lambda}$ be the distribution and density
functions of a noncentral chi-square variable with $m$ degrees of freedom
and noncentrality $\lambda$. For $\delta>0$, define
\begin{align}
 \lambda_s(z,h;\delta)
   &=\|\delta^{-1}(I_m-u_su_s')z+h\|^2,\label{v58:ncp}\\
 w_s(z,h;\delta)
   &=\int_0^\infty f_{m,\lambda_s(z,h;\delta)}(x)
       \prod_{j\ne s}\{1-F_{m,\lambda_j(z,h;\delta)}(x)\}\,dx.
       \label{v58:selection-weight}
\end{align}
These are explicit functions of the normal observations and the fixed
candidate directions, involving a one-dimensional score integral.

\begin{theorem}[The variance-ratio function]\label{v58:normal-functional}
For every fixed candidate collection, each selected-label error has mean
zero, and
\begin{equation}\label{v58:label-average}
 R_\delta=\sum_{s:p_s>0}p_s r_s
        =1+\rho^2\{G_\delta(C;m)-1\}.
\end{equation}
For $\delta>0$, the function in this expression is
\begin{equation}\label{v58:positive-noise-functional}
 G_\delta(C;m)=
 \E_{Z,H}\sum_{s=1}^J(u_s'Z)^2w_s(Z,H;\delta).
\end{equation}
In particular,
\[
 p_s=\E_{Z,H}w_s(Z,H;\delta),\qquad
 r_s=1-\rho^2+
       \rho^2\frac{\E_{Z,H}[(u_s'Z)^2w_s(Z,H;\delta)]}{p_s}.
\]
Although a representation by $u_1,\ldots,u_J$ appears in
\eqref{v58:positive-noise-functional}, its value depends on that collection
only through $C$ and the ambient dimension $m$.

At zero path noise, let $W\sim N_J(0,C)$, allowing singular $C$. Then
\begin{equation}\label{v58:zero-noise-functional}
 G_0(C)=\E\max_{s\le J}W_s^2
 =\int_0^\infty
  \left[1-\Pr\{-\sqrt x\le W_s\le\sqrt x\text{ for every }s\}\right]dx.
\end{equation}
The function $G_\delta(C;m)$ is continuous in $\delta\in[0,\infty)$,
satisfies $0\le G_\delta(C;m)\le G_0(C)\le\operatorname{rank}(C)$,
and has limits
\begin{equation}\label{v58:noise-limits}
 \lim_{\delta\downarrow0}G_\delta(C;m)=G_0(C),\qquad
 \lim_{\delta\to\infty}G_\delta(C;m)=1.
\end{equation}
Consequently, for $0<\rho<1$ and any target $q>1$,
\begin{equation}\label{v58:functional-threshold}
 R_\delta>q
 \quad\Longleftrightarrow\quad
 G_\delta(C;m)>1+\frac{q-1}{\rho^2}.
\end{equation}
By \eqref{v58:noise-limits}, large path noise removes amplification in the limit. The threshold in \eqref{v58:functional-threshold} follows because the ratio is affine in $\rho^2$; its slope is $G_\delta(C;m)-1$.
Multiplying every pre-treatment discrepancy and every ATT error by the same
positive sampling scale does not change the ratio.
\end{theorem}

\begin{proof}
Conditional on $Z=z,H=h$, the vectors
$\delta^{-1}(I_m-u_su_s')z+h-H_s$ are independent normal vectors with
identity covariance. Therefore $\frac{Q_s^\delta}{\delta^2}$ have independent,
not necessarily identically distributed, noncentral chi-square laws with
the noncentralities in \eqref{v58:ncp}. Their laws are continuous, so
ties have zero conditional probability. Conditioning on the score of
candidate $s$ gives \eqref{v58:selection-weight}, and the weights sum to
one. In particular, $w_s(z,h;\delta)$ is the conditional probability that
candidate $s$ is selected.

Simultaneously replacing $Z,H,H_1,\ldots,H_J$ by their negatives preserves
the selected label and reverses every projection $u_s'Z$. Thus
$\E[u_s'Z\,\mathbf1\{S_\delta=s\}]=0$. Conditional on all fit
information, $e_{S_\delta}$ has zero mean and variance one. This yields
zero mean both for each selected label and for the overall selected error,
and gives
\[
 \Var(U_{S_\delta})
  =1-\rho^2+\rho^2\E(u_{S_\delta}'Z)^2.
\]
Conditioning on $Z,H$ proves the exact functional and the formulas for
$p_s$ and $r_s$. The law of total variance, with all label-conditional means
zero, proves the selected-label average in \eqref{v58:label-average}.

Two collections with the same Gram matrix are related by an isometry of
their spans, which extends to an orthogonal transformation of $\R^m$.
Applying that transformation to the rotationally invariant independent
normal vectors leaves the joint scores and projections unchanged in law.
This proves that $C$ and $m$ are the only directional inputs to the
positive-noise function. At zero noise,
$Q_s^0=\|Z\|^2-(u_s'Z)^2$, so the selected squared projection is the
largest one. The vector of projections has covariance $C$, giving
\eqref{v58:zero-noise-functional} by the tail-integral formula.

For a linear subspace $H$, let $P_H$ denote the Euclidean orthogonal projection onto $H$. For every realization,
\[
 0\le (u_{S_\delta}'Z)^2
 \le\max_s(u_s'Z)^2
 \le\|P_{\operatorname{span}(\mathcal U)}Z\|^2.
\]
Taking expectations gives the bounds and an integrable envelope. At any
fixed $\delta>0$, ties have probability zero, so continuity of the finite
list of score functions makes the selected label locally constant as the
noise parameter changes, almost surely. At zero, the scores converge to
$\|Z\|^2-(u_s'Z)^2$; even if tied directions occur, every zero-noise
minimizer has the same squared projection. Finiteness of the candidate
pool then gives $(u_{S_\delta}'Z)^2\to\max_s(u_s'Z)^2$ almost surely
as $\delta\downarrow0$. Dominated convergence proves continuity including
at zero.

For large noise,
$\frac{Q_s^\delta}{\delta^2}\to\|H-H_s\|^2$ almost surely. Conditional on
$H$, these limiting scores are independent continuous variables and their
minimizer $S_\infty$ is unique almost surely. Hence
$S_\delta\to S_\infty$ almost surely. This limiting choice is independent
of $Z$, so $\E(u_{S_\infty}'Z)^2=1$. The same envelope proves the
large-noise limit. The remaining assertions follow from
\eqref{v58:label-average}; a common positive scale changes neither the
ranking nor a variance ratio.
\end{proof}

We can therefore evaluate the ratio from the complete correlation pattern
among candidate directions, the dimension in which path noise operates,
the relative path-noise scale, and the loading of the relevant pre-treatment
fluctuations into ATT error. For intermediate noise levels, changes in $C$
or in $\delta$ enter through the full functional; the two endpoint limits
alone do not determine their comparative statics.
The ratio in \eqref{v58:label-average} averages over the selected label.
It differs from averaging $\Var(U_{S_\delta}\mid X)$ over complete
pre-treatment information $X$: when $X$ reveals $Z$ and the choice,
that conditional variance is only $1-\rho^2$, and the additional
uncertainty comes from variation in the conditional mean.

\paragraph{Repeated and orthogonal directions.}
At zero noise, duplicating a direction or its negative leaves
\eqref{v58:zero-noise-functional} unchanged. If candidates consist of
copies of $K$ mutually orthogonal directions, with any signs and
multiplicities, then
\begin{equation}\label{v58:orthogonal-function}
 G_0(C)=\int_0^\infty
       [1-\{2\Phi(\sqrt x)-1\}^{K}]\,dx.
\end{equation}
In \eqref{v58:orthogonal-function}, the number of distinct directions determines the ratio,
with repeated candidate labels leaving it unchanged. These statements follow immediately
from the maximum of the $K$ independent squared normal projections.
\subsection{Candidate directions and the size of amplification}\label{v57:normal-geometry-section}
Fix unit vectors $u_1,\ldots,u_J\in\R^m$ before sampling. Let $Z,\zeta_0,\ldots,\zeta_J$ be independent $N_m(0,I_m)$ vectors and $e_1,\ldots,e_J$ independent standard normal variables, independent of those vectors. For $0<\rho<1$, $\delta\ge0$, define
\[
 \begin{gathered}
 A_s=(u_s'Z)^2,\quad
 Q_s^\delta=\|(I_m-u_su_s')Z+\delta(\zeta_0-\zeta_s)\|^2,
 \\
 S_\delta=\arg\min_sQ_s^\delta,
 \quad U_s=\rho u_s'Z+\sqrt{1-\rho^2}e_s.
 \end{gathered}
\]
Let $H=\operatorname{span}\{u_s:s\le J\}$ have dimension $r$, and define
\[
 G_{\mathcal U}=\E\max_s(u_s'Z)^2,\qquad
 \gamma_{\mathcal U}=\inf_{z\in H,\|z\|=1}\max_s|u_s'z|,
 \qquad R_\delta=\Var(U_{S_\delta}).
\]
The denominator of $R_\delta$ is one, the variance of every fixed comparison.

\begin{theorem}[Geometry, noise, and large variance ratios]\label{v57:normal-geometry}
The selected error has mean zero, and
\[
 R_\delta=1-\rho^2+\rho^2\E A_{S_\delta}.
\]
At $\delta=0$,
\begin{equation}\label{v57:normal-geometry-ratio}
 R_0=1+\rho^2(G_{\mathcal U}-1),\qquad
 \max\{1,r\gamma_{\mathcal U}^2\}\le G_{\mathcal U}\le r.
\end{equation}
Adding directions to a fixed collection weakly increases $R_0$; duplicating a direction or its negative leaves it unchanged. For any $q>1$, $R_0$ exceeds $q$ exactly when $G_{\mathcal U}>1+\frac{q-1}{\rho^2}$. In particular, tenfold amplification requires $r>1+\frac{9}{\rho^2}$ in this family.

For positive noise, write
\[
 \xi_s^\delta=2\delta\{(I_m-u_su_s')Z\}'(\zeta_0-\zeta_s)
      +\delta^2\|\zeta_0-\zeta_s\|^2,
 \qquad K_\delta=\E\max_s|\xi_s^\delta|.
\]
Then
\begin{equation}\label{v57:normal-noise-bound}
 R_0-\rho^2\E\{\max_s\xi_s^\delta-\min_s\xi_s^\delta\}
 \le R_\delta\le R_0,
 \qquad R_\delta\ge R_0-2\rho^2K_\delta.
\end{equation}
For
\[
 U_{m,J}=\inf_{0<t<\frac{1}{4}}
 \frac{\log J-(\frac{m}{2})\log(1-4t)}{t},
\]
one has
\begin{equation}\label{v57:normal-K-bound}
 K_\delta\le2\delta\sqrt{mU_{m,J}}+\delta^2U_{m,J}.
\end{equation}
Thus, if a deterministic $g\le G_{\mathcal U}$ satisfies $\Delta:=1+\rho^2(g-1)-q>0$, the explicit sufficient condition
\begin{equation}\label{v57:normal-delta-condition}
 0<\delta<
 \frac{\sqrt{m+\frac{\Delta}{2\rho^2}}-\sqrt m}{\sqrt{U_{m,J}}}
\end{equation}
gives $R_\delta>q$.
\end{theorem}
\begin{proof}
Simultaneous sign reversal of all pre-treatment normal innovations preserves the fit scores and the selected label and reverses $u_{S_\delta}'Z$. The selected $e_{S_\delta}$ has zero mean and unit variance conditional on all selection information. This proves the mean and variance identity. At zero noise $Q_s^0=\|Z\|^2-A_s$, so the selected squared projection is the maximum. This maximum dominates any fixed projection, is bounded by $\|P_HZ\|^2$, and is bounded below by $\gamma_{\mathcal U}^2\|P_HZ\|^2$. Taking expectations gives \eqref{v57:normal-geometry-ratio} and the nested-collection statements.

Choose $s_*$ maximizing $A_s$. Optimality of $S_\delta$ gives
\[
 0\le A_{s_*}-A_{S_\delta}
 \le\xi_{s_*}^\delta-\xi_{S_\delta}^\delta
 \le\max_s\xi_s^\delta-\min_s\xi_s^\delta
 \le2\max_s|\xi_s^\delta|.
\]
Taking expectations proves \eqref{v57:normal-noise-bound}. Let $W_s=\|\zeta_0-\zeta_s\|^2$. Each $W_s$ has moment generating function $(1-4t)^{-\frac{m}{2}}$; independence across candidates is not required. The log-sum-exp inequality and Jensen's inequality give
\[
 \E\max_sW_s\le t^{-1}\log\sum_s\E e^{tW_s}
 =\frac{\log J-\frac{m}{2}\log(1-4t)}{t}.
\]
Also, pointwise,
$\max_s|\xi_s^\delta|\le2\delta\|Z\|\sqrt{\max_sW_s}+\delta^2\max_sW_s$.
Cauchy--Schwarz and $\E\|Z\|^2=m$ give \eqref{v57:normal-K-bound}. Substituting into the lower variance bound and solving for $\delta$ gives \eqref{v57:normal-delta-condition}.
\end{proof}

\begin{corollary}[Clusters of candidate directions]\label{v57:normal-clusters}
Suppose every $u_s$ is within distance $\eta$ of either $v_k$ or $-v_k$ for one of $K$ fixed unit vectors $v_1,\ldots,v_K$. Let $H$ contain all these directions and have dimension $r$, and write $G_{\mathcal V}=\E\max_k(v_k'Z)^2$. Then
\[
 \sqrt{G_{\mathcal U}}\le\sqrt{G_{\mathcal V}}+\eta\sqrt r,
 \qquad
 R_0\le1-\rho^2+\rho^2\min\{r,(\sqrt{G_{\mathcal V}}+\eta\sqrt r)^2\}.
\]
If, in addition, every $v_k$ is within Euclidean distance $\eta$ of either $u_s$ or $-u_s$ for some $s$, then
$|\sqrt{G_{\mathcal U}}-\sqrt{G_{\mathcal V}}|\le\eta\sqrt r$.
\end{corollary}
\begin{proof}
For each selected direction and the corresponding signed center,
$|u_s'Z|\le|v_k'Z|+\eta\|P_HZ\|$. Maximize, take the $L^2$ norm, and apply its triangle inequality. The reverse comparison under two-sided approximation is identical. Apply the theorem's rank bound for the stated refinement.
\end{proof}

\begin{proposition}[Conventional coverage in the projection design]\label{v57:normal-coverage}
Let $c=\Phi^{-1}(1-\frac{\alpha}{2})$, $0<\alpha<1$, and
\[
 h(a)=\Phi\{\frac{c-\rho\sqrt a}{\sqrt{1-\rho^2}}\}
       -\Phi\{\frac{-c-\rho\sqrt a}{\sqrt{1-\rho^2}}\}.
\]
The coverage of $\widehat\tau_{S_\delta}\pm c$ is $C_\delta=\E h(A_{S_\delta})$. At zero noise, $C_0\le1-\alpha$, with strict inequality whenever the collection contains two nonparallel directions. Adding directions weakly decreases $C_0$. For positive noise,
\[
 0\le C_\delta-C_0\le2L K_\delta,
 \qquad L=\frac{\rho^2}{(1-\rho^2)\sqrt{2\pi e}}.
\]
In particular, $2LK_\delta<(1-\alpha)-C_0$ is sufficient for strict undercoverage at positive noise.
\end{proposition}
\begin{proof}
Conditional on the selection information, the remaining normal error has variance $1-\rho^2$. Symmetry of a centered interval means its conditional coverage depends only on $A_{S_\delta}$, yielding the formula. As a function of $\sqrt a$, coverage strictly decreases away from zero. The maximum squared projection dominates any fixed squared projection, whose mixture gives exact coverage $1-\alpha$. Two nonparallel normal projections produce a strict improvement in the maximum with positive probability. This proves strictness and the nested-collection statement. To bound the derivative in $a$, write the coverage as $g(t)$ at $t=\rho\sqrt a$, where
$g'(t)=\frac{\phi\bigl(\frac{c+t}{s}\bigr)-\phi\bigl(\frac{c-t}{s}\bigr)}{s}$, $s=\sqrt{1-\rho^2}$. Since $\sup_x|\phi'(x)|=\frac{1}{\sqrt{2\pi e}}$, $|g'(t)|\le\frac{2t}{s^2\sqrt{2\pi e}}$. Hence $|h'(a)|\le L$, including its limit at zero. Apply the pointwise squared-projection regret bound in the theorem and take expectations.
\end{proof}

The DiD construction in Section~\ref{disp:gaussian} realizes these scores and ATT errors with parallel trends for every candidate and positive definite full outcome covariance for $\delta>0$ and $\rho<1$.
\begin{corollary}[Sharp Gaussian bound on false rejection]\label{v91:sharp-size}
Suppose $J\ge1$ fixed-comparison test statistics $(W_1,\ldots,W_J)$ are jointly Gaussian and each is standard normal under the common ATT null. For $0<\alpha<1$, the conventional two-sided level-$\alpha$ test for comparison $j$ rejects if $|W_j|>c_\alpha$, where $c_\alpha=\Phi^{-1}(1-\alpha/2)$. Let $S$ be any rule that always reports one comparison. The unconditional probability that the reported conventional test falsely rejects satisfies
\[
 \Pr(|W_S|>c_\alpha)\le 1-(1-\alpha)^J.
\]
This bound is the supremum even over selection by pre-treatment fit in Gaussian DiD designs with positive definite outcome covariance and a common ATT.
\end{corollary}
\begin{proof}
For each $j$, the nonrejection set $\{w\in\R^J:|w_j|\le c_\alpha\}$ is symmetric and convex. The Gaussian rectangle inequality of \citet{oa-Sidak1967} bounds their joint probability below by $(1-\alpha)^J$. Equivalently, applying the Gaussian correlation inequality successively to the intersection of these sets shows that their joint probability is at least $(1-\alpha)^J$. The rejection event under the reported rule is a subset of the union of the fixed-comparison rejection events. For a singular law, apply the inequality to the inverse images of the strips under a linear map from a standard Gaussian vector to $(W_1,\ldots,W_J)$. The bound averages over reports, while rejection conditional on a specific report has that report's probability as its denominator.

For sharpness, take $J$ orthonormal directions in Section~\ref{disp:gaussian}, with all fixed comparisons identifying the same ATT. With zero path noise, minimizing pre-treatment loss selects the largest $|u_j'Z|$. The absolute projections are independent and continuous, so the leading projection is unique almost surely. As path noise decreases to zero and the fit--error loading $\rho$ tends to one, the selected index equals the zero-noise winner eventually for almost every draw and the selected statistic converges in absolute value to $\max_j|u_j'Z|$. This maximum has no mass at $c_\alpha$. Its rejection probability is $1-(1-\alpha)^J$. Each design in the approximating sequence has positive path noise and $\rho<1$, so its outcome covariance is positive definite by the construction in Section~\ref{disp:gaussian}.
\end{proof}

\subsection{Candidate count and the winner's pre-trend test}
\label{size92:regional}
Table~\ref{tab:fixed-pool-coverage-details} gives coverage in the baseline fixed-pool designs. Here we vary fit--error covariance using independent regions and the three-period panel in Section~\ref{v63:simulation-embedding}. Within each region, $n=200$ observations of $(V,H)$ are normal with mean zero. The treated region has variance $r$ for $V$, controls have variance one, and all regions have unit variance for $H$. Within-region covariances of $V$ and $H$ are $\gamma$ in the treated region and $-\gamma$ in each control. The main design uses $r=4$ and $\gamma=.95$, with treated and control correlations $.475$ and $-.95$. The untreated paths satisfy parallel trends and the effect on the treated is zero. For the contrast using control $s$, covariance with its own pre-treatment gap is zero, while covariance with any rival gap is $\gamma/n$ through the shared treated outcomes.

We report the candidate with the smallest absolute pre-treatment gap if it passes a two-sided 5 percent pre-trend test, using its estimated variance and a standard normal cutoff. We draw the joint sample means and their same-sample covariance using the exact Gaussian and Bartlett--Wishart laws. Feasible inference reproduces both decisions on the conditional line with estimated direction and variance, rejecting when the pivot lies outside $[.025,.975]$. The known-covariance check uses the same realized screen with population direction and variance; sample means and covariance are independent in this design.

Figure~\ref{fig:fixed-pool-size}A reports rejection among screened winners. At $J=39$, usual and feasible rejection are 20.25 and 5.15 percent of 99,846 reports. Setting $\gamma=0$ gives 5.09 and 5.17 percent with the same denominator. The replication files give the complete regional grid, checks at $n=50$ and $800$, and rank-only results based on all 100,000 draws. For the known-variance conventional test in this model, as the independent pool expands the selected pre-change approaches the treated pre-change and the standardized error approaches a centered normal with variance $1+\gamma^2(r+1)/2$. This limiting calculation describes the conventional test; each design used to evaluate feasible inference fixes the candidate pool.
\subsection{Candidate diversity, fit--error dependence, and false rejection}
\label{size92:projection}
We draw independent Gaussian panel vectors from equations~\eqref{disp:panel-pre} and~\eqref{disp:panel-post} in Section~\ref{disp:gaussian}, with zero treatment effect and ten centered pre-treatment coordinates (eleven pre-treatment periods). Fit uses only pre-treatment outcomes. The main grid fixes path noise at $\delta=.08$ and loading at $\rho=.90$ and adds orthogonal candidate directions. The directions are fixed before each simulation. To isolate similarity at a fixed count, write $e_j$ here for coordinate direction $j$ in $\mathbb R^{10}$. Corollary~\ref{v57:normal-clusters} bounds amplification at zero path noise when directions are close. The nearly duplicate pool uses $u_1=e_1$ and $u_j=\sqrt{1-.05^2}\,e_1+.05e_j$ for $j=2,\ldots,5$. We also vary the loading and noise at five orthogonal directions and increase the sample size from 200 to 800. Every candidate identifies the same zero ATT, and the positive path noise and idiosyncratic post-treatment innovations keep the covariance nonsingular.

The code draws the scaled sample mean and the needed column of its unbiased same-sample covariance from their exact joint normal and Wishart laws. In sample size $n$, write $R$ for the vector of centered gaps, $U_s$ for the selected comparison's unit-variance error, $\beta_s=\operatorname{Cov}(R,U_s)$, and $C_R=\operatorname{Var}(R)$. Conditional on the sample-mean selection, the sample covariance retains its Wishart law by normal mean--covariance independence. For the selected contrast, its sample variance is $\widehat v=\chi^2_{n-1}/(n-1)$. The corresponding sample regression direction has the distribution
\[
 \widehat\beta_s=\beta_s+\frac{G_s}{\sqrt{(n-1)\widehat v}},
 \qquad G_s\sim N(0,C_R-\beta_s\beta_s'),\quad G_s\ \text{independent of }\widehat v .
\]
To see this, regress the centered gap observations on the centered scalar contrast conditional on the latter. The usual test uses $\widehat v$. The selective test uses $\widehat v$ and $\widehat\beta_s$ and recomputes all fit inequalities along the conditional line. The known-covariance test uses the population direction and variance. The two selective pivots integrate over every retained interval.

Figure~\ref{fig:fixed-pool-size}B averages over all 100,000 draws without a screen. With five candidates, changing orthogonal directions to nearly duplicate ones reduces usual rejection from 19.99 to 5.30 percent (feasible: 5.09 and 5.11 percent). Keeping the directions orthogonal and changing only the loading from $.90$ to $.60$ lowers usual rejection to 12.00 percent; increasing only path noise from $.08$ to $.24$ lowers it to 16.81 percent. The replication files report the Monte Carlo standard errors, known-covariance benchmarks, sample-size checks and zero-noise loading thresholds.
\begin{figure}[tbp]
\centering
\includegraphics[width=.96\linewidth]{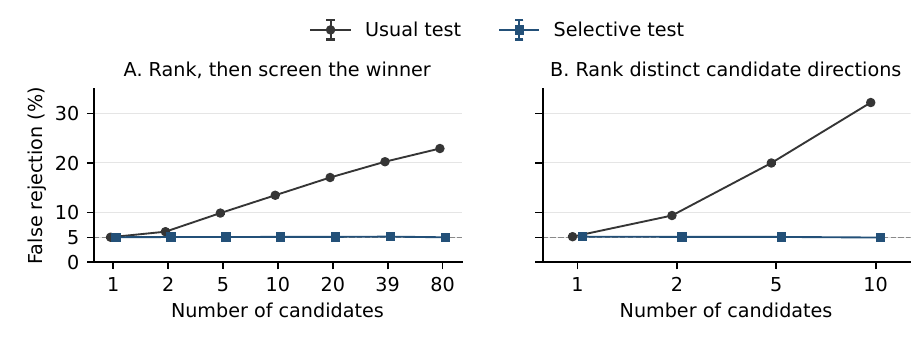}
\caption{False rejection as the number of controls increases}
\label{fig:fixed-pool-coverage}\label{fig:fixed-pool-size}\label{tab:fixed-pools-v15}\label{v63:full-fixed-pools}
\note{Both panels test a true zero effect at the nominal 5 percent level (dashed line), with 100,000 draws per fixed design. Panel A conditions on the winner passing its pre-trend test. Panel B uses orthogonal candidate directions and averages over all selections. Both selective tests estimate the joint covariance from the same sample. Bars show one Monte Carlo standard error. The replication files give the full regional and projection grids, including rejection counts and Monte Carlo standard errors.}
\end{figure}
\subsection{Fixed-comparison screening and screening after a search}
\label{v93:matched-oa}
Sections~\ref{size92:regional} and~\ref{size92:projection} vary the number and similarity of controls. Here we compare fixed and selected controls under the same pre-trend screen.
We first fix comparison 1 and report its conventional two-sided effect test only when it passes its pre-trend test, as in \citet[Corollary~1]{oa-deCDH2026}. We then apply that same reporting convention after choosing a control by fit. In a third procedure we report the best-fitting control without a screen. Both use the Gaussian panel in Section~\ref{disp:gaussian}, with $n=200$ observations, $m=10$ centered pre-treatment coordinates, path noise $\delta=.08$, and effect loading $\rho=.95$. We use one, two, three, five, or ten orthogonal directions. All candidates satisfy parallel trends and the true ATT is zero. Use $R_s$ for the root-$n$ centered pre-treatment gaps and $U_s$ for the standardized ATT estimation error for control $s$. Here $u_s$ is the direction associated with that control. The pre-test accepts a comparison if
\[
 W_s=\frac{\|R_s\|^2-(u_s'R_s)^2}{1+2\delta^2}
       +\frac{(u_s'R_s)^2}{2\delta^2}
       \le q_{m,.95},
\]
where $q_{m,.95}$ is the 95th percentile of the chi-squared distribution with $m$ degrees of freedom. The fixed protocol reports candidate 1 if it passes. The search protocol first minimizes $\|R_s\|^2$ and then reports that candidate if it passes the same test. The two protocols use the same simulated means and the same cutoff.

For a fixed $s$, $\operatorname{Cov}(R_s,U_s)=\rho(I_m-u_su_s')u_s=0$, with $I_m$ the $m$-dimensional identity matrix. Under joint normality, the fixed-control screen and ATT estimate are independent, and the conventional known-variance test has exact conditional size. For another candidate $j$, $\operatorname{Cov}(R_j,U_s)=\rho(I_m-u_ju_j')u_s$, which equals $\rho u_s$ for orthogonal directions. Thus the ranking responds to gaps correlated with the reported ATT error. The region in which $s$ wins and passes need not be convex. For instance, choose common shock vectors $a u_s+b u_j$ and $-a u_s+b u_j$ for $0<b<a$ and zero path innovations. Both give $s$ the best fit, and by scaling $a,b$ down, both satisfy its screen. At their midpoint, $j$ wins. The screen by itself is an ellipsoid in $R_s$.

The conventional ATT test uses the usual normal cutoff with estimated variance. We simulate its fixed and selected ATT variances and the selected contrast's covariance direction from the joint sample covariance of the Gaussian observations. Conditional on the selected contrast's centered sample vector, the fixed contrast decomposes into its projection onto that vector and an independent Gaussian vector in the remaining sample directions. This gives the exact joint Wishart draws needed by the two protocols. The selective test recomputes every fit comparison and the winner's screen along the line and integrates over the full accepted set. To measure what the ranking adds, we also compute a diagnostic of our own that supplies our selective pivot with only the acceptance ellipsoid of the displayed pre-test.

\begin{table}[!htbp]
\centering\setlength{\tabcolsep}{3pt}\small
\caption{Tests of a common zero effect under three reporting rules}\label{tab:matched-screen}
\begin{tabular}{@{}rrrrrrrr@{}}
\toprule
 & \multicolumn{2}{c}{Fix and screen} & \multicolumn{1}{c}{Rank only} & \multicolumn{4}{c}{Rank, then screen the winner} \\
\cmidrule(lr){2-3}\cmidrule(lr){4-4}\cmidrule(l){5-8}
Controls & Reports & Usual & Usual & Reports & Usual & Screen only & Full rule \\
\midrule
1  & 95,069 & 5.18 & 5.16 & 95,069 & 5.18 & 5.16 & 5.16 \\
2  & 94,969 & 5.18 & 9.88 & 96,464 & 9.74 & 9.67 & 5.14 \\
3  & 94,993 & 5.22 & 13.92 & 97,394 & 13.75 & 13.60 & 5.03 \\
5  & 95,092 & 5.01 & 21.16 & 98,359 & 20.87 & 20.70 & 5.06 \\
10 & 94,962 & 5.14 & 35.83 & 99,407 & 35.57 & 35.39 & 4.96 \\
\bottomrule
\end{tabular}
\note{Percent of reports that reject at the nominal 5 percent level. Ranking alone reports all 100,000 draws per row. The last two columns replay either the displayed pre-test or the complete procedure. The gap covariance used for screening is known. All effect tests use sample covariance. At five candidates, Monte Carlo standard errors are .071 points for fixed-and-screened, .129 for ranking alone, and .130, .129 and .070 for the three tests of the screened winner.}
\end{table}

\paragraph{Which samples does the search retain?}
With five candidates, comparison 1 passes its screen in 95,092 draws. It remains the winner in 19,529 of them. Under the usual test, 4,058 of these retained draws reject the true null (20.78 percent), compared with 705 of the 75,563 draws in which it passes but loses (0.93 percent). The search thus retains 85.20 percent of the fixed comparison's false rejections and only 20.54 percent of its reports. Replaying the complete rule reduces rejection to 1,005 of the same 19,529 draws (5.15 percent). The other four winners' full-rule rates range from 4.75 to 5.31 percent. With ten candidates, comparison 1 rejects in 35.94 percent of its 9,876 retained reports and in 1.57 percent of its 85,086 lost reports under the usual test. The full-rule rate for the retained reports is 5.08 percent.

Pooling all five winners' screened reports, the usual test alone rejects in 15,603 draws and the full-rule test alone rejects in 55. The paired difference is 15.81 percentage points (Monte Carlo standard error .12 point). We also shift every candidate's true effect by one or two of its population standard errors while keeping all pre-treatment gaps and their ranking fixed. The full-rule test rejects zero in 27.94 and 55.04 percent of the same 98,359 reports, respectively (Monte Carlo standard errors .14 and .16 point). The replication files record the paired decisions, counts for each winner and shifted-effect results.

\paragraph{Loading, screen covariance, and similar controls.}
At a lower loading of .90, ranking and screening give 19.51 percent rejection with the usual test and 5.09 percent with the full rule. At the headline loading of .95, nearly identical candidate directions give 5.33 and 5.06 percent. When we estimate the gap covariance from the same sample and use the asymptotic chi-squared cutoff, the usual test rejects in 5.17 percent of fixed reports and 21.07 percent of searched reports. The full-rule test rejects in 5.10 percent of searched reports. Normal sample means and covariances are independent in this design. We condition on the estimated screen matrix and replay that screen. Using the Hotelling cutoff gives 21.20 and 5.10 percent after search for the usual and full-rule tests. An independent check using 25,000 draws of the individual records, with the sample Wald screen, gives 4.98 percent for the fixed comparison and 21.02 and 4.94 percent for the searched comparison.

\paragraph{Why screening need not remove the ranking distortion.}
In the orthogonal design, rotate coordinates so that $u_s'Z=Z_s$, where
$Z_1,\ldots,Z_m$ are independent standard normals. As $\delta\downarrow0$
and $\rho\uparrow1$, candidate $s$ wins if $|Z_s|$ is largest among the
first $J$ coordinates, and its effect error tends to $Z_s$. Its screen
statistic tends to $\sum_{k\ne s}Z_k^2+\eta_s^2$, where
$\eta_s=u_s'(\zeta_0-\zeta_s)/\sqrt2$ is standard normal and independent of
$Z$. Thus the screen still has $m$ degrees of freedom. A fixed candidate's
screen is independent of its own limiting effect error. The distribution of
$|Z_s|$ among reports of winner $s$ instead weights its original density by
\[
  w_s(v)=\Pr\left(
  \sum_{k\ne s}Z_k^2+\eta_s^2\le q_{m,.95},\ 
  \max_{\substack{j\le J\\j\ne s}}|Z_j|\le v
  \right).
\]
For $J\ge2$ the weight is increasing on $0<v<\sqrt{q_{m,.95}}$
and constant thereafter. Weighting by $w_s$ strictly increases the
probability of every upper tail of $|Z_s|$. Consequently the
known-variance two-sided test overrejects after ranking and screening in
this limit. By continuity the same strict inequality holds for sufficiently
small positive $\delta$ and $\rho$ sufficiently close to one. The argument
uses the population-covariance screen and the known-variance effect test;
the simulations above evaluate the stated positive-noise designs and use
estimated covariance for effect inference.
\subsection{Gaussian reference experiments with archived policy covariances}
\label{v96:policy-reference}
To connect Figure~\MainRef{fig:matched-screen} with the applications, we use the archived joint covariance of the Buenos Aires policing and Washington minimum-wage outcome summaries. We fix each covariance and preserve the observed pre-treatment means. Under the null, we set the post-treatment mean of each primitive group to its pre-treatment mean, so all candidate DiD contrasts are zero. In 50,000 Gaussian draws per application, we rank candidates by centered squared fit and then apply a 5 percent omnibus pre-trend screen to the winner. This experiment treats the archived covariance as known. The imposed zero DiD contrasts permit nonzero pre-treatment gaps in the preserved means. The fixed comparator applies the same screen to the first candidate, all unprotected blocks in Buenos Aires and Alabama in Washington.

The conventional fixed-and-screened tests reject in 3.54 and 4.89 percent of eligible Buenos Aires and Washington draws; the rank-and-screen rates are in the main simulation section. In the observed data the nearest boundaries are 0.76 standard errors above the policing estimate and 0.40 below the wage-band estimate. Under alternative effects of $-.05$ thefts per block per recorded month and $-2$ percentage points, the adjusted test rejects in 34.96 and 53.45 percent of the corresponding reports. The replication files report denominators, interval widths, and the simulation code.
\subsection{Selecting the most significant comparison among those that pass}\label{cherry93:att-screen}
This experiment uses three periods and $J$ controls. Let $B_0,\ldots,B_J,F,G,g_1,\ldots,g_J,U_1,\ldots,U_J$ be independent standard normal innovations for each sampled unit. The treated region's outcomes are $(B_0-G/2,B_0+G/2,B_0+F)$. Control $j$ has outcomes $(B_j-(G-g_j)/2,B_j+(G-g_j)/2,B_j+F-U_j)$. These transformations are invertible, so the outcome covariance is positive definite. All potential outcome means equal zero. The pre-treatment gap in changes is $g_j$, and the difference-in-differences contrast subtracting the pre-treatment average is $U_j$. Each eligible comparison identifies the same zero effect on the treated.

A candidate is eligible if the absolute value of its root-sample-size scaled pre-treatment gap does not exceed $\Phi^{-1}(.975)$, where $\Phi$ is the standard normal distribution function. We report the candidate with the largest absolute effect statistic among those eligible, with a fixed tie order, and report no effect if none passes. The effect statistics use the sample standard errors for all candidates. Given the realized sample covariance, the Gaussian means retain their original covariance by normal mean--covariance independence. We use the full event that the winner passes and every rival either fails or has a smaller absolute statistic. The selective calculation integrates over all intervals of the conditional line satisfying those conditions. Its feasible version uses the estimated covariance for the direction and reference variance. The oracle version uses the true direction and variance while preserving the realized scales used by the reporting rule.

For $J=1,2,5,10$, the usual test rejects in 5.18, 9.55, 21.97, and 39.41 percent of reports, respectively, and the feasible selective test rejects in 5.16, 5.14, 5.00, and 5.05 percent. The experiment uses 100,000 replications with 200 independent units per replication in each design. With known standard errors for all decisions the exact conventional rejection probability given a report is
\[
 \frac{1-(1-.95\cdot.05)^J}{1-.05^J}.
\]
The numerator is the probability that at least one candidate passes its screen and would reject the effect null. The denominator is the probability that at least one candidate passes. This gives 21.60 percent at $J=5$ and 38.53 percent at $J=10$. For the realized estimated standard errors, the oracle conditions on the competing gaps and effects as well as the sample covariance. If the winning effect has absolute value $a$ on its population-standardized scale and must exceed the largest eligible rival statistic on that scale, $b$, the exact $p$-value is $(1-\Phi(a))/(1-\Phi(b))$. Here $b$ is zero if the winner is the sole eligible candidate.

\subsection{A distribution-free criterion for variance amplification}\label{amp:distribution-free-conditions}

Let $X$ contain all pre-treatment information used to choose a control, and let
$S=S(X)\in\{1,\ldots,J\}$. Write $U_s=\widehat\tau_s-\tau$,
$I_s=\mathbf1\{S=s\}$, and $p_s=P(S=s)$. Assume $\E U_s=0$ and
$\sigma_s^2=\E U_s^2\in(0,\infty)$ for every fixed candidate. The comparison
benchmark is $B=\sum_s p_s\sigma_s^2$, the error variance from an independent
control draw with the same selection probabilities. Define
$m_s(X)=\E(U_s\mid X)$, $v_s(X)=\Var(U_s\mid X)$, and
$q_s(X)=v_s(X)+m_s(X)^2=\E(U_s^2\mid X)$.

\begin{proposition}[Conditional moments and amplification]\label{amp:conditional-moments}
For any joint distribution satisfying these conditions,
\begin{align}
\Var(U_S)&=\E v_S(X)+\Var\{m_S(X)\},\label{amp:conditional-variance}\\
\Var(U_S)-B&=\sum_{s=1}^J\Cov\{I_s,q_s(X)\}
 -[\E m_S(X)]^2.\label{amp:conditional-increment}
\end{align}
Consequently, for any $r>0$,
\begin{equation}\label{amp:exact-threshold}
 \frac{\Var(U_S)}{B}>r
 \quad\Longleftrightarrow\quad
 \sum_s\Cov\{I_s,q_s(X)\}>(r-1)B+[\E m_S(X)]^2.
\end{equation}
If $m_s(X)=0$ almost surely for every $s$, then the reported estimator remains
unbiased and the criterion depends entirely on whether selection favors high
conditional error variances. In particular, zero covariance between $X$ and
every $U_s$ does not by itself remove this source of amplification.

Let $\nu_s=\E(U_s\mid S=s)$ for $p_s>0$. For the expected-variance ratio in main equation~\MainEqRef{v58:main-expected-ratio},
\begin{equation}\label{v58:nonnormal-expected-ratio}
 \mathcal R_E
 =1+\frac{\sum_s\Cov\{I_s,q_s(X)\}
             -\sum_{s:p_s>0}p_s\nu_s^2}{B}.
\end{equation}
Equivalently,
$B\mathcal R_E=\sum_{s:p_s>0} p_s[\E\{v_s(X)\mid S=s\}+\Var\{m_s(X)\mid S=s\}]$.
The total reported-error variance is $B\mathcal R_E+\Var_S(\nu_S)$.
\end{proposition}
\begin{proof}
Since $S$ is measurable with respect to $X$, conditioning on $X$ gives
$\E(U_S\mid X)=m_S(X)$ and $\Var(U_S\mid X)=v_S(X)$.
Equation~\eqref{amp:conditional-variance} is the law of total variance. Also,
\[
\E U_S^2=\sum_s\E\{I_s q_s(X)\}
 =\sum_s p_s\E q_s(X)+\sum_s\Cov\{I_s,q_s(X)\}
 =B+\sum_s\Cov\{I_s,q_s(X)\}.
\]
Subtracting $(\E U_S)^2=[\E m_S(X)]^2$ gives \eqref{amp:conditional-increment} and hence \eqref{amp:exact-threshold}. Subtracting $\sum_{s:p_s>0}p_s\nu_s^2$ from $\E U_S^2$ instead gives equation~\eqref{v58:nonnormal-expected-ratio}; the remaining statements are the law of total variance within and across selected labels. Conditional mean zero gives $\E U_S=0$ but places no restriction
on the dependence of $v_s(X)$ on the selection event.
\end{proof}

\begin{corollary}[Moment bounds for a finite candidate pool]\label{amp:moment-upper}
If $\sigma_s^2=\sigma^2$ for every $s$, then for any selection rule
\begin{equation}\label{amp:finite-moment-upper}
 \frac{\Var(U_S)}{\sigma^2}\le J.
\end{equation}
If, in addition, $q>2$ and $\E|U_s|^q\le K\sigma^q$ for all $s$, then
\begin{equation}\label{amp:qmoment-upper}
 \frac{\Var(U_S)}{\sigma^2}
 \le \min\{J,(JK)^{\frac{2}{q}}\}.
\end{equation}
\end{corollary}
\begin{proof}
Pointwise $U_S^2\le\max_s U_s^2\le\sum_s U_s^2$, so expectation gives \eqref{amp:finite-moment-upper}. For \eqref{amp:qmoment-upper}, the monotonicity of finite-dimensional norms
and Jensen's inequality for $x^{\frac{2}{q}}$ give
\[
\E\max_s U_s^2
 \le \E\left(\sum_s|U_s|^q\right)^{\frac{2}{q}}
 \le\left(\sum_s\E|U_s|^q\right)^{\frac{2}{q}}
 \le(JK)^{\frac{2}{q}}\sigma^2.
\]
Variance is no greater than the second moment, whether or not selection
introduces bias.
\end{proof}

\paragraph*{Two candidates with nonnormal fit fluctuations.}
The own-versus-other loading comparison in the main text also has a
finite-variance extension. Let $Z_1,Z_2$ be independent and identically
distributed, symmetric about zero, and have variance one. Let $e_1,e_2$ be
independent of these variables, with mean zero and variance one. For
$a^2+b^2\le1$, define
\[
S=\arg\min_{j=1,2}Z_j^2,\qquad
U_1=aZ_1+bZ_2+\sqrt{1-a^2-b^2}e_1,\qquad
U_2=bZ_1+aZ_2+\sqrt{1-a^2-b^2}e_2,
\]
with a fixed rule for ties. Then
\begin{equation}\label{amp:two-nonnormal}
 \E U_S=0,\qquad
 \Var(U_S)=1+\kappa_Z(b^2-a^2),\qquad
 \kappa_Z=\tfrac12\E|Z_1^2-Z_2^2|\in[0,1].
\end{equation}
Indeed $U_S=aZ_S+bZ_{3-S}+\sqrt{1-a^2-b^2}e_S$.
Individual sign reversals leave $S$ unchanged, so the means and
cross-product expectation vanish. Moreover,
$\E Z_S^2=\E\min(Z_1^2,Z_2^2)=1-\kappa_Z$ and
$\E Z_{3-S}^2=1+\kappa_Z$. Conditional on $(Z_1,Z_2)$, the selected residual
has mean zero and variance one. These observations prove
\eqref{amp:two-nonnormal}. The bound on $\kappa_Z$ follows from
$|Z_1^2-Z_2^2|\le Z_1^2+Z_2^2$. No fourth moment is required. Section~\ref{v60:dependent-close-fit-section} constructs panels with the same full covariance and different amplification. Table~\ref{v60:dependent-state-comparison-table} compares four such designs.
\subsection{Dependent close-fit states and the number of search opportunities}
\label{v60:dependent-close-fit-section}

The probability of finding at least one close-fitting control determines
amplification when close fit carries higher conditional uncertainty. The
following result allows the close-fit states to be dependent.

\begin{theorem}[Dependent close-fit states]\label{v60:dependent-close-fit}
Fix $J\ge2$, $\epsilon\in(0,1)$ and $\lambda\in(0,1)$. Let $X$ contain
all selection information, let $S=S(X)$, and let $A_1,\ldots,A_J$ be
$X$-measurable indicators with $\Pr(A_s=1)=\epsilon$ for every $s$.
Their dependence is unrestricted. Suppose
\[
 \E(U_s\mid X)=0,\qquad
 \Var(U_s\mid X)=1-\lambda+\lambda \frac{A_s}{\epsilon},
 \qquad A_S=\max_{s\le J}A_s.
\]
Write $\pi=\Pr(\bigcup_s\{A_s=1\})$. Every fixed error has mean zero
and variance one, and every positive-probability selected-label error has
mean zero. The expected-variance ratio satisfies
\begin{equation}\label{v60:dependent-close-fit-ratio}
 \mathcal R_E=\Var(U_S)
 =1+\lambda\left(\frac{\pi}{\epsilon}-1\right).
\end{equation}
Consequently, amplification occurs if and only if $\pi>\epsilon$, and
\begin{equation}\label{v60:dependent-close-fit-sharp-range}
 1\le\mathcal R_E\le
 1+\lambda\{\min(J,\frac{1}{\epsilon})-1\}.
\end{equation}
For every fixed $(J,\epsilon,\lambda)$, both bounds and every intermediate
value are attained by some joint distribution of the indicators and errors
satisfying these conditions. For $p_s=\Pr(S=s)>0$, the label-specific
variance is
\begin{equation}\label{v60:dependent-label-variance}
 \Var(U_s\mid S=s)
 =1-\lambda+\frac{\lambda}{\epsilon}\Pr(A_s=1\mid S=s).
\end{equation}
It equals \eqref{v60:dependent-close-fit-ratio} for every selected label
when the joint construction is exchangeable in candidate labels and the
selection rule preserves that symmetry almost surely.
For every $q>1$,
\begin{equation}\label{v60:dependent-q-threshold}
 \mathcal R_E>q
 \quad\Longleftrightarrow\quad
 \pi>\epsilon\{1+\frac{q-1}{\lambda}\}.
\end{equation}
By \eqref{v60:dependent-q-threshold}, an indicator law attaining a ratio above $q$ exists if and only if
$\min(J,\frac{1}{\epsilon})>1+\frac{q-1}{\lambda}$.
\end{theorem}
\begin{proof}
The conditional moments give $\E U_s=0$ and $\Var(U_s)=1$. Since $S$ is
$X$-measurable,
\[
 \E(U_S\mid X)=0,\qquad
 \E(U_S^2\mid X)=1-\lambda+\lambda \frac{A_S}{\epsilon}.
\]
Taking expectations and using $A_S=\max_s A_s$ proves the variance formula.
Also $\E(U_s\mid S=s)=0$, so the law of total variance makes the
unconditional variance equal the expected selected-label variance and the
benchmark denominator is one. Conditioning the same second-moment formula
on $S=s$ proves \eqref{v60:dependent-label-variance}.

The union probability satisfies $\epsilon\le\pi\le\min(J\epsilon,1)$.
The lower endpoint follows by setting all indicators equal to a single
Bernoulli variable. If $J\epsilon\le1$, use disjoint events of probability
$\epsilon$ to attain the upper endpoint. If $J\epsilon\ge1$, let $T\sim\operatorname{Uniform}[0,1)$ and define
$A_s=\mathbf1\{(T-\frac{s-1}{J})\bmod 1<\epsilon\}$ for $s=1,\ldots,J$. Each event has probability $\epsilon$, and their union covers $[0,1)$, so $\pi=1$.
Mixtures of the lower- and upper-endpoint indicator laws preserve their
marginals and yield every union probability between the endpoints. For any
such law, take $X=(A_1,\ldots,A_J)$, select the first index attaining the
maximum, and set
$U_s=\{1-\lambda+\lambda \frac{A_s}{\epsilon}\}^{\frac{1}{2}}e_s$ with independent
standard normal $e_s$ independent of $X$. This attains every stated ratio.
Alternatively, the panel construction below attains the same range.

Under candidate exchangeability and an equivariant almost-sure choice,
$\Pr(S=s)=\frac{1}{J}$ and
$\Pr(A_s=1,S=s)=\frac{\pi}{J}$ by symmetry. Hence
$\Pr(A_s=1\mid S=s)=\pi$, proving the last assertion. An independent
uniform permutation of candidate labels symmetrizes any of the indicator
laws used for attainability, without changing its marginal or union
probabilities.
\end{proof}

\begin{corollary}[Groups sharing close-fit states]\label{v60:close-fit-groups}
Under Theorem~\ref{v60:dependent-close-fit}, partition the candidates into
$K$ nonempty groups. Suppose the indicators are equal within each group
and the $K$ group indicators are independent Bernoulli variables with
probability $\epsilon$. Then, regardless of group sizes,
\begin{equation}\label{v60:close-fit-group-ratio}
 \pi=1-(1-\epsilon)^K,\qquad
 \mathcal R_E=1+\lambda
 \left\{\frac{1-(1-\epsilon)^K}{\epsilon}-1\right\}.
\end{equation}
\end{corollary}
\begin{proof}
At least one candidate has its close-fit state precisely when at least one
of the $K$ independent group indicators equals one. Substitute its
probability into \eqref{v60:dependent-close-fit-ratio}.
\end{proof}
Thus independent opportunities to find a close-fit state, rather than the
number of candidate labels, determine this expected-variance ratio.
Different group sizes can give different label-specific variances even
though \eqref{v60:close-fit-group-ratio} always gives their
selection-probability-weighted average.

\begin{proposition}[Identical marginal laws and covariance, different amplification]
\label{v60:dependent-close-fit-panel}
Fix $h\in(0,\frac{1}{4})$ and $\beta,\sigma_0,\sigma_c>0$. Let the indicator vector
$(A_1,\ldots,A_J)$ have any joint law with common marginal success
probability $0<\epsilon<1$. Independently of it, take independent Rademacher
signs $D_s,R_s$, independent standard normal variables $e_0,\ldots,e_J$,
and independent identically distributed continuous symmetric variables
$\xi_0,\ldots,\xi_J$ supported on $[-h,h]$ with positive variance. All
these primitive collections are mutually independent. Define
\begin{align}
 V_0&=\xi_0,& H_0&=\sigma_0e_0,\nonumber\\
 V_s&=(1-A_s)D_s+\xi_s,&
 H_s&=\frac{\beta}{\sqrt\epsilon}A_sR_s+\sigma_ce_s,\quad s=1,\ldots,J,
 \label{v60:dependent-panel-primitives}\\
 S&=\arg\min_s|V_0-V_s|,& U_s&=H_0-H_s.\nonumber
\end{align}
Let $B=\sigma_0^2+\sigma_c^2+\beta^2$ and $\lambda=\frac{\beta^2}{B}$.
The standardized errors $\frac{U_s}{\sqrt B}$ satisfy
Theorem~\ref{v60:dependent-close-fit}, with
$X=(V_0,\ldots,V_J)$. Varying the dependence of the indicators while
holding their marginal probabilities fixed leaves every regional marginal
law, every fixed-comparison error law, and the full means and covariance
of $(V_0,H_0,\ldots,V_J,H_J)$ unchanged. In particular,
$\Cov(V_j,U_s)=0$ for all $j,s$. The ratio nevertheless varies over the
entire sharp range in \eqref{v60:dependent-close-fit-sharp-range}.

There is a three-period DiD realization with positive-definite full outcome
covariance, perfect population fit and the same ATT for every candidate.
The full covariance in that panel is also unchanged as indicator
dependence varies; the regional outcome vectors can be dependent.

For $0<\alpha<1$, let $c_\alpha$ be the exact symmetric
fixed-comparison critical value, set
$\nu^2=\sigma_0^2+\sigma_c^2$, $a=\frac{\beta}{\sqrt\epsilon}$, and write
\[
 g_c(x)=\Phi\{\frac{c-x}{\nu}\}-\Phi\{\frac{-c-x}{\nu}\}.
\]
Then
\begin{equation}\label{v60:dependent-close-fit-coverage}
 \Pr(|U_S|\le c_\alpha)
 =1-\alpha-(\pi-\epsilon)
        \{g_{c_\alpha}(0)-g_{c_\alpha}(a)\}.
\end{equation}
It is strictly below $1-\alpha$ exactly when $\pi>\epsilon$.
For $p_s>0$, the corresponding conditional coverage replaces $\pi$ in
\eqref{v60:dependent-close-fit-coverage} by
$\pi_s=\Pr(A_s=1\mid S=s)$. In an exchangeable construction the
conditional and marginal coverage formulas coincide for every candidate.
\end{proposition}
\begin{proof}
In \eqref{v60:dependent-panel-primitives}, the disjoint supports identify $A_s$ from $V_s$: $A_s=1$ if and only if
$|V_s|\le h$ almost surely. Conditional on the pre-treatment information,
$H_0-H_s$ has mean zero and variance
$\nu^2+\beta^2 \frac{A_s}{\epsilon}$. If $A_s=1$, its fit distance is at most
$2h$; if $A_s=0$, the distance is at least $1-2h>2h$. Hence
$A_S=\max_s A_s$. Conditional on the indicators and $V_0$, the candidate
fit variables are independent and non-atomic, so the minimizer is unique
almost surely. After division by $\sqrt B$, the required conditional
variance is $1-\lambda+\lambda \frac{A_s}{\epsilon}$.

Each regional pair $(V_s,H_s)$ depends on the indicators through $A_s$
only, and its marginal law is therefore unchanged. Centered independent
signs give, for $s\ne t$,
\[
 \Cov(V_s,V_t)=\Cov(H_s,H_t)=\Cov(V_s,H_t)=0,
 \qquad \Cov(V_s,H_s)=0.
\]
The variables for region zero also have zero cross-covariances with all
candidate pairs. Within a candidate region,
\[
 \Var(V_s)=1-\epsilon+\Var(\xi_s),\qquad
 \Var(H_s)=\beta^2+\sigma_c^2;
\]
region zero has variances $\Var(\xi_0)$ and $\sigma_0^2$. All are positive
and independent of the indicators' dependence. These calculations establish
all covariance claims, including the unchanged fixed-error distribution
$(1-\epsilon)N(0,\nu^2)+\frac{\epsilon N(a,\nu^2)}{2}+
\frac{\epsilon N(-a,\nu^2)}{2}$.

For the panel, add independent mean-zero normal level shocks
$L_0,\ldots,L_J$ with positive variances, independent of all preceding
variables, and set
\[
 Y_{j,-1}=L_j-\frac{V_j}{2},\qquad Y_{j,0}=L_j+\frac{V_j}{2},\qquad
 Y_{j,1}=L_j+H_j+\tau\mathbf1\{j=0\}.
\]
Every untreated outcome mean is zero at every date. Centered squared fit
is $\frac{(V_0-V_s)^2}{2}$ and the DiD contrast based on the pre-treatment average
is $\tau+U_s$. The covariance of the vector of all $(L_j,V_j,H_j)$ is
block diagonal with strictly positive diagonal entries, whatever the
indicator dependence. The displayed outcome transformation is invertible,
so the complete outcome covariance is positive definite and unchanged.
The regional vectors need not be independent because the indicators can
be dependent.

Conditional on all pre-treatment information, the selected normal residual
$\sigma_0e_0-\sigma_ce_S$ is $N(0,\nu^2)$ and the selected sign $R_S$ is
Rademacher; they are independent of the selection information and of one
another. Thus selection changes the mixture weight of shifts $\pm a$
from $\epsilon$ to $\pi$. The fixed-error density is everywhere positive,
so $c_\alpha>0$ exists uniquely and satisfies
\[
 (1-\epsilon)g_{c_\alpha}(0)+\epsilon g_{c_\alpha}(a)=1-\alpha.
\]
Replacing $\epsilon$ by $\pi$ proves the coverage formula. For $c,a>0$,
$g_c(a)<g_c(0)$: differentiating with respect to $x>0$ gives
$\frac{\phi\bigl(\frac{c+x}{\nu}\bigr)-\phi\bigl(\frac{c-x}{\nu}\bigr)}{\nu}<0$.
Conditioning on $S=s$ changes the same mixture weight to $\pi_s$.
Exchangeability and the almost surely equivariant choice yield
$\pi_s=\pi$ as in the theorem.
\end{proof}

\begin{corollary}[Independent close fits]\label{amp:rare}
In Proposition~\ref{v60:dependent-close-fit-panel}, suppose the indicators are independent and $J\ge2$. Put $p_J=1-(1-\epsilon)^J$. The three-period panel has independent regions, positive-definite outcome covariance, perfect population fit, and a common ATT. Every fixed error has variance $B$ and zero covariance with all pre-treatment changes. The selected-error law is the same conditional on each winning label as across all labels. Its mean is zero and its variance divided by $B$ is
\begin{equation}\label{amp:rare-ratio}
R_J=1+\lambda\{p_J/\epsilon-1\}.
\end{equation}
At each label the exact symmetric fixed-comparison critical value $c_\alpha$ gives coverage
\begin{equation}\label{amp:rare-coverage}
1-\alpha-(p_J-\epsilon)\{g_{c_\alpha}(0)-g_{c_\alpha}(a)\}<1-\alpha .
\end{equation}
The jointly normal law with the same means and full covariance has variance ratio one and exact coverage conditional on each winner when it uses its own fixed-comparison critical value.
\end{corollary}
\begin{proof}
Independence gives $\pi=p_J>\epsilon$ in Proposition~\ref{v60:dependent-close-fit-panel} and makes the regions independent. Its construction gives finite moments of all orders. By candidate symmetry and uniqueness of the winner, $\Pr(U_S\in C,S=s)=\Pr(U_S\in C)/J$ for every measurable $C$. Substitution into the general result gives \eqref{amp:rare-ratio} and \eqref{amp:rare-coverage}. In the matched normal law the zero cross-covariance makes the entire vector of post-treatment changes independent of the vector used for ranking, leaving the $N(0,B)$ error law unchanged after each choice.
\end{proof}
For $r>1$, $R_J>r$ exactly when $p_J/\epsilon>1+(r-1)/\lambda$. Since $p_J/\epsilon=\sum_{k=0}^{J-1}(1-\epsilon)^k$, the increment $R_{J+1}-R_J=\lambda(1-\epsilon)^J$ is positive and decreases with $J$. At fixed $J\ge2$ the ratio rises with $\lambda$ and falls with $\epsilon$. As $J\to\infty$ its limit is $1+\lambda(1/\epsilon-1)$; as $\epsilon\downarrow0$ at fixed $J$ it tends to $1+\lambda(J-1)$. Then taking $\lambda\uparrow1$ gives the sharp supremum $J$ in Corollary~\ref{amp:moment-upper} within positive-definite DiD panels using pre-treatment fit alone.

At $\epsilon=.01$, $\beta^2=.9$, and $\sigma_0^2=\sigma_c^2=.05$, the exact 95\% fixed-comparison critical value is $0.6481386$. For $J=12$ and $20$, the variance ratios are $10.32536$ and $16.48838$, and coverages are $85.05713\%$ and $78.48602\%$, respectively.

Table~\ref{v60:dependent-state-comparison-table} compares state-dependence specifications with the same marginal laws and full covariance.

\begin{table}[tbp]\centering\small\setlength{\tabcolsep}{4pt}
\caption{Identical marginal outcome laws and covariance with different selection risks}
\label{v60:dependent-state-comparison-table}
\begin{tabular}{p{.40\linewidth}rrr}\toprule
Dependence of close-fit states & \shortstack{Union\\probability} & \shortstack{Expected\\variance ratio} & \shortstack{Coverage\\(\%)}\\\midrule
One common state & $.0100$ & $1.00$ & $95.00$\\
Three independent groups & $.0297$ & $2.77$ & $93.11$\\
Twelve independent states & $.1136$ & $10.33$ & $85.06$\\
Disjoint states & $.1200$ & $10.90$ & $84.44$\\\bottomrule
\end{tabular}
\note{Deterministic evaluations of
\eqref{v60:dependent-close-fit-ratio} and
\eqref{v60:dependent-close-fit-coverage}, with $J=12$, $\epsilon=.01$,
$\beta^2=.9$, $\sigma_0^2=\sigma_c^2=.05$, and
$\xi_j\sim\operatorname{Uniform}[-.1,.1]$. Thus $B=1$ and $\lambda=.9$.
The three groups each contain four candidates. Every row has the same
regional marginal outcome laws, full outcome covariance and fixed-error
distribution. Coverage uses the exact symmetric critical value of that
fixed-error distribution, whose coverage is 95\%. These are evaluations
of the theoretical functions, not Monte Carlo estimates.}
\end{table}

\FloatBarrier

\section{Information from pooling paths and finite-sample calibration}
\label{oa99:extensions}

\subsection{Information retained by pooling paths}\label{oa99:information}
Fix a reported contrast $T=c'Y$ in the Gaussian reference model, with mean $\theta$ and known variance $v>0$. Let $L$ be a detailed choice path and $H=g(L)$ the reported label. All paths in a given label must have the same $c$ as required in main-text condition~\MainEqRef{sd:compatibility}. Fix the projection residual $R=r$ and a report $H=h$ of positive conditional probability. Write $\pi_l=\Pr_\theta(L=l\mid r,h)$ and $m_l=\E_\theta(T\mid r,L=l)$ for each path of positive probability, and $m_h=\sum_l\pi_l m_l$. Let $\mathcal I_h$ and $\mathcal I_l$ denote the Fisher information for $\theta$ conditional on $(r,h)$ and $(r,l)$, respectively. Then
\begin{equation}\label{oa99:information-gap}
 \mathcal I_h-\sum_l\pi_l\mathcal I_l
   =\frac{1}{v^2}\sum_l\pi_l(m_l-m_h)^2 .
\end{equation}
To see this, write $A$ for the set of contrast values producing the label at the fixed residual. Its conditional density is proportional to $\ind\{t\in A\}\exp\{\theta t/v-t^2/(2v)\}$. Differentiating its log normalizing constant gives $\E_\theta(T\mid r,A)/v$; Gaussian tails justify differentiation on any bounded neighborhood of $\theta$. Thus the score is $(T-\E_\theta[T\mid r,A])/v$ and information is its conditional second moment. Applying the variance decomposition over the paths in $h$ gives \eqref{oa99:information-gap}, including for countably many paths. The gain is positive exactly when the path means differ with positive probability. The formula compares local information; interval lengths for a realized sample can move in either direction.

For the zero-effect four-cohort experiment in Section~\ref{sec:simulation-supp} (design in Section~\ref{supp:staggered-simulation}, coverage in Table~\ref{tab:sd-path-validation}), we evaluate \eqref{oa99:information-gap} on the same 200,000 draws for the first two cohorts' aggregate time-zero effect. For each drawn residual and report, we integrate the standard normal density, its first moment, and its second moment over each of the nine detailed path sets and over their union. Averaging over draws gives information of 0.5546 under the pooled report and 0.4109 after also conditioning on the detailed path. The difference is 25.91 percent of the pooled mean, with a paired Monte Carlo standard error of 0.08 percentage point. The relative loss compares these two averages. Every draw contributes, including draws for which numerical interval inversion would fail.

\subsection{Finite-sample inference with independent calibration contrasts}\label{oa99:calibration}
Main-text Theorem~\MainRef{thm:plugin} allows researchers to use the estimated joint covariance. In a Gaussian panel with independent calibration contrasts one can also put a finite-sample lower bound on coverage for every report of positive probability. We adapt the two calibration statistics of \citet{oa-HN2026weighted}, who study optimized control weights, to the deterministic ranking and screen that produced a single reported comparison.

For a fixed candidate, let $T$ be its effect estimate, $\theta=\E T$ its mean, and $S$ the vector of underlying coordinates used by the reporting rule $f$. Suppose $(T,S)$ is jointly normal. Write $v=\Var(T)>0$ and $a=\Cov(S,T)/v$. At observed values $(y,s)$ yielding report $j$, a trial $a$ gives the set $B(a)=\{t:f(s+a(t-y))=j\}$. For known $(v,a)$ the equal-tailed pivot uses the $N(\theta,v)$ distribution restricted to $B(a)$. At the true values $S-aT$ is independent of $T$. Conditioning first on this residual restricts $T$ to exactly the values that reproduce the report, so its conditional cumulative distribution function evaluated at $T$ is uniform, also for a union of disjoint intervals. Averaging over selected residuals gives the same uniform law given $j$.

Suppose $\nu$ independent centered Gaussian calibration rows have the covariance of $(T,S)$ and are independent of all analysis inputs to $f$ and $T$. Write their effect column as $b$ and score matrix as $U$. In $r$ score coordinates with nonsingular residual covariance, $0<r<\nu$, set
\[
 Q=b'b,\qquad \widehat a=U'b/Q,\qquad
 G=(U-b\widehat a')'(U-b\widehat a').
\]
With probability one $G$ is invertible. Two independent calibration statistics have the joint law
\begin{equation}\label{oa99:calibration-laws}
 \left(\frac Qv,\ \frac{\nu-r}{r}\,Q(\widehat a-a)'G^{-1}(\widehat a-a)\right)
 \ \sim\ \chi^2_\nu\otimes F_{r,\nu-r}.
\end{equation}
For the derivation write $U=ba'+D$. The residual matrix $D$ is independent of $b\sim N(0,vI_\nu)$ and has rank $r$ almost surely. Given $D$, let $P_D$ project onto its column space. The fraction $x=b'P_Db/Q$ is the ratio of the first of two independent chi-squares with $r$ and $\nu-r$ degrees of freedom to their sum. Thus $x$ is independent of $Q/v$ and $\frac{\nu-r}{r}x/(1-x)$ has the stated $F$ law. Substitution of $\widehat a-a=D'b/Q$ and $G=D'D-D'bb'D/Q$ gives $Q(\widehat a-a)'G^{-1}(\widehat a-a)=x/(1-x)$. This establishes \eqref{oa99:calibration-laws} conditionally on $D$ and hence after averaging over $D$.

Invert a rectangle in \eqref{oa99:calibration-laws} whose two reference probabilities multiply to at least $1-\delta$ to obtain $K$ for $(v,a)$, where $0<\delta<\alpha<1$. For every point in $K$, recompute $B(a)$ and invert its pivot with equal-tailed error $e=(\alpha-\delta)/(1-\delta)$. Use a measurable interval enclosing the union of these acceptance sets in $\theta$. At the true $(v,a)$, conditional on report $j$, the ordinary pivot covers with probability $1-e$. Independent calibration yields $\Pr((v,a)\in K\mid j)\ge1-\delta$ and independence of these two successes within the report. On joint success the enclosing interval contains $\theta$. Its report-conditional coverage is therefore at least $(1-\delta)(1-e)=1-\alpha$. In the regional panel, Gaussian individual outcome vectors allow calibration without setting anyone aside. For $n_g$ independent observations in region $g$, take mutually orthonormal zero-sum unit vectors $h_{gk}$ across its individuals. Each $n_g^{-1/2}\sum_i h_{gk,i}X_{gi}$ has the covariance of that region's mean; the resulting vectors are independent of all regional means and of each other under regional independence. Use the same transformations that define $T,S$. This construction permits up to $\min_g(n_g-1)$ rows when the report uses the regional means. Using the calibration rows themselves to screen the candidates requires a separate conditional calibration argument.

\paragraph{Finite endpoints.} For each $(v,a)\in K$, put $x=t-y$ and restrict the centered $N(0,v)$ measure to $\{x:y+x\in B(a)\}$. Suppose every such measure assigns mass at least $m>0$ to each of $[h,\infty)$ and $(-\infty,-h]$ for some $h>0$, and $v\le V<\infty$ on $K$. Then every accepted $\theta$ satisfies
\[
 |\theta-y|\le (V/h)\log\{2/(em)\}.
\]
Indeed the law under $\theta$ tilts this measure by $\exp\{(\theta-y)x/v\}$. For $\theta\ge y$ the integral over $x\le0$ is at most one, and the total integral is at least $m\exp\{(\theta-y)h/v\}$. The equal-tailed acceptance condition places at least $e/2$ below zero. Reflecting $x$ gives the other bound. A common neighborhood of $y$ inside every $B(a)$ suffices when $v$ is bounded above and away from zero. The mass condition also permits gaps within the selected set.

\paragraph{Cost for a rare report.} Consider two candidates with losses $(S-2.5)^2$ and $(S-.5)^2$ for $S\sim N(0,1)$. Candidate 1 wins when $S>1.5$, about 6.7 percent of samples. Its estimate is $T=\theta+.85S+\sqrt{1-.85^2}\,Z$ with $Z$ independent standard normal. To test $\theta=0$ at 5 percent, take $\delta=.01$ and use calibration reference probabilities each equal to $\sqrt{1-\delta}$, with a central chi-squared interval and upper $F$ cutoff. With 120 independent calibration rows and 30,000 draws conditional on candidate 1 winning, the known-covariance test rejects 4.95 percent under the null and 12.37 percent at $\theta=1.5$. The projected interval gives rejection frequencies bounded numerically by 1.53--1.57 percent and 2.70--2.78 percent, respectively. The two bounds bracket the numerical search over $K$ on the same draws; the largest Monte Carlo standard error for the reported frequencies is 0.19 percentage point. This example quantifies the power cost of protecting a rare report against uncertainty in the covariance. The replication files include the variance and slope search and results for 24 calibration rows.

\end{document}